\documentclass[journal,10pt,draftclsnofoot,]{IEEEtran}
\usepackage{xcolor}
\usepackage{amsmath, amsthm}
\usepackage{amsfonts}
\usepackage{mathtools}
\usepackage{algorithmic}
\usepackage{algorithm}
\DeclarePairedDelimiterX{\norm}[1]{\lVert}{\rVert}{#1}
\usepackage{graphicx}
\usepackage{subcaption}
\usepackage{dsfont}
\usepackage{bbm}
\usepackage{cite}
\usepackage{hyperref}
\usepackage{enumitem}
\usepackage{romannum}
\usepackage{verbatim}
\newtheorem{assumption}{Assumption}
\newtheorem{theorem}{Theorem}
\newtheorem{corollary}{Corollary}

\newtheorem{lemma}{Lemma}
\newtheorem{proposition}{Proposition}
\newtheorem{definition}{Definition}

\begin{document}
	\pagenumbering{arabic}
	
	\title{ Asymptotic Analysis of One-bit Quantized Box-Constrained Precoding in Large-Scale Multi-User Systems }
 \author{ Xiuxiu~Ma,~\IEEEmembership{Student Member,~IEEE}, Abla~Kammoun,~\IEEEmembership{Member,~IEEE}, Mohamed-Slim Alouini~\IEEEmembership{Fellow Member,~IEEE}
		and~Tareq~Y.~Al-Naffouri,~\IEEEmembership{Fellow Member,~IEEE}
		\thanks{X. Ma, A. Kammoun, M. Alouini and T. Y.  Al-Naffouri are with the Division of Computer, Electrical and Mathematical Science \& Engineering, King Abdullah University of Science and Technology (KAUST), Thuwal, KSA. E-mails: (\{xiuxiu.ma;  abla.kammoun; slim.alouini;tareq.alnaffouri\}@kaust.edu.sa)}}

	\onecolumn
	\maketitle
 \begin{abstract}
This paper addresses the design of  multi-antenna precoding strategies, considering hardware limitations such as low-resolution digital-to-analog converters (DACs), which necessitate the quantization of transmitted signals. The typical approach starts with optimizing a precoder, followed by a quantization step to meet hardware requirements. This study analyzes the performance of a quantization scheme applied to the box-constrained regularized zero-forcing (RZF) precoder in the asymptotic regime, where the number of antennas and users grows proportionally. The box constraint, initially designed to cope with low-dynamic range amplifiers, is used here to control quantization noise rather than for amplifier compatibility. A significant challenge in analyzing the quantized precoder is that the input to the quantization operation does not follow a Gaussian distribution, making traditional methods such as Bussgang's decomposition unsuitable. To overcome this, the paper extends the Gordon's inequality and introduces a novel Gaussian Min-Max Theorem to model the distribution of the channel-distorted precoded signal. The analysis derives the tight lower bound for the signal-to-distortion-plus-noise ratio (SDNR) and the bit error rate (BER), showing that optimal tuning of the amplitude constraint improves performance.
 \end{abstract}
	
	\section{Introduction}
	Multi-user massive multiple-input single-output (MU-mMISO) technology is a cornerstone of next-generation communication systems. By equipping base stations with large-scale antenna arrays, significant improvements in spectral and energy efficiency can be achieved \cite{6375940,5g,LuLSAZ14}. However, when conventional deployment methods are adopted, these advancements come with challenges, such as increased power consumption and higher deployment costs, as each antenna requires a dedicated radio-frequency chain (RFC).
    At the transmitter side, the primary power-consuming and costly components of a radio-frequency chain (RFC) are the power amplifier and the digital-to-analog converter (DAC). To address these challenges and reduce both costs and power consumption, several solutions have been proposed. One approach involves minimizing the number of RFCs through antenna selection, where only a subset of antennas is activated during each channel use \cite{as,as_arxiv}. Another strategy focuses on designing precoding techniques with peak-to-average-power ratio (PAPR) constraints, ensuring the transmit power at each antenna stays below a specified threshold \cite{con_pa,papr,glse}. Similarly, constant-envelope precoding techniques enforce the same amplitude for the precoded signal across all antennas \cite{MohammedL12,MohammedL13,PanM14}. Both approaches enable the use of low-dynamic-range power amplifiers, enhancing their efficiency by preventing operation in the non-linear regime. Furthermore, quantized precoding methods have been proposed, which restrict each precoded entry to a finite set of values. This not only supports the deployment of lower-resolution DACs but also  low-dynamic-range power amplifiers by inherently limiting the power to predefined levels, thereby achieving greater overall efficiency. Quantized precoding techniques can be broadly classified into two major categories. The first category, referred to as linear quantized precoding, involves applying quantization to the output of a linear precoder. This precoder can be either an existing one, typically zero-forcing or matched filter precoding \cite{mu,8628239,buss2,7472304,Li2016DownlinkAR}, or carefully optimized to account for the distortion caused by the finite resolution DAC \cite{mu,MezghaniGN09,KakkavasMMBN16,mee}. The second category, referred to as  quantized non-linear precoding, derives non-linear precoders  by solving a complex optimization problem. These precoders do not admit explicit formulation in general and are solved by  maximizing or minimizing a specific performance metric of interest \cite{low_res_1,8445995,9813427,Jedda2018QuantizedCE,10287256,JDCGS17,8811616,10097175,8462805,9007045,8331077}. In the following, we review the literature on analytical studies of quantized precoding and position our contributions within the context of this existing body of work.

\noindent{\bf Related works.} An extensive body of literature addresses the analysis and design of linear quantized precoders. Many of these precoders  are obtained by quantizing existing linear precoders such as zero-forcing or matched filter precoders. These precoders are generally suboptimal because their design does not account for the errors introduced by quantization. Performance analysis of such precoders has often relied on Bussgang's decomposition, which decomposes any quantized random Gaussian signal into a scaled signal component and an uncorrelated distortion \cite{mu,8628239,buss2,7472304,buss}. Although Bussgang's decomposition is typically combined with non rigorous assumptions, such as the independence of the distortion from the channel and the precoded signal—it has consistently yielded accurate predictions \cite{buss_qzf,buss2,9271818}. Recently, a rigorous justification for the applicability of Bussgang's decomposition to a broad class of linear quantized precoders was established in \cite{citeme}. Building on Bussgang's decomposition, several studies have developed enhanced linear quantized precoders that account for the impact of quantization noise \cite{Pe_bound}. These designs aim to optimize various performance metrics, including mean square error (MSE) \cite{MezghaniGN09,CandidoJMSN19}, energy efficiency \cite{mee}, and weighted sum rate \cite{KakkavasMMBN16}.
In addition to linear quantized precoders, non-linear quantized precoders have been proposed in recent studies \cite{JDCGS17,8462805,9007045,8331077,10287256}. These approaches involve solving complex optimization problems that account for finite bit resolution to meet specific design constraints \cite{7851074,low_res_1,8640023,Jedda2018QuantizedCE}. While these precoders generally offer superior performance in terms of symbol error rate (SER), their performance analysis remains largely unexplored due to the lack of explicit formulations. 

\noindent{\bf Contributions.} In this paper, we analyze the asymptotic performance of the quantized box-constrained precoding as the number of transmitter antennas and users grows large simultaneously. This precoder is derived by applying quantization to the output of a box-constrained precoder. In our prior work, we employed the Gaussian min-max theorem to evaluate the performance of the non-quantized box-constrained precoder \cite{tsp}. 
Unlike traditional linear quantized precoders, which typically rely on linear designs such as zero-forcing, the box-constrained precoder is inherently non-linear. Although quantization levels inherently limit the amplitude of the precoding entries, in the quantized precoding studied in this paper, we still reserve the box constranit, for it introduces an additional degree of freedom that can be optimized to improve system performance, as demonstrated in our analysis. By adjusting the box constraint, the proposed quantized precoder effectively compensates for quantization noise, marking it as the first non-linear precoder to be theoretically analyzed for this effect.
The non-linear nature of the box-constrained precoder prevents its distribution from converging to a Gaussian distribution, even in the large system limit, rendering Bussgang's decomposition inaccurate. To address this, we develop a new version of the Gaussian min-max theorem, specifically designed to analyze the impact of quantization noise. Although our primary focus is on one-bit quantization, we believe this approach can be extended to explore non-linear impairments in non-linear precoding scenarios where Bussgang's decomposition is no longer applicable. Our key contributions are outlined as follows.
\begin{enumerate}
\item{Precise asymptotic performance analysis:} we derive precise expressions for the tight lower bound of the signal-to-distortion-plus-noise ratio (SDNR) and the bit error rate (BER) for both the non-quantized and quantized box-constrained precoders. These expressions are presented in closed-form and depend solely on the channel statistics and system design parameters, including the quantization level and the box constraint. Additionally, we demonstrate that, unlike linear precoding, Bussgang's decomposition fails to provide accurate results in this context. We attribute this inaccuracy to the non-Gaussian nature of the precoder's distribution induced by the box constraint.
\item{Wider scope of application:} contrary to previous works that assume Gaussian signaling, our study focuses on binary shift keying (BPSK) signaling. This choice is primarily motivated by our consideration of one-bit precoding. However, our approach can be readily extended to higher-order symbol modulations while still assuming one-bit precoding. Although extending the analysis to higher-resolution precoding is considered feasible, we limit our scope to one-bit precoding to simplify the analysis and clearly convey the core ideas underlying our proof.
\item{A new Gaussian min-max theorem (GMT):} our work builds on a novel version of the Gaussian min-max theorem, specifically tailored to account for the effects of quantization noise. This new formulation draws inspiration from our recent study on thresholded $\ell_1$-norm precoding \cite{as_arxiv}, where we developed a  Gaussian min-max theorem to analyze the asymptotic behavior of thresholded precoding solutions. While both approaches share foundational similarities, the inclusion of the quantization function necessitates distinct adaptations to address its unique characteristics. We believe that these contributions are crucial and mark the emergence of new analytical techniques for studying the non-linear effects of precoding methods.
\end{enumerate}
The remainder of the paper is organized as follows. Section \ref{sec_model} presents the system model and formulates the problem. In Section \ref{sec:non_quantized}, we review the results regarding the performance of the box-constrained precoder. Section \ref{sec:quantized} then introduces the main results related to the performance of the quantized box-constrained precoder. Finally, in Section \ref{sec:numerical_results}, we provide a set of numerical results that validate the accuracy of our theoretical findings, even for typical system dimensions, and highlight the role of the box constraint in mitigating the impact of quantization noise. All proofs are deferred to the Appendices to allow us to focus on presenting and discussing the main results.

\noindent{\bf Notations.} The notation $\mathcal{N}(0,\sigma^2)$ refers to a Gaussian random variable with mean zero and variance $\sigma^2$. For a vector ${\bf v}$, we denote by $[{\bf v}]_i$ or $v_i$ its $i$-th element, and by $\|{\bf v}\|$, $\|{\bf v}\|_1$, $\|{\bf v}\|_{\infty}$ its $\ell_2$, $\ell_1$ and $\ell_{\infty}$ norms, respectively. 
For $a,b \in \mathbb{R}$, we denote by $\pi_{(a,b)}({\bf v})$ the set of indices of elements in ${\bf v}$ lying in the interval $(a,b)$. For $\mathcal{I}$ a set of integer indices, we denote by $\#\mathcal{I}$ the number of elements in $\mathcal{I}$.
For a random sequence $(X_n)_{n\in \mathbb{N}}$, we denote by $X_n \xrightarrow[]{P} X$ to indicate that $X_n$ converges to $X$ in probability. For $p \in \mathbb{N}$, and a set $\mathcal{S}_x \subset \mathbb{R}^p$, the notation $\mathcal{S}_x^c$ will denote the complement of $\mathcal{S}_x$ in $\mathbb{R}^p$.

\section{System model and problem formulation}
	\label{sec_model}
 We consider a downlink transmission between a base station equipped with $n$ transmit antennas and $m$ single antenna user terminals, over a frequency-flat channel. Each antenna is equipped with one-bit DAC while the ADC's of receivers are assumed to be of infinite resolution \footnote{The bit error rate results would remain unchanged even if we assumed a one-bit ADC at each receiver, as we are already using BPSK modulation. In this case, the detection process relies solely on determining the sign of the received signal, which aligns with the operation of a one-bit ADC. }.
 The downlink received signal writes as:
 $$
 {\bf y}={\bf Hx}_q +{\bf z}
 $$
 where ${\bf y}=[y_1,\cdots,y_m]^T$ denotes the received signals at all users, ${\bf x}_q$ is the transmitted precoded vector,  ${\bf H}=[{\bf h}_1,\cdots,{\bf h}_m]^{T}$ denotes the downlink channel with ${\bf h}_k$ being the channel between the base station and user $k$, $k=1,\cdots,m$, and ${\bf z}=[z_1,\cdots,z_m]^{T}$ is the noise vector whose elements independetly follow $\mathcal{N}(0,\sigma^2)$. The noise-free received signal, denoted as ${\bf e}_q = {\bf Hx}_q$, refers to the channel-distorted transmitted signal. For brevity, we refer to it as "distortion."

Let ${\bf s}=[s_1,\cdots,s_m]^{T}$ be the vector drawn from BPSK constellation containing information symbols transmitted to all users. 
  Typically, precoder design involves solving an optimization problem with various objective functions. In this paper, we assume perfect knowledge of the channel matrix ${\bf H}$ and focus on a precoder designed to penalize the transmit-receive difference. This is achieved by minimizing the mean squared error (MSE) between the information symbol vector $\sqrt{\rho}{\bf s}$ and the channel distorted received vector ${\bf Hx}$ where $\rho$ is a power control parameter. Given the use of one-bit DACs, the entries of the precoded vector ${\bf x}$ are restricted to the set $\mathcal{L}=\{\pm L\}$ where $L$ is the quantization level. 
More specifically, we consider a precoder designed to minimize:
	\begin{align}
		\hat{\bf x}=\arg\min_{{\bf x}\in \mathcal{S}_x} & \ \ \mathcal{P}({\bf x};\rho) \label{eq:nonlse0}
	\end{align}
 where $\mathcal{S}_x$ is a certain set to which entries of the precoded signal belongs and 
 $$
 \mathcal{P}({\bf x};\rho):= \frac{\|{\bf Hx}-\sqrt{\rho}{\bf s}\|^2}{n}+\lambda\frac{\|{\bf x}\|^2}{n}
 $$
	with $\lambda>0$ being a positive regularization parameter controlling  the trade-off between minimizing the difference and the regularization term $\frac{1}{n}\|{\bf x}\|^2$.
When $\mathcal{S}_x$ is discrete, solving the problem in \eqref{eq:nonlse0} incurs a prohibitively high computational cost, particularly for systems with a large number of antennas and users, as is typical in massive MISO systems.
To address this issue, we first compute a non-quantized solution by relaxing the set  $\mathcal{L}^n$ to  $\mathcal{S}_x=\{{\bf x}\in \mathbb{R}^n \  | \  \|{\bf x}\|_{\infty}\leq A\}$  followed by the application of one-bit quantization to this solution. Specifically, we obtain the precoded vector through the following steps:
		\begin{align}
			&\hat{{\bf x}}=\arg\min_{{\bf x}\in\mathcal{S}_{\bf x}}\mathcal{P}({\bf x}), \label{rls_Q}\\
			 &\text{Compute } \mathcal{Z}(\hat{\bf x})=L{\rm sgn}(\hat{\bf x}),\nonumber\\
          &  {\bf x}_q=\mathcal{Z}(\hat{\bf x}),\nonumber
		\end{align}
    where $\mathcal{P}({\bf x}):=\mathcal{P}({\bf x};1)$.
 In the context of a quantized precoder, setting $\rho=1$ simplifies the problem because the power is managed through the quantization level $L$. By setting $\rho=1$, we acknowledge that $\rho$ does not provide an additional degree of freedom. All possible solutions can be fully spanned by adjusting $A$ and $\lambda$. 
We assume that the users scale the received signal by a factor ${\kappa}$ \footnote{We assume here that the same factor is used across all users since all users experience the same channel. The value of $\kappa$ will be specified later. }.
The received signal at user $k$ can be decomposed as:
$$
\kappa y_k=s_k+\kappa {\bf h}_k^{T}{\bf x}_q-s_k+\kappa z_k.
$$
In this work, we propose to analyze the following metrics. 

\noindent{\bf Signal-to-distortion-plus-noise ratio.}
We define the signal-to-distortion-plus-noise ratio (SDNR) at user $k$ as:
$$
{\rm SDNR}_{k,q}=\frac{1}{\mathbb{E}[|\kappa{\bf h}_k^{T}{\bf x}_q-s_k|^2]+\kappa^2\sigma^2}
$$
and define the average SDNR across users as:
$$
{\rm SDNR}_{{\rm avg},q}=\frac{1}{m}\sum_{k=1}^m \frac{1}{{\mathbb{E}[|\kappa{\bf h}_k^{T}{\bf x}_q-s_k|^2]+\kappa^2\sigma^2}}.
$$
Using Jensen's inequality, a lower bound of the average SDNR is given by:
$$
{\rm SDNR}_{{\rm avg},q}^{l}:=\frac{1}{\frac{1}{m}\sum_{k=1}^m \mathbb{E}[|\kappa {\bf h}_k^{T}{\bf x}_q-s_k|^2]+\kappa^2\sigma^2}.
$$
Our analysis centers on examining this lower bound, which we conjecture to be tight based on the arguments presented in \cite{as_arxiv}.

\noindent{\bf Bit error rate.} The bit error rate (BER) is defined as:
$$
{\rm BER}_q=\frac{1}{m}\sum_{k=1}^m {\bf 1}_{\{{\rm sign} (\kappa y_k)\neq s_k\}}.
$$
To evaluate the impact of the quantization on the performance, we propose to compare the quantized box-constrained precoder with the non-quantized box-constrained  precoder given by:
\begin{equation*}
\hat{\bf x}_{\rm box}=\arg\min_{ {\bf x}\in\mathcal{S}_x} \mathcal{P}({\bf x};\rho) 
\end{equation*}
 For clarity, we will refer to the former as the "quantized precoder" and the latter as the "box precoder" when the context allows.

Similarly to the quantized precoder, we assume that the users scale the received signal by a factor $\varsigma$. Hence, the received signal at user $k$ can be decomposed as:
$$
\varsigma y_k=s_k +\varsigma {\bf h}_k^{T}\hat{\bf x}_{\rm box}-s_k+\varsigma z_k.
$$
We characterize the performance of the box precoder using the following metrics:

\noindent{\bf Signal-to-distortion-plus-noise ratio. } We define the signal-to-distortion-plus-noise ratio (SDNR) at user $k$ as:
$$
{\rm SDNR}_{k,{\rm box}}=\frac{1}{\mathbb{E}[|\varsigma{\bf h}_k^{T}{\bf x}_{\rm box}-s_k|^2]+\varsigma^2\sigma^2}
$$
and define the average SDNR across users as:
$$
{\rm SDNR}_{{\rm avg},{\rm box}}=\frac{1}{m}\sum_{k=1}^m \frac{1}{{\mathbb{E}[|\varsigma{\bf h}_k^{T}{\bf x}_{\rm box}-s_k|^2]+\varsigma^2\sigma^2}}.
$$
Similarly, we can obtain the following lower bound of the average SDNR:
$$
{\rm SDNR}_{{\rm avg},{\rm box}}^{l}= \frac{1}{\frac{1}{m}\sum_{k=1}^m \mathbb{E}[|\varsigma {\bf h}_k^{T}\hat{\bf x}_{\rm box}-s_k|^2]+\varsigma^2\sigma^2}.
$$
\noindent{\bf Bit error rate.} We define the bit error rate (BER) of the box precoder as:
$$
{\rm BER}_{\rm box}=\frac{1}{m}\sum_{k=1}^m \boldsymbol{1}_{\{{\rm sign}(\varsigma y_k)\neq s_k\}}.
$$
 A natural approach to constructing a quantized precoder from a non-quantized one involves applying quantization directly to the precoder. This methodology has been essentially used to develop quantized linear precoding, where quantization is applied to the outputs of classical linear precoding techniques such as matched filter, zero-forcing, or regularized zero-forcing. The performance of quantized linear precoding has traditionally been analyzed heuristically using the Bussgang decomposition \cite{citeme, mu}. The validity of this heuristic was recently established in \cite{citeme} for linear-quantized precoding. However, in this work, we impose an additional constraint by requiring the precoder vector to belong to the set $\mathcal{S}_x:=\{{\bf x}\in \mathbb{R}^n \ \  \|{\bf x}\|_{\infty}\leq A\}$. This additional constraint results in applying quantization to a non-linear precoder, rendering the approach used in \cite{citeme} inapplicable.
The constraint of bounding the precoder vector's infinity norm is typically motivated by the limited dynamic range of power amplifiers, as previously considered in our work in \cite{tsp}. In this context, however, this motivation does not stand out since the quantization process inherently limits the amplitude of the precoder entries. Nevertheless, as we will demonstrate, imposing a limit $A$ on the amplitude of the non-quantized precoder introduces an additional degree of freedom to finer-control the subsequent quantization noise, and thus enhancing performance.

\section{Performance analysis of the non-quantized box-constrained precoder}
\label{sec:non_quantized}
The performance analysis of the box-precoder was recently conducted in our works \cite{tsp} and \cite{as_arxiv}, where we derived closed-form asymptotic approximations for the bit error rate (BER) and the average signal-to-distortion-plus-noise ratio (SDNR) lower bound. Our analysis leveraged channel statistics in the asymptotic regime, where the number of antennas and users grow at the same rate. Specifically, we considered the following assumptions:

\begin{assumption}[Growth rate regime]
    The number of antennas n and the number of users m grow to infinity at a fixed ratio $\delta=\frac{m}{n}$. \label{ass:growth_rate}
\end{assumption}

\begin{assumption}[Channel model]
The channel matrix ${\bf H}$ has independent and identically distributed Gaussian entries with
zero mean and variance equal to $\frac{1}{n}$. \label{ass:channel_model}
\end{assumption}
In our works \cite{tsp} and \cite{as_arxiv}, we leverage the Convex Gaussian Min-Max Theorem (CGMT) to analyze the box precoder. We show that the asymptotic behavior of the precoder solution and the distortion vector $\hat{\bf e}_{\rm box}={\bf H}\hat{\bf x}_{\rm box}$ is governed by the saddle point of a scalar deterministic max-min problem: the solutions to this problem directly parameterize the distributions of both the precoder and the distortion vector.  Below, we review the key findings from \cite{as_arxiv} pertaining to these aspects.

\subsection{Asymptotic distribution of the non-quantized precoder and its distortion vector}
As seen in our work \cite{as_arxiv}, the performance of the box-precoder is characterized by the solutions to a limiting deterministic max-min scalar optimization problem. Below, we review our result in \cite{as_arxiv} dealing with existence and uniqueness of the solutions to the max-min scalar optimization problem. 
\begin{proposition}
For $\rho>0$ and $A>0$ consider the following max-min optimization problem:
\begin{equation}
\max_{\beta\geq 0}\min_{\tau\geq 0} \varphi(\tau,\beta;\rho,A) \label{eq:scalar_opt}
\end{equation}
where 
$$
\varphi(\tau,\beta;\rho,A):=\frac{\tau\beta \delta}{2}+\frac{\rho\beta}{2\tau}-\frac{\beta^2}{4}+\mathbb{E}_{H}\Big[\min_{|x|\leq A} \frac{\beta}{2\tau}x^2+\lambda x^2-\beta Hx\Big].
$$

The above optimization problem admits a unique solutions  $\tau^\star(\rho,A)$ and $\beta^\star(\rho,A)$ if and only if $\lambda>0$ or $\lambda=0$ and $\delta\geq 1$. 

For ${H}$ a standard Gaussian random variable, we define $X(H)$ the following random variable:
\begin{equation}
X(H):=\left\{
\begin{array}{ll}
-A & \text{if }  H\leq -A\alpha^\star\\
\frac{H}{\alpha^\star} &\text{if } -A\alpha^\star\leq H\leq A\alpha^\star\\
A & \text{if } H\geq A\alpha^\star
\end{array}
\right.\label{eq:XH}
\end{equation}
where $\alpha^\star=\frac{1}{\tau^\star(\rho,A)}+\frac{2\lambda}{\beta^\star(\rho,A)}$. Then $(\tau^\star(\rho,A),\beta^\star(\rho,A))$ are the unique solutions to the following system of equations:
\begin{equation*}
\left\{
\begin{array}{ll}
\tau^2\delta&=\rho+\mathbb{E}[|X(H)|^2]\\
\beta&=2\tau\delta -2\mathbb{E}[H X(H)]
\end{array}
\right. .
\end{equation*}
Furthermore,  at the saddle point $(\tau^\star(\rho,A),\beta^\star(\rho,A))$, $\varphi(\tau^\star(\rho,A),\beta^\star(\rho,A))$ simplifies to
$$
\varphi(\tau^\star(\rho,A),\beta^\star(\rho,A))=\frac{(\beta^\star(\rho,A))^2}{4}+\lambda \mathbb{E}_{H}[|X(H)|^2].
$$\begin{proof}See Appendix B in \cite{as_arxiv}.\end{proof}
\label{prop:max_min}
\end{proposition}
The parameters of the above scalar deterministic optimization problem will be used to characterize the performance of the box-precoder and the quantized precoder. For the quantized precoder the parameter $\rho$ is set to $1$.  

\noindent{\bf About the case $\delta=1$ and $\lambda=0$.} It is important to note that the scalar max-min optimization problem admits a unique saddle point even when $\delta = 1$ and $\lambda = 0$. In the special case of $A = \infty$, $\delta=1$ and $\lambda=0$ which corresponds to the zero-forcing precoding, the maximum in $\beta$ is attained at $\beta = 0$, and the uniqueness of $\tau$ cannot be guaranteed. However, when $A$ is finite, the uniqueness of the solution can be understood as a result of $A$ acting as a regularization parameter, helping to stabilize the solution.

Building on the result from Proposition \ref{prop:max_min}, we will now adopt the following assumption for the remainder of the analysis:
\begin{assumption}
We assume either $\lambda=0$ and $\delta\geq 1$ or $\lambda>0$.
\label{ass:cond_lambda}
\end{assumption}
In the following theorems, we present our results on the optimal cost of the optimization problem \eqref{eq:nonlse0}, as well as the asymptotic behavior of the precoding vector $\hat{\bf x}_{\rm box}$ and the distortion vector $\hat{\bf e}_{\rm box}={\bf H}\hat{\bf x}_{\rm box}$.  Specifically, we demonstrate that the distributions of both the precoding vector and the distortion vector converge to deterministic distributions. These deterministic distributions are parameterized by the solutions to the associated max-min scalar optimization problem. Additionally, we measure the deviations from these asymptotic distributions using the Wasserstein distance. Recall that for $d\in \mathbb{N}$ the
Wassertein $r$-distance between two measures $\mu$ and $\eta$ supported in $\mathbb{R}^{d}$ is defined as:
$$
\mathcal{W}_r(\mu,\eta):=(\inf_{\gamma \in \mathcal{C}(\mu,\eta)} \int \|{x}-y\|^r\gamma(dx,dy))^{\frac{1}{r}}
$$
where $\mathcal{C}(\mu,\eta)$ denotes the set of all couplings  of $\mu$ and $\eta$.

The scalar max-min optimization problem is a deterministic optimization problem whose optimal cost is asymptotically equivalent to $\mathcal{P}(\hat{\bf x}_{\rm box}, \rho)$. More formally, the following result holds true: 
\begin{theorem}[Convergence of the optimal cost] Under Assumption \ref{ass:growth_rate}, \ref{ass:channel_model} and \ref{ass:cond_lambda}, there exists constants $C$ and $c$, $\gamma$ such that the optimal cost $\min_{\|{\bf x}\|_{\infty}\leq A} \mathcal{P}({\bf x};\rho)$ satisfies for all $\epsilon$ sufficiently small:
$$
\mathbb{P}\Big[|\min_{\|{\bf x}\|_{\infty}\leq A} \mathcal{P}({\bf x},\rho)-\varphi(\tau^\star(\rho,A),\beta^\star(\rho,A))|\geq \gamma\epsilon\Big]\leq \frac{C}{\epsilon}\exp(-cn\epsilon^2)
$$
where $(\tau^\star(\rho,A),\beta^\star(\rho,A))$ is the unique saddle point to the scalar optimization problem in \eqref{eq:scalar_opt}. 
\begin{proof}See Appendix \ref{BA}.\end{proof}
\end{theorem}
For ${\bf x}\in \mathbb{R}^n$, we define the empirical measure of ${\bf x}$ as:
$$
\hat{\mu}({\bf x})=\frac{1}{n}\sum_{i=1}^n{\delta}_{x_i}.
$$

Let $\nu^\star$ be the distribution of the random variable $X(H)$. In the sequel, we present our results regarding the box precoder. Denote by $\tau^\star$, and $\beta^\star$ the unique saddle point to \eqref{eq:scalar_opt} for a given control power parameter $\rho$. 

The following theorem establishes convergence rate of the Wassertein distance between $\nu^\star$ and $\hat{\mu}(\hat{\bf x}_{\rm box})$.
\begin{theorem}[Convergence of the empirical measure of the precoder vector] Under Assumption \ref{ass:growth_rate}, \ref{ass:channel_model} and \ref{ass:cond_lambda}, 
for all $r\in [1,\infty)$, there exists constants $C$ and $c$ such that for all $\epsilon$ sufficiently small:
$$
\mathbb{P}\Big[(\mathcal{W}_r(\hat{\mu}(\hat{\bf x}_{\rm box}),\nu^\star))^r\geq \epsilon\Big]\leq \frac{C}{\epsilon^2}\exp(-cn\epsilon^4).
$$\begin{proof}See Theorem 3 in \cite{as_arxiv}.\end{proof}
\label{th:precoder_vector}
\end{theorem}
For ${\bf e}, {\bf s}\in \mathbb{R}^m$, define the joint empirical distribution as follows:
$$
\hat{\mu}({\bf e},{\bf s})=\frac{1}{m}\sum_{i=1}^m \delta_{e_i,s_i}.
$$
For ${G}$ standard Gaussian random variable, and $S$ a random variable taking $+1$ and $-1$ with equal probability, define the variable ${E}(G,S)$ as:
\begin{equation}
E(G,S)=\beta^\star\frac{\sqrt{(\tau^\star)^2\delta-\rho}{G}}{2\tau^\star\delta}+\sqrt{\rho}S(1-\frac{\beta^\star}{2\tau^\star\delta}). \label{eq:distortion_as}
\end{equation}
Denote by $\hat{\mu}(\hat{\bf e}_{\rm box},{\bf s})$ the joint empirical distribution of the distortion vector and the symbol vector. Let $\mu^\star_{ES}$ be the joint distribution of the couple $(E(G,S),S)$. Then, we prove that the joint empirical distribution $\hat{\mu}(\hat{\bf e}_{\rm box},{\bf s})$ can be approximated in the asymptotic regime by $\mu_{ES}^\star$ in the following sense:
\begin{theorem}[Convergence of the joint empirical measure of the distortion vector and the symbol vector]
Under Assumption \ref{ass:growth_rate},\ref{ass:channel_model} and \ref{ass:cond_lambda}, there exists positive constants $C$ and $c$ such that for all $\epsilon$ sufficiently small:
$$
\mathbb{P}\Big[\mathcal{W}_2(\hat{\mu}(\hat{\bf e}_{\rm box},{\bf s}),\mu_{ES}^\star)\geq \epsilon\Big]\leq \frac{C}{\epsilon^2}\exp(-cn\epsilon^4).
$$
\begin{proof}See Theorem 4 in \cite{as_arxiv}.\end{proof}
\label{th:disortion_vector}
\end{theorem}
These results allow us to define the following asymptotic characterization of the joint distribution of the received signal and the transmitted symbol for the non-quantized case.
\begin{definition}[Asymptotic distributional characterization of the received signal for the box precoder]
Let ${S}$ be a random variable taking $+1$ and $-1$ with equal probabilities. For $G$ a standard Gaussian random variable define $E(G,S)$ as in \eqref{eq:distortion_as}. Let $Z$ be a Gaussian random variable independent of $S$ and ${G}$ with mean zero and variance $\sigma^2$. Define \begin{align}Y=E(G,S)+Z\label{old_y}\end{align} and let $\mu_{YS}^\star$ be the joint distribution of $(Y,S)$. 
\end{definition}
Our result establishes asymptotic convergence of the joint empirical distribution of $\frac{1}{m}\sum_{i=1}^m\delta_{[{\bf H}\hat{\bf x}_{\rm box}+{\bf z}]_i,s_i}$ and $\mu_{YS}^\star$. In other words, the performance of the received signal by each user will be equivalent to:
\begin{equation}
Y=E(G,S)+Z. \label{eq:distr_charac}
\end{equation}
Given the expression of $E(G,S)$ it appears natural to scale the received signal by $\varsigma=\frac{1}{\sqrt{\rho}(1-\frac{\beta^\star}{2\tau^\star\delta})}$, resuling into:
\begin{equation}
\varsigma Y= \beta^\star\frac{\sqrt{(\tau^\star)^2\delta-\rho}}{2\tau^\star\delta\sqrt{\rho}(1-\frac{\beta^\star}{2\tau^\star\delta})}G+S+\frac{1}{\sqrt{\rho}(1-\frac{\beta^\star}{2\tau^\star\delta})}Z. \label{eq:distributional_charac}
\end{equation}
\subsection{Review of the performance analysis of the non-quantized precoder}
Based on the asymptotic characterizations of the precoder vector, the joint distribution of the received signal and the transmitted symbol, we can derive the following asymptotic metrics for the box precoder:
\begin{theorem}[Convergence of the precoder power \cite{as_arxiv}] Under Assumption \ref{ass:growth_rate}, \ref{ass:channel_model} and \ref{ass:cond_lambda}, the per-antenna power $P_b(\hat{\bf x}_{\rm box})$ converges to:
$$
P_b(\hat{\bf x}_{\rm box})\xrightarrow[]{P} \overline{P}_b:=\delta(\tau^\star)^2-\rho.
$$
\label{th:power_precoder}
\end{theorem}
\begin{proof}See Corollary 1 in \cite{as_arxiv}.\end{proof}

Based on the asymptotic characterization in \eqref{eq:distributional_charac}, the lower bound of the SDNR satisfies the below convergence. 
\begin{theorem}[Convergence of the SDNR lower bound]
Under Assumption \ref{ass:growth_rate}, \ref{ass:channel_model} and \ref{ass:cond_lambda}, the lower bound of the average SDNR of the box precoder converges to:
$$
{\rm SDNR}_{\rm avg,box}^{l}\xrightarrow[]{P} \overline{\rm SDNR}_{\rm avg,box}^{l}:=\frac{\rho(1-\frac{\beta^\star}{2\tau^\star\delta})^2}{(\beta^\star)^2\frac{(\tau^\star)^2\delta-\rho}{4(\tau^\star)^2\delta^2}+\sigma^2}.
$$
\end{theorem}
Consider a detector estimating the vector ${\bf s}$ as $\hat{\bf s}={\rm sign}(\varsigma {\bf y})$. Based on the asymptotic characterization in \eqref{eq:distributional_charac}, the BER of the box precoder satisfies the following convergence.
\begin{theorem}[Convergence of the BER]Under Assumption \ref{ass:growth_rate}, \ref{ass:channel_model} and \ref{ass:cond_lambda}, the BER of the box precoder converges to:
$$
{\rm BER}_{\rm box} \xrightarrow[]{P}\overline{\rm BER}_{\rm box}:= Q(\frac{\sqrt{\rho}(1-\frac{\beta^\star}{2\tau^\star\delta})}{\sqrt{(\beta^\star)^2\frac{(\tau^\star)^2\delta-\rho}{4(\tau^\star)^2\delta^2}+\sigma^2}}).
$$
\end{theorem}
Before proceeding to the performance analysis of the quantized precoder, an important remark is necessary. Consider the scenario where the power control parameter $\rho$ is tuned such that $\frac{\overline{P}_b}{\sigma^2} = {\rm SNR}_{\rm tx}$, where ${\rm SNR}_{\rm tx}$ represents a target transmit SNR value within the range $(0, \frac{P}{\sigma^2})$ where $P=A^2$. \footnote{As shown later in Figure \ref{fig1}, the parameter $\rho$ can be adjusted to achieve any value of $P_b$ within the range $(0,  {P})$.} Using the result from Theorem \ref{th:power_precoder}, this tuning can be performed by setting $\rho$ such that $\tau^\star\delta = \sqrt{({\rm SNR}_{\rm tx} \sigma^2 + \rho)\delta}$. With this relationship, we can express $\overline{\rm SDNR}_{\rm avg,box}^{l}$ as:
\[
\overline{\rm SDNR}_{\rm avg,box}^{l} = \frac{\rho \left( 1 - \frac{\beta^\star}{2\sqrt{\delta(\sigma^2{\rm SNR}_{\rm tx}+\rho)}} \right)^2}{\frac{(\beta^\star)^2 \sigma^2 {\rm SNR}_{\rm tx}}{4 (\sigma^2 {\rm SNR}_{\rm tx} + \rho) \delta} + \sigma^2}.
\]

From the above expression for $\overline{\rm SDNR}_{\rm avg,box}^{l}$, it is clear that the performance of the precoder depends not only on the target ${\rm SNR}_{\rm tx}$ but also on the noise variance $\sigma^2$. This contrasts with a linear precoder, where the performance is solely determined by the value of ${\rm SNR}_{\rm tx}$. The difference arises from the non-linear nature of the precoder. In this case, tuning $\rho$ to achieve the target ${\rm SNR}_{\rm tx}$ introduces a non-linear dependency of the parameters on both the target ${\rm SNR}_{\rm tx}$ and the noise variance.

\section{Performance analysis of the quantized box-constrained precoder}
\label{sec:quantized}
The objective of this section is to derive the asymptotic performance metrics of the quantized precoder. These metrics depend on the saddle point of the optimization problem in \eqref{eq:scalar_opt} when $\rho = 1$. For simplicity, throughout this section, we denote the solutions to the scalar optimization problem in \eqref{eq:scalar_opt} with $\rho = 1$ as $\tau^\star$ and $\beta^\star$. Before presenting our main results, we introduce the following notations. Let 
\begin{align*}
\xi^\star &= \frac{L \mathbb{E}[|H|]}{\delta \tau^\star},\\
\zeta^\star &= L^2 - \frac{2L^2 \mathbb{E}[|H|]}{\tau^\star \delta} \mathbb{E}[|X(H)|] + \frac{L^2 (\mathbb{E}[|H|])^2}{(\tau^\star \delta)^2} \left( (\tau^\star)^2 \delta - 1 \right)
\end{align*}
where $H\sim \mathcal{N}(0,1)$. 
Let $\hat{\bf d}_q:={\bf H}{\bf x}_q$ denote the distortion vector for the quantized precoder. 
 Define:
\begin{align*}
\tilde{E}(G,S)=\sqrt{\zeta^\star} G+\xi^\star S, 
\end{align*}
and denote by $\tilde{\mu}_{ES}^\star$ the joint distribution of the couple $(\tilde{E}(G,S),S)$. Then, we prove that the joint empirical distribution $\hat{\mu}(\hat{\bf d}_q,{\bf s})$ can be approximated in the asymptotic regime by  $\tilde{\mu}_{ES}^\star$ in the following sense:
\begin{theorem}[Convergence of the joint empirical measure of the distortion vector and the symbol vector] Under Assumption \ref{ass:growth_rate}, \ref{ass:channel_model} and \ref{ass:cond_lambda}, there exists positive constants $\chi$, $C$ and $c$  such that for all $\epsilon$ sufficiently small:
$$
\mathbb{P}\Big[(\mathcal{W}_2(\hat{\mu}(\hat{\bf d}_q,{\bf s}), \tilde{\mu}_{ES}^\star))^2>\chi\epsilon^2\Big]\leq \frac{C}{\epsilon^6}\exp(-cn\epsilon^{12}).
$$\begin{proof} {See} Appendix \ref{proof_main}.\end{proof}
\label{th:disortion_vector_quantized}
\end{theorem}
Similar to the box-precoder, the above theorem allows us to define the following asymptotic characterization of the joint distribution of the received signal and the transmitted symbol for the quantized precoder.
\begin{definition}[Asymptotic distributional characterization of the received signal for the quantized precoder]
Let 
\begin{align*}
\tilde{Y}=\tilde{E}(G,S)+Z 
\end{align*}
and let $\tilde\mu_{YS}^\star$ be the joint distribution of $(\tilde{Y},S)$.
\end{definition}
Our result allows for establishing asymptotic convergence of the joint empirical distribution $\frac{1}{m}\sum_{i=1}^m \delta_{[{\bf Hx}_q+{\bf z}]_i,s_i}$ and $\tilde{\mu}_{YS}^\star$. Based on the expression of $\tilde{E}(G,S)$, we assume that the users scale the received signal by $\kappa=\frac{1}{\xi^\star}$. By doing so, the scaled received signal behaves like:
\begin{equation}
\kappa {Y}=S+\frac{\sqrt{\zeta^\star}}{\xi^\star} {G}+\frac{1}{\xi^\star}Z. \label{eq:asym_charac}
\end{equation}

Based on the asymptotic characterization in \eqref{eq:asym_charac}, the lower bound of the SDNR satisfies the following convergence. 
\begin{theorem}[Convergence of the SDNR lower bound] Under Assumption \ref{ass:growth_rate}, \ref{ass:channel_model} and \ref{ass:cond_lambda}, the lower bound of the average SDNR of the quantized precoder converges to:
$$
{\rm SDNR}_{{\rm avg},q}^{l}\xrightarrow[]{P}\overline{\rm SDNR}_{{\rm avg},q}^{l}:=\frac{(\xi^\star)^2}{\zeta^\star +\sigma^2}.
$$
\end{theorem}
Consider a detector estimating the vector ${\bf s}$ as $\hat{\bf s}={\rm sign}(\kappa {\bf y})$. Then, the BER of the quantized precoder satisfies the following convergence. 
\begin{theorem}[Convergence of the BER] Under Assumption \ref{ass:growth_rate}, \ref{ass:channel_model} and \ref{ass:cond_lambda}, the BER of the quantized precoder converges to:
$$
{\rm BER}_q\xrightarrow[]{P}\overline{\rm BER}_{q}:=Q(\frac{\xi^\star}{\sqrt{\zeta^\star +\sigma^2}}).
$$
\end{theorem}
Before proceeding, let us consider a scenario similar to that of the box precoder, where the goal is to achieve a target transmit SNR value \({\rm SNR}_{\rm tx} = \frac{L^2}{\sigma^2}\). Accordingly, \(L\) must be set to \(L = \sqrt{{\rm SNR}_{\rm tx} \sigma^2}\). In this case, it is straightforward to observe that \(\overline{\rm SDNR}_{\rm avg,q}^l\) simplifies to:
\[
\overline{\rm SDNR}_{\rm avg,q}^l = \frac{{\rm SNR}_{\rm tx} \frac{(\mathbb{E}[|H|])^2}{\delta^2(\tilde{\tau}^\star)^2}}{{\rm SNR}_{\rm tx} \mathbb{E}\left[\left(1 - \frac{\mathbb{E}[|H|]X(G)}{\tau^\star\delta}\right)^2\right] + 1}.
\]
Unlike the box-precoder, the performance of the quantized precoder depends solely on the target SNR \({\rm SNR}_{\rm tx}\) and is independent of the noise variance \(\sigma^2\). This distinction arises because, despite its non-linear nature, the power of the quantized precoder depends linearly on the value of \(L^2\). Additionally, parameters such as \(\tau^\star\) and \(\beta^\star\) depend on \(A\), which, in this context, is not influenced by power but instead acts as a parameter controlling the quantization noise.

\noindent{\bf Comparison with Bussgang's decomposition based methods.}
Let $\hat{\bf x}$ be the output of the non-quantized precoder. 
Heuristic techniques based on Bussgang’s decomposition model $\hat{\bf x}$ as a Gaussian random vector. In our case, however, it is evident that $\hat{\bf x}$ cannot be accurately approximated by a Gaussian distribution when $A$ is finite, as indicated in \eqref{eq:XH}. Nevertheless, we temporarily set aside this limitation and proceed with Bussgang’s decomposition to evaluate its inaccuracy under these conditions.
As the dimensions grow large, the variance of the elements of $\hat{\bf x}$ is approximated to $\delta(\tau^\star)^2-1$. Following the same rationale as Bussgang’s decomposition, we analyze the output after applying the non-linearity $\mathcal{Z}(\hat{\bf x})$:
$$
\mathcal{Z}(\hat{\bf x})\simeq {\Theta} \hat{\bf x}+\boldsymbol{\eta}
$$
where ${\Theta}= \frac{\mathbb{E}[R\mathcal{Z}(\sqrt{(\tau^\star)^2\delta-1)}R)]}{\sqrt{(\tau^\star)^2\delta-1}}=\frac{L\mathbb{E}[|R| ]}{\sqrt{(\tau^\star)^2\delta-1}}$, $R\sim \mathcal{N}(0,1)$,  and $\boldsymbol{\eta}$ represents residual noise that is uncorrelated with $\hat{\bf x}$. Thus, 
$$
\mathbb{E}[\boldsymbol{\eta}\boldsymbol{\eta}^{T}]=L^2\left(1-(\mathbb{E}[|R|])^2\right){\bf I}_n.
$$
Under the Gaussian approximation  assumption, the received signal ${\bf y}$ can be expressed as:
$$
{\bf y}\simeq \Theta {\bf H}\hat{\bf x} +{\bf H}\boldsymbol{\eta} +{\bf z}.
$$Approximating ${\bf H}\boldsymbol{\eta}$ as a vector with zero mean and variance
$$
\mathbb{E}[{\bf H}\boldsymbol{\eta}\boldsymbol{\eta}^{T}{\bf H}^T]= L^2\left(1- (\mathbb{E}[|R|])^2\right){\bf I}_n,
$$
and pursuing similar heuristics, one might apply the distributional characterization from the analysis in \eqref{eq:distr_charac}. Based on Bussgang’s decomposition, the following asymptotic characterization is deduced for the joint distribution of the received signal of the quantized precoder and the transmitted symbol vector:
$$
Y=\Theta E(G,S) + L\sqrt{1- (\mathbb{E}[|R|])^2 }\tilde{G}+Z= \frac{L\mathbb{E}[|R|]}{\sqrt{\delta(\tau^\star)^2-1}}\left(\beta^\star \frac{\sqrt{(\tau^\star)^2\delta-1}G}{2\tau^\star\delta}+  S(1-\frac{\beta^\star} {2\tau^\star\delta})\right)+ L\sqrt{1- (\mathbb{E}[|R|])^2 }\tilde{G}+Z
$$
where $\tilde{G}$ is uncorrelated from ${S}$ and $G$. 
 Recall the distributional characterization $$\tilde{Y}=\tilde{E}(G,S)+Z=\sqrt{\zeta^\star} G+\xi^\star S+Z$$ derived from rigorous analysis. 
One can refer to Theorem $3$ and $4$ of our previous work \cite{tsp} to help verify that, the distributional characterization of the quantized precoder, when the box constraint $A$ is set to $\infty$,  aligns with the result obtained from Bussgang’s decomposition. However, it is evident that the two distributional characterizations diverge when $A$ is finite. This discrepancy highlights the limitations of the heuristic method based on Bussgang’s decomposition in this particular setting.

\section{Numerical results.}\label{sec:numerical_results}

In this section, we present numerical results to evaluate the accuracy of our findings and to compare the practical characteristics of the box-precoder and the quantized precoder. Specifically, our analysis focuses on the SDNR lower bound and BER, examining both theoretical predictions and their empirical counterparts. In all simulations, markers represent empirical results, while solid lines indicate theoretical predictions.

\subsection{Investigation of the impact of $\rho$ and $A$  and the noise variance $\sigma^2$ on the performance of the box-precoder}
\begin{figure}  [h]
		\centering
		\includegraphics[width=3in]{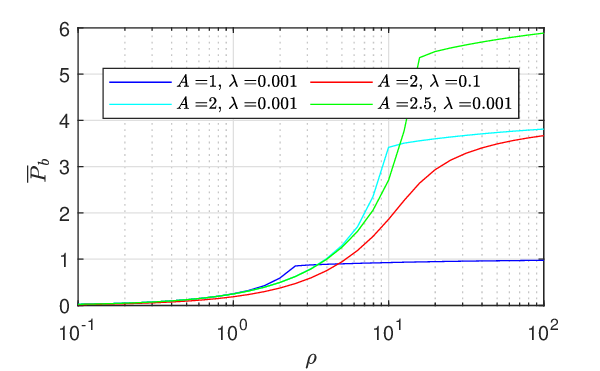}
		\caption{Averaged per-antenna transmit power $\bar P_b$ of the box-precoder varies with $\rho$. $\delta=0.2$.}
		\label{fig1}
	\end{figure}
In previous sections, we referred to $\rho$ as the power control parameter. This is because $\rho$ can be adjusted to enable the box-precoder to operate at any average per-antenna transmit power $\overline{P}_b<P=A^2$. This behavior is visualized in Fig.\ref{fig1}, where we observe that a smaller $\rho$ results in a lower average transmit power. Conversely, increasing $\rho$ causes this power to rise, eventually approaching the maximum allowed power $P$. 
It is also important to note that, since $\lambda$ penalizes the power, there are different $(\lambda,\rho)$ combinations that yield the same $\overline{P}_b$.
     
This property offers a practical means of optimizing the value of the regularization parameter \(\lambda\). For any given value of \(\lambda\), the parameter \(\rho\) can be adjusted to achieve the desired \(\overline{P}_b\). By doing so, the value of \(\lambda\) that maximizes performance can then be selected. This approach is utilized to design the optimal regularized box precoder.  
Next, we investigate the impact of the noise variance \(\sigma^2\) under the condition that the transmit signal-to-noise ratio (transmit SNR), defined as  $
{\rm SNR}_{\rm tx} = \frac{\overline{P}_b}{\sigma^2}\in(0,\frac{P}{\sigma^2}),
$
is fixed at a specific target SNR value, and the optimal regularization parameter $\lambda$ is numerically determined. 
\begin{figure}  [h]
		\centering
		\includegraphics[width=3in]{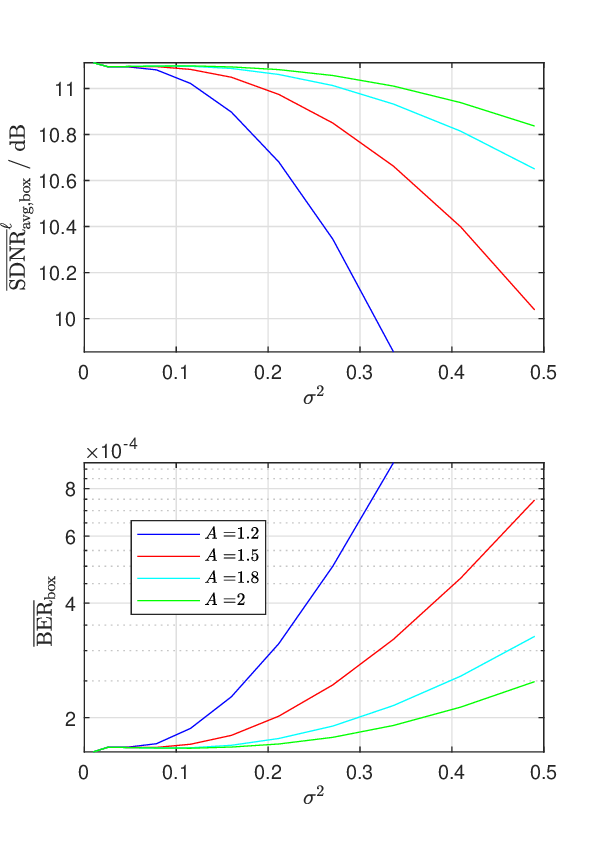}
		\caption{Performance of the box-precoder varies with noise when ${\rm SNR}_{\rm tx}=5{\rm dB}$.  $(\lambda,\rho)$ is searched to reach the desired $\overline{P}_b$ and make performance optimal. $\delta=0.2$. }
		\label{fig2}
	\end{figure}
Fig. \ref{fig2} represents the SDNR and the BER of the box-precoder versus the variance noise $\sigma^2$ for different values of $P=A^2$. We can see that its performance degrades as noise increases, even when the transmit SNR is maintained by simultaneously increasing $\overline{P}_b$.
This degradation is attributed to the finite power constraint 
$A$: although tuning $\rho$ allows control of $\overline{P}_b$, a high $\overline{P}_b$ close to $P=A^2$ requires a very large $\rho$, which forces the amplitude of precoded entries to approach $A$. This restriction reduces the degree of freedom in the precoder solution. To mitigate this degradation, a higher $A$ is required. However, achieving this would necessitate a power amplifier with a higher dynamic range.

\subsection{Investigating the role of $A$ in managing quantization noise for the quantized precoder}
\begin{figure} 
		\centering
		\includegraphics[width=3in]{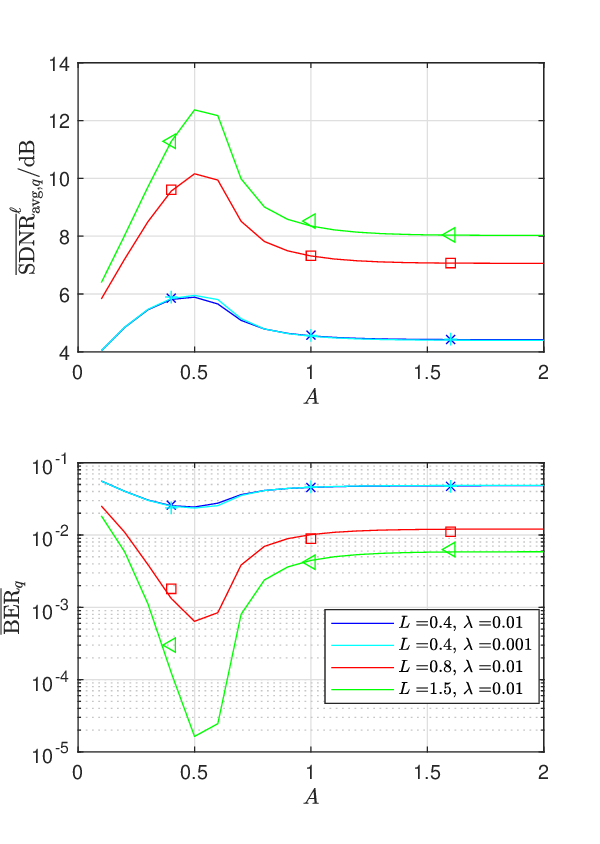}
		\caption{ 
        Performance of the the quantized precoder versus $A$. Results are obtained for $n=1000$, $\delta=0.2$ and $\sigma^2=0.09$. Empirical results are averaged over $50$ realizations of $({\bf H}, {\bf s}, {\bf z})$.
        }
		\label{fig3}
	\end{figure}
We consider a scenario where the parameter \(L\) of the quantized precoder is chosen to achieve a specific target \({\rm SNR}_{\rm tx}\) value. This is accomplished by setting \(L = \sqrt{{\rm SNR}_{\rm tx} \sigma^2}\). As previously discussed, the performance of the quantized precoder depends solely on \({\rm SNR}_{\rm tx}\) and not on the noise variance \(\sigma^2\).  
Under this configuration, we evaluate the performance of the quantized precoder as a function of \(A\) for various values of \(L\) and the regularization parameter \(\lambda\). As illustrated in Figure~\ref{fig3}, the performance can be optimized by appropriately selecting the value of \(A\). Furthermore, it is observed that while the optimal range for \(A\) remains consistent across all considered parameters, the performance improvement becomes more significant for larger values of \(L\).  
Similar to the case of the box precoder, Figure~\ref{fig3} provides a practical framework for simultaneously optimizing the values of \(\lambda\) and \(A\) to achieve the best performance. 
    \begin{figure}
     \centering
     \begin{subfigure}[b]{0.45\textwidth}
         \centering
         \includegraphics[width=\textwidth]{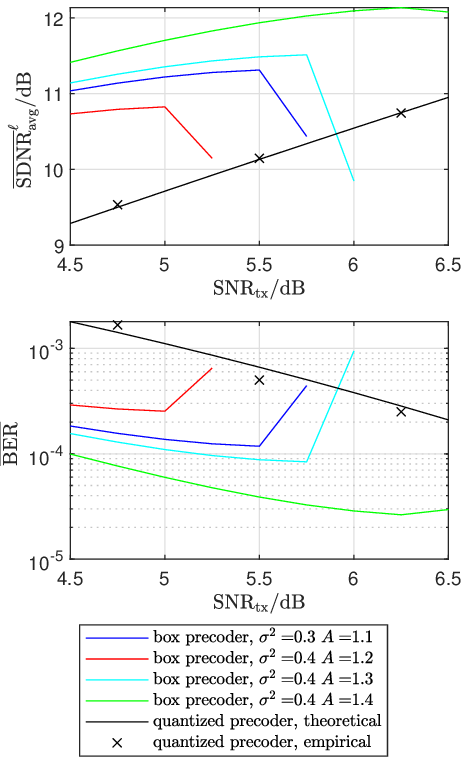}
     \end{subfigure}
    ~
     \begin{subfigure}[b]{0.45\textwidth}
         \centering
         \includegraphics[width=\textwidth]{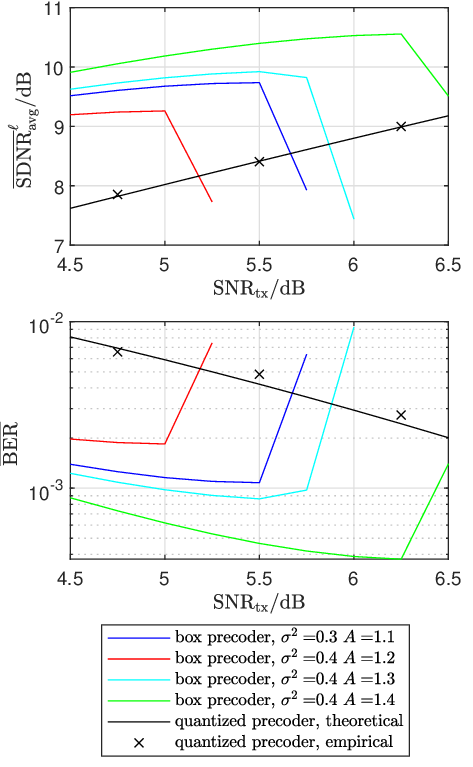}
     \end{subfigure}
        \caption{Performance comparison between the box-precoder and the quantized precoder versus the transmit SNR. Tunable parameters $(\lambda,\rho)$ for the box-precoder and $(A,\lambda)$ for the quantized precoder are set to their optimal values. $n$ is set to $800$. Empirical results for the quantized precoder are averaged over $50$ realizations of $({\bf H,s,z})$. Left: $\delta=0.15$; Right: $\delta=0.2$.}
        \label{fig4}
\end{figure}

\subsection{Comparing the performance of the box-precoder and the quantized precoder. }In Fig.~\ref{fig4}, we compare the performance of the box-precoder and the quantized precoder. The results reveal that the performance of the quantized precoder steadily improves as the transmit SNR increases. In contrast, the box-precoder exhibits a performance peak, which is followed by a decline due to the maximum allowed power constraint. Notably, this performance peak occurs earlier when the spatial degrees of freedom ($\delta$) are reduced or when noise levels are higher.  
As a result of the box-precoder's performance degradation at higher transmit SNR values, the quantized precoder demonstrates superior performance in high-SNR regimes.

\section{Conclusion} 
In this paper, we provided a precise asymptotic characterization of the performance of a quantized precoder, derived by applying quantization to the output of the box-constrained precoder, within the asymptotic regime where the number of antennas and users grow proportionally. Unlike the traditional use of the box constraint for low-dynamic range power amplifiers, here it is employed to control the impact of quantization noise. The introduction of the box constraint makes our precoder non-linear, without a closed-form solution. Unlike prior works relying on Bussgang's decomposition, our theoretical analysis is based on a novel Gaussian min-max theorem, which we developed to derive accurate and rigorously proven asymptotic performance metrics. We believe that the tools introduced here provide valuable insights that could inform future studies on the impact of hardware impairments in large-scale multi-user systems.
\appendices
\section{About the standard Convex Gaussian Min-max Theorem (CGMT) framework.}
This section provides a concise overview of the existing literature on the CGMT. 

\label{sec:standard}
\noindent{\bf Review of the CGMT. } The CGMT framework consists of two probability comparisons that respectively relate the left and right tail probabilities of a primary  process to those of an auxiliary  process \cite{ori,mesti}. More formally, consider the following Gaussian processes:
\begin{align*}
\text{Primary optimization problem:}   \ \ \ X({\bf x},{\bf u})&={\bf u}^{T}{\bf G}{\bf x} +\psi({\bf x},{\bf u})\\
\text{Auxiliary optimization problem:} \ \ \ \ Y({\bf x},{\bf u})&= \|{\bf x}\|{\bf g}^{T}{\bf u}-\|{\bf u}\|{\bf h}^{T}{\bf x}+ \psi({\bf x},{\bf u})
\end{align*}
where ${\bf G}^{m\times n}$, ${\bf g}\in \mathbb{R}^m$ and ${\bf h}\in \mathbb{R}^n$ are independent and have all independent standard Gaussian entries and function $\psi:\mathbb{R}^n\times \mathbb{R}^m\to\mathbb{R}$  is continuous. 
Consider $\mathcal{S}_x$, $\mathcal{S}_u$ two compact sets in $\mathbb{R}^n$ and $\mathbb{R}^m$ . 
Then for any $t\in \mathbb{R}$, 
\begin{equation}
\mathbb{P}[\min_{{\bf x}\in \mathcal{S}_x}\max_{{\bf u}\in \mathcal{S}_u} X({\bf x},{\bf u})\leq t]\leq 2\mathbb{P}[\min_{{\bf x}\in \mathcal{S}_x}\max_{{\bf u}\in \mathcal{S}_u} Y({\bf x},{\bf u})\leq t ].\label{eq:gordon}
\end{equation}
If additionally, the sets $\mathcal{S}_x$ and $\mathcal{S}_u$ are convex and the function $\psi$ is convex in ${\bf x}$ and concave in ${\bf u}$, then the order of the min-max operation can be inverted, and we obtain a comparison between the right-tail probabilities of the primary and auxiliary optimization problems. Specifically, 
 for any $t\in \mathbb{R}$, 
\begin{equation}
\mathbb{P}[\min_{{\bf x}\in \mathcal{S}_x}\max_{{\bf u}\in \mathcal{S}_u} X({\bf x},{\bf u})\geq t]\leq 2\mathbb{P}[\min_{{\bf x}\in \mathcal{S}_x}\max_{{\bf u}\in \mathcal{S}_u} Y({\bf x},{\bf u})\geq t ] .\label{eq:cgmt}
\end{equation}
\noindent{\bf Implications for Practical Applications.} In practice, inferring properties of solutions to the auxiliary optimization problem is much easier than inferring those for the primary optimization problem. This allows us to transfer properties from the auxiliary process to the primary process. For instance, suppose the optimal cost of the auxiliary process concentrates around a constant $\overline{\phi}$ as the dimensions $n$ and $m$ grow large. Specifically, for any $\epsilon>0$, 
$$
\mathbb{P}\Big[\min_{{\bf x}\in \mathcal{S}_x}\max_{{\bf u}\in \mathcal{S}_u} Y({\bf x},{\bf u})\leq \overline{\phi}-\epsilon\Big] \to 0
$$
and 
$$
\mathbb{P}\Big[\min_{{\bf x}\in \mathcal{S}_x}\max_{{\bf u}\in \mathcal{S}_u} Y({\bf x},{\bf u})\geq \overline{\phi}+\epsilon\Big] \to 0.
$$
Under these conditions, the probability inequalities in \eqref{eq:gordon} and \eqref{eq:cgmt} imply that:
$$
\mathbb{P}\Big[\min_{{\bf x}\in \mathcal{S}_x}\max_{{\bf u}\in \mathcal{S}_u} X({\bf x},{\bf u})\leq \overline{\phi}-\epsilon\Big] \to 0
$$
and 
\begin{equation}
\mathbb{P}\Big[\min_{{\bf x}\in \mathcal{S}_x}\max_{{\bf u}\in \mathcal{S}_u} X({\bf x},{\bf u})\geq \overline{\phi}+\epsilon\Big] \to 0. \label{eq:PO}
\end{equation}
Thus, the optimal cost of the primary optimization problem also concentrates around $\overline{\phi}$. 

If the optimal cost of the auxiliary optimization problem concentrates, and if the solution to the auxiliary optimization problem does not belong to a certain set with probability approaching one, this property can be transferred to the solution of the primary optimization problem.
More formally, assume  there exists a deterministic set $\tilde{\mathcal{S}}_x\subset\mathcal{S}_x$ such that:
$$
\mathbb{P}\Big[\min_{{\bf x}\in \tilde{\mathcal{S}}_x}\max_{{\bf u}\in \mathcal{S}_u} Y({\bf x},{\bf u})\leq \overline{\phi}+2\epsilon\Big] \to 0,
$$
then from \eqref{eq:gordon}
$$
\mathbb{P}\Big[\min_{{\bf x}\in \tilde{\mathcal{S}}_x}\max_{{\bf u}\in \mathcal{S}_u} X({\bf x},{\bf u})\leq \overline{\phi}+2\epsilon\Big] \to 0,
$$
and thus necessarily the solution of the primary optimization problem does not belong to $\tilde{\mathcal{S}}_x$ since from \eqref{eq:PO} with probability approaching one,
$$
\min_{{\bf x}\in \mathcal{S}_x}\max_{{\bf u}\in \mathcal{S}_u} X({\bf x},{\bf u})\leq \overline{\phi}+\epsilon.
$$

\section{Review of the results regarding the non-quantized precoder $\hat{\bf x}_{\rm box}$}

\label{sec:review}
In this section, we review key results on the behavior of the non-quantized precoder presented in our previous work \cite{as_arxiv}. These results will be used as essential components in the proof of our findings for the quantized precoder.

Starting from the formulation of the non-quantized precoder, we use the relation $\|{\bf y}\|^2=\max_{{\bf u}}{\bf u}^{T}{\bf y}-\frac{\|{\bf u}\|^2}{4}$ to rewrite $\mathcal{P}({\bf x};\rho)$ as:
$$
 \mathcal{P}({\bf x};\rho):= \max_{{\bf u}} \frac{1}{\sqrt{n}}{\bf u}^{T}{\bf H}{\bf x} -\frac{\sqrt{\rho}}{\sqrt{n}}{\bf u}^{T}{\bf s}-\frac{\|{\bf u}\|^2}{4}+\frac{\lambda}{n}\|{\bf x}\|^2,
$$
leading to  the following primary optimization problem:
$$
\min_{\|{\bf x}\|_{\infty}\leq B} \max_{{\bf u}}\frac{1}{\sqrt{n}}{\bf u}^{T}{\bf H}{\bf x} -\frac{\sqrt{\rho}}{\sqrt{n}}{\bf u}^{T}{\bf s}-\frac{\|{\bf u}\|^2}{4}+\frac{\lambda}{n}\|{\bf x}\|^2.
$$
To analyze the behavior of the solution to this problem, we introduce the associated auxiliary optimization problem: 
$$
\min_{\|{\bf x}\|_{\infty}\leq B} \mathcal{L}({\bf x})
$$
where 
$$
 \mathcal{L}({\bf x}):=\max_{{\bf u}} \frac{1}{n}{\bf u}^{T}{\bf g}\|{\bf x}\|-\|{\bf u}\|\frac{1}{n}{\bf h}^{T}{\bf x} -\frac{\sqrt{\rho}}{\sqrt{n}}{\bf u}^{T}{\bf s}-\frac{\|{\bf u}\|^2}{4}+\frac{\lambda}{n}\|{\bf x}\|^2.
$$
In \cite{as_arxiv}, we investigated the properties of the auxiliary optimization problem, as well as a random equivalent to its solution, to derive insights into the primary optimization problem. Below, we summarize the key results that will be used in our proofs. 
\subsection{Asymptotic equivalent for the optimal cost.}\label{BA}
In C.\Romannum{1} of \cite{as_arxiv}, we showed that $\mathcal{L}({\bf x})$ is asymptotically equivalent to $\mathcal{L}^{\circ}({\bf x})$ defined as:
$$
\mathcal{L}^{\circ}({\bf x}):=\Big(\sqrt{\delta}\sqrt{\frac{\|{\bf x}\|^2}{n}+\rho}-\frac{1}{n}{\bf h}^{T}{\bf x}\Big)_{+}^2+\frac{\lambda}{n}\|{\bf x}\|^2
$$
in that with probability at least $1-C\exp(-cn\epsilon^2)$
\begin{equation}
\sup_{\|{\bf x}\|_\infty\leq B}|\mathcal{L}({\bf x})-\mathcal{L}^{\circ}({\bf x})|\leq K\epsilon. \label{eq:f1}
\end{equation}
By studying the behavior of $\mathcal{L}^\circ({\bf x})$, we proved that $\min_{\|{\bf x}\|_{\infty}\leq B}\mathcal{L}({\bf x})$ concentrates around $\varphi(\tau^\star(\rho,B),\beta^\star(\rho,B))$, where $\tau^\star(\rho,B),\beta^\star(\rho,B)$ is the unique saddle point of the scalar optimization problem in \eqref{eq:scalar_opt}. Specifically, we proved that there exists constants $C$ and $c$ and $\gamma$ such that for all $\epsilon$ sufficiently small:
\begin{equation}
\mathbb{P}\Big[\Big|\min_{\|{\bf x}\|_\infty\leq B}\mathcal{L}({\bf x})-\varphi(\tau^\star(\rho,B),\beta^\star(\rho,B))\Big|\geq \gamma\epsilon\Big]\leq \frac{C}{\epsilon}\exp(-cn\epsilon^2). \label{eq:f2}
\end{equation}
Using both probability inequalities of the CGMT in \eqref{eq:gordon} and \eqref{eq:cgmt}, this directly implies that:
\begin{equation}
\mathbb{P}\Big[\Big|\min_{\|{\bf x}\|_\infty\leq B}\mathcal{P}({\bf x};\rho)-\varphi(\tau^\star(\rho,B),\beta^\star(\rho,B))\Big|\geq \gamma\epsilon\Big]\leq \frac{C}{\epsilon}\exp(-cn\epsilon^2).
\label{eq:previous_result}
\end{equation}
\subsection{Asymptotic location of the solution to the auxiliary optimization.} Let $\overline{\bf y}^{\rm AO}:=X({\bf h})$  where function $X$ defined in \eqref{eq:XH} is applied entry-wise to each entry of the vector ${\bf h}$. Then, in C.\Romannum{1} of \cite{as_arxiv} we also proved that,  the minimizer of $\mathcal{L}({\bf x})$ lies within a ball centered around $\overline{\bf y}^{\rm AO}$ with probability approaching one. More specifically, there exists constants $C$ and $c$ such that for all $\epsilon>0$ sufficiently small with probability at least $1-C\exp(-cn\epsilon^2)$, 
\begin{equation}
\forall \ \ \|{\bf x}\|_{\infty}\leq B, \ \ \frac{1}{n}\|{\bf x}-\overline{\bf y}^{\rm AO}\|^2\geq \epsilon  \ \ \Longrightarrow \ \ \mathcal{L}({\bf x})\geq \varphi(\tau^\star(\rho,B),\beta^\star(\rho,B))+\gamma \epsilon. \label{eq:useful}
\end{equation}
Moreover, based on standard Gaussian concentration inequalities, we proved that $\overline{\bf y}^{\rm AO}$ satisfies the following concentration inequalities:
\begin{align}
    &\mathbb{P}\Big[\Big|\frac{1}{n}\|\overline{\bf y}^{\rm AO}\|^2-(\delta (\tau^\star(\rho,B))^2-\rho)\geq \epsilon\Big|\Big]\leq C\exp(-cn\epsilon^2) \label{eq:conc1}\\
    &\mathbb{P}\Big[\Big|\frac{1}{n}{\bf h}^{T}\overline{\bf y}^{\rm AO}-\mathbb{E}[X(H)H]\Big|\geq \epsilon\Big]\leq C\exp(-cn\epsilon^2) \label{eq:conc2}
\end{align}
where $C$ and $c$ are some constants and $\epsilon>0$. 
\subsection{Asymptotic characterization of the solution to the primary optimization.}
Based on \eqref{eq:useful}, in C.\Romannum{2},\Romannum{3} and \Romannum{4} of \cite{as_arxiv}, we establish the concentration properties of various quantities derived from the  solution of the primary optimization problem. Below, we summarize two key results in this regard.
\subsubsection{Concentration of Lipschitz functionals of the solution to the primary optimization}
Let $f$ be a $1$-Lipschitz function. Consider the set $\mathcal{S}_x$ defined as:
$$
\mathcal{D}_0=\{{\bf x}\in \mathbb{R}^n,   \Big|{f}(\frac{1}{\sqrt{n}}{\bf x})-\mathbb{E}[f(\frac{1}{\sqrt{n}}\overline{\bf y}^{\rm AO})]\Big|\geq \epsilon\}.
$$
Then,  there exists constants $C$ and $c$ such that for all $\epsilon$ sufficiently small:
\begin{equation}
\mathbb{P}\Big[\min_{\substack{\|{\bf x}\|_{\infty}\leq B\\ {\bf x}\in \mathcal{D}_{0}}}\mathcal{P}(x;\rho)\leq \varphi(\tau^\star(\rho,B),\beta^\star(\rho,B))+\gamma\epsilon \Big] \leq \frac{C}{\epsilon}\exp(-cn\epsilon^2). \label{eq:lipschitz}
\end{equation}
 \subsubsection{Concentration of fraction of elements in an interval} For ${\bf x}\in \mathbb{R}^n$ and $a,b$ two reals in the interval $(-B,B)$, denote by $\pi_{(a,b)}({\bf x})$ the quantity given by:
 $$
 \pi_{(a,b)}({\bf x})=\#\{i=1,\cdots, n, \ \ [{\bf x}]_i\in (a,b)\}
 $$
 which represents the number of elements of ${\bf x}$ in the interval $(a,b)$. Let $\mathcal{S}_x$ be the set:
 $$
 \mathcal{D}_{1}=\{{\bf x} \ |\ \frac{\pi_{(a,b)}({\bf x})}{n}-Q(\frac{a}{\tilde{\tau}^\star})+Q(\frac{b}{\tilde{\tau}^\star})|\geq \tilde{\lambda}\epsilon\}
 $$
 where $\tilde{\lambda}=\min(1,\frac{(\tilde{\tau}^\star)^2}{32})$. 
 Then,
 \begin{equation}
 \mathbb{P}\Big[\min_{\substack{\|{\bf x}\|_{\infty}\leq B\\ {\bf x}\in \mathcal{D}_{1}}}\mathcal{P}(x;\rho)\leq \varphi(\tau^\star(\rho,B),\beta^\star(\rho,B))+\gamma\epsilon^3 \Big] \leq \frac{C}{\epsilon^3}\exp(-cn\epsilon^{6}) .\label{eq:number}
 \end{equation}
\section{Proof of Theorem \ref{th:disortion_vector_quantized}}\label{proof_main}
In this section, we define a constant as any quantity that depends solely on the fixed parameters $\delta$, $\lambda$, $A$,  $\sigma^2$ and $L$. Unless otherwise specified, the symbols $C$ and $c$ will represent constants, which may vary between different equations or lines.
\subsection{Preliminaries and organization of the proof.}
\noindent{\bf Goal.} Let \(\hat{\bf d}_q =  {\bf H} \hat{\bf x}_q \), where \(\hat{\bf x}_q =\mathcal{Z}({\bf x})= L \, {\rm sign}(\hat{\bf x})\) with $\hat{\bf x}$ being the solution to \eqref{rls_Q}.  The goal of this section is to prove that:
\begin{equation}
\mathbb{P}\Big[\hat{\bf d}_q \in \mathcal{S}_d^{\epsilon} \Big] \leq \frac{C}{\epsilon^6}\exp(-cn\epsilon^{12}) \label{eq:desired}
\end{equation}
where $\mathcal{S}_d^\epsilon$ is the set defined as:
$$
\mathcal{S}_d^\epsilon:=\{{\bf d}\in \mathbb{R}^m \ |\ (\mathcal{W}_2(\hat{\mu}({\bf d},{\bf s}),\tilde{\mu}^\star_{ES}))^2>\chi\epsilon^2\}
$$
and $\epsilon$ is sufficiently small.
For ${\bf x}\in \mathbb{R}^{n}$, define $d_T({\bf x})={\bf H}\mathcal{Z}({\bf x})$. The proof of \eqref{eq:desired} amounts thus to showing that for all ${\bf x}$ such that $\|{\bf x}\|_{\infty}\leq A$ and $d_{T}({\bf x})\in \mathcal{S}_d^\epsilon$, the value of ${\displaystyle\min_{\substack{\|{\bf x}\|_{\infty}\leq A\\ {d}_T({\bf x})\in \mathcal{S}_d^\epsilon}}\mathcal{P}({\bf x})}$ exceeds the optimal cost $\mathcal{P}(\hat{\bf x})$ with probability approaching one. In other words, 
\begin{equation}
\mathbb{P} \Big[ \min_{\substack{\|{\bf x}\|_{\infty} \leq A\\  d_T({\bf x})\in \mathcal{S}_d^\epsilon}}  \mathcal{P}({\bf x}) \leq \min_{\|{\bf x}\|_{\infty} \leq A} \mathcal{P}({\bf x}) \Big] \leq \frac{C}{\epsilon^6}\exp(-cn\epsilon^{12}) .\label{eq:tobe_proven1}
\end{equation}

\noindent{\bf Preliminaries. } The proof relies on a novel Gaussian min-max theorem, which, as will be shown next, enables the handling of probability inequalities involving the min-max of Gaussian processes. However, some preparatory steps are required to enable the use of the new Gaussian min-max theorem. 

First, we exploit our previous result in \eqref{eq:previous_result} from classical convex Gaussian min-max theorem analysis of the solution in ${\bf x}$ that shows that:
\begin{equation}
\mathbb{P}\Big[\min_{\|{\bf x}\|_{\infty} \leq A} \mathcal{P}({\bf x})\geq \varphi(\tau^\star,\beta^\star)+\gamma\epsilon^{6}\Big]\leq \frac{C}{\epsilon^{6}}\exp(-cn\epsilon^{12}). \label{eq:comb1b}
\end{equation}
With this, we can upper bound the probability in the left-hand side of \eqref{eq:tobe_proven1} as:
\begin{align}
&\mathbb{P} \Big[\min_{\substack{\|{\bf x}\|_{\infty} \leq A\\  d_T({\bf x})\in \mathcal{S}_d^\epsilon}}  \mathcal{P}({\bf x}) \leq \min_{\|{\bf x}\|_{\infty} \leq A}\mathcal{P}({\bf x}) \Big]\nonumber \\
&\leq \mathbb{P} \Big[\min_{\substack{\|{\bf x}\|_{\infty} \leq A\\  d_T({\bf x})\in \mathcal{S}_d^\epsilon }}  \mathcal{P}({\bf x})\leq \varphi(\tau^\star,\beta^\star)+\gamma\epsilon^{6}\Big]+\frac{C}{\epsilon^{6}}\exp(-cn\epsilon^{12}) .\label{eq:comb1}
\end{align}
To prove \eqref{eq:tobe_proven1}, it thus suffices to show that:
\begin{equation}
\mathbb{P} \Big[ \min_{\substack{\|{\bf x}\|_{\infty} \leq A\\  d_T({\bf x})\in \mathcal{S}_d^\epsilon}}  \mathcal{P}({\bf x}) \leq \varphi(\tau^\star,\beta^\star)+\gamma\epsilon^{6}\Big]\leq\frac{C}{\epsilon^6}\exp(-cn\epsilon^{12}) .\label{eq:res}
\end{equation}

Next, using the classical convex Gaussian min-max theorem, we establish that there exists a set $\mathcal{S}_x^\circ$ such that the optimization outside a particular set \(\mathcal{S}_x^\circ\) results in a higher cost than \(\varphi(\tau^\star, \beta^\star) + \gamma \epsilon^6\). Specifically, the set \(\mathcal{S}_x^\circ\) is defined as follows 
\begin{equation*}
\mathcal{S}_x^\circ = \mathcal{S}_{x1}\cap \mathcal{S}_{x2}\cap \mathcal{S}_{x3} 
\end{equation*}
where
\begin{align*}
\mathcal{S}_{x1}&=\{{\bf x} \ | \ |\frac{1}{n}\|{\bf x}\|_1-\mathbb{E}|X||\leq C_{x1}\epsilon^3\},\\
\mathcal{S}_{x2}&=\{{\bf x} \ | \  |\frac{1}{\sqrt{n}}\|{\bf x}\|-\sqrt{\delta(\tau^\star)^2-1}|)|\leq C_{x2}\epsilon^3\},\\
\mathcal{S}_{x3}&=\{{\bf x} \ | \    \ C_ln\epsilon^2\leq  \pi_{(-L\epsilon^2,L\epsilon^2)}({\bf x}) \leq C_u n \epsilon^2\}.
\end{align*}
The set \(\mathcal{S}_x^\circ\) consists of vectors \({\bf x}\) that satisfy:  (1)  concentration requirements, where $\frac{1}{n}\|{\bf x}\|_1$ and $\frac{1}{\sqrt{n}}\|{\bf x}\|$ are within a tolerance of \(C_{x1} \epsilon^3\)  and \(C_{x2}\epsilon^3\); and  (2) a count constraint, with roughly \(O(n \epsilon^2)\) elements of \({\bf x}\) falling within \((-L \epsilon^2, L \epsilon^2)\), controlled by constants \(C_l\) and \(C_u\). The constants $C_{x1},C_{x2},C_l$ and $C_u$  are carefully chosen such that
\begin{equation}
\mathbb{P} \Big[ \min_{\substack{\|{\bf x}\|_{\infty} \leq A \\ {\bf x} \notin \mathcal{S}_x^\circ}} \mathcal{P}({\bf x})\leq \varphi(\tau^\star, \beta^\star) + \gamma \epsilon^6 \Big] \leq \frac{C}{\epsilon^6} \exp(-c n \epsilon^{12}) .\label{eq:necessary_bound}
\end{equation}

Finally, let
\begin{equation}
\mathcal{S}_{x4}=\{{\bf x} \  | \  \ \ \forall k=1,\cdots,4,  \ \ \#\pi_{\mathcal{I}_k}({\bf x})\geq C_{0x}n \}  \label{eq:Sx4}
\end{equation}
where $\mathcal{I}_1=(-\frac{4A}{5},-\frac{3A}{5})$, $\mathcal{I}_2=(-\frac{2A}{5},-\frac{A}{5})$, $\mathcal{I}_3=(\frac{A}{5},\frac{2A}{5})$ and $\mathcal{I}_4=(\frac{3A}{5},\frac{4A}{5})$. 
We can also show that there exists $C_0$ such that
\begin{equation}
\mathbb{P}\Big[\min_{\substack{\|{\bf x}\|_{\infty}\leq A\\ {\bf x}\notin\mathcal{S}_{x4}}}\mathcal{P}({\bf x})\leq \varphi(\tau^\star,\beta^\star)+\gamma \epsilon^3\Big]\leq \frac{C}{\epsilon^3}\exp(-cn\epsilon^6).\label{eq:Sx4_ineq}
\end{equation}
Note that \eqref{eq:Sx4_ineq} also implies that:
\begin{equation}
\mathbb{P}\Big[\exists {\bf x} \ \text{such that } {\bf x}\notin \mathcal{S}_{x4}\  \ \text{and } \mathcal{P}({\bf x})\leq \varphi(\tau^\star,\beta^\star)+\gamma \epsilon^3\Big]\leq \frac{C}{\epsilon^3}\exp(-cn\epsilon^6). \label{eq:Sx4_repr}
\end{equation}
Starting from \eqref{eq:necessary_bound}, we can 
 upper bound the probability term on the left-hand side of \eqref{eq:res} as:
\begin{align}
& \mathbb{P} \Big[ \min_{\substack{\|{\bf x}\|_{\infty} \leq A \\ d_T({\bf x})\in \mathcal{S}_d^{\epsilon}}}   \mathcal{P}({\bf x}) \leq \varphi(\tau^\star, \beta^\star) + \gamma \epsilon^6 \Big] \nonumber \\
 \leq& \mathbb{P} \Big[ \min_{\substack{  {\bf x} \in \mathcal{S}_x^\circ,\|{\bf x}\|_\infty\leq A\\ d_T({\bf x})\in \mathcal{S}_d^{\epsilon}}} \mathcal{P}({\bf x})\leq \varphi(\tau^\star, \beta^\star) + \gamma \epsilon^6 \Big] + \frac{C}{\epsilon^6} \exp(-c n \epsilon^{12}).\label{eq:comb2}
\end{align}
By combining \eqref{eq:comb1} and \eqref{eq:comb2}, we conclude that to prove \eqref{eq:tobe_proven1}, it suffices to show that:
\begin{equation}
\mathbb{P} \Big[ \min_{\substack{  {\bf x} \in \mathcal{S}_x^\circ,\|{\bf x}\|_\infty\leq A\\ d_T({\bf x})\in \mathcal{S}_d^{\epsilon}}} \mathcal{P}({\bf x})\leq \varphi(\tau^\star, \beta^\star) + \gamma \epsilon^6 \Big]\leq \frac{C}{\epsilon^6}\exp(-cn\epsilon^{12}) .\label{eq:essential}
\end{equation}
The keystone of the proof lies in establishing \eqref{eq:essential}.

Meanwhile, the proofs of \eqref{eq:comb1b}
 and \eqref{eq:necessary_bound} rely on leveraging the results discussed in Appendix \ref{sec:review}, which are derived from calculations involving the convex Gaussian min-max theorem outlined in Section \ref{sec:standard}. Indeed to prove \eqref{eq:necessary_bound}, it suffices to check that for each  set $D_x=\mathcal{S}_{xi}^c, i=1,\cdots,3$, 
 \begin{equation}
 \mathbb{P} \Big[ \min_{\substack{\|{\bf x}\|_{\infty} \leq A \\ {\bf x} \in D_x}} \mathcal{P}({\bf x})\leq \varphi(\tau^\star, \beta^\star) + \gamma \epsilon^6 \Big] \leq \frac{C}{\epsilon^6} \exp(-c n \epsilon^{12}). \label{eq:one_to_prove}
 \end{equation}
 When $D_x=\mathcal{S}_{x1}^{c}$ or $D_x=\mathcal{S}_{x2}^{c}$, \eqref{eq:one_to_prove} follows by applying \eqref{eq:lipschitz} to functions $f=\frac{1}{\sqrt{n}}\sum_{i=1}^n |[{\bf x}]_i|$ and $f({\bf x})=\frac{\|{\bf x}\|}{\sqrt n}$, respectively, along with the concentration inequality in \eqref{eq:conc1}. As for $D_x=\mathcal{S}_{x3}$, this follows directly from applying \eqref{eq:number} to the set:
 $$
 \hat{D}_{x3}:=\{{\bf x} \ | \ |\frac{\pi_{(-L\epsilon^2,L\epsilon^2)}({\bf x})}{n}-Q(\frac{-L\epsilon^2}{\tilde{\tau^\star}})+Q(\frac{L\epsilon^2}{\tilde{\tau}^\star})|\geq \epsilon^2\}
 $$
 and note that the absolute value of the derivative of function $g:x\mapsto Q(\frac{Lx}{\tilde{\tau^\star}})$ is bounded and bounded from below to zero when $x\in(0,1)$, meaning that there exists constants $m$ and $M$ such that $m\leq \sup_{0\leq x\leq 1}|g'(x)|\leq M$. This implies that $m\epsilon^2\leq g(-\epsilon^2)-g(\epsilon^2)\leq M\epsilon^2$. 
 
 To prove that \eqref{eq:Sx4_ineq} holds when $D_x=\mathcal{S}_{x4}^{c}$, we note that
$$
\mathcal{S}_{x4}^{c}= \cup_{k=1}^4\{{\bf x} \ | \ \#\pi_{\mathcal{I}_k}({\bf x})\leq C_{0x}n\}.
$$
We may thus apply \eqref{eq:number} to the sets
$$
\hat{D}_{x4}^k:=\{{\bf x} \ | \ |\frac{\pi_{\mathcal{I}_k}({\bf x})}{n}-Q(\frac{a_l^k}{\tilde{\tau^\star}})+Q(\frac{b_r^k}{\tilde{\tau}^\star})|\geq \epsilon \}
$$
where $a_l^k$ and $b_l^k$ are such that $\mathcal{I}_k=(a_l^k,b_l^k)$. By leveraging the properties of the error function $x\mapsto Q(\frac{x}{\tilde{\tau}^\star})$, we deduce   \eqref{eq:one_to_prove} in an analogous way to the case of  $\mathcal{S}_{x4}$.

In the following sections, we will focus on proving \eqref{eq:essential}. This proof hinges on a new Gaussian min-max theorem that will be introduced in the next section. However, the application of this new theorem is not straightforward; it requires some non-standard and technical steps, which will be elaborated on subsequently. We organize the remainder of this section as follows: first, we present the new Gaussian min-max theorem as an independent theorem, as it may be of interest in its own right. Next, we outline the technical steps necessary for its application, before finally applying the theorem to our proof.

\subsection{New Gaussian min-max theorem.}
In this section, we present a novel Gaussian min-max theorem that facilitates the analysis of \eqref{eq:essential}. We believe this new theorem holds independent significance and can serve as an alternative to Bussgang's decomposition for examining the performance of quantized solutions to stochastic optimization problems. 

\noindent{\bf Notations.} We define $\mathbb{R}_{\star}^n$ as the subset of all ${\bf x} \in \mathbb{R}^n$ in which none of the entries are zero.  For $r>0$ and ${\bf x}\in \mathbb{R}_{\star}^n$, we define function $\rho_r:\mathbb{R}_\star^{n}\to [-1,1]$ as:
\begin{equation}
\rho_r({\bf x})=\frac{({\bf x}-r\mathcal{Z}({\bf x}))^T\mathcal{Z}({\bf x})}{\|{\bf x}-r\mathcal{Z}({\bf x})\|\|\mathcal{Z}({\bf x})\|} \label{eq:rho_r}
\end{equation}
where $\mathcal{Z}({\bf x})=L{\rm sign}({\bf x})$.
For $\theta\in[-1,1]$ and $r>0$, we define the set $\mathcal{S}_{x,\theta}^{r}$ as:
\begin{equation}
\mathcal{S}_{x,\theta}^{r}=\{{\bf x}\in \mathbb{R}_\star^{n} \  |  \   \rho_r({\bf x})=\theta \ \text{and \ } \#\pi_{(-Lr,Lr)}({\bf x})=0  \}. \label{eq:set_Sx}
\end{equation}
	\begin{theorem}\label{1bit} For $\theta\in[-1,1]$ and $r>0$,  consider two random Gaussian processes defined on $\mathcal{S}_{\bf d}\times\tilde{\mathcal{S}}_{x,\theta}^{r}\times\mathcal{S}_{\bf u}\times\mathcal{S}_\gamma$,
		where $\mathcal{S}_{d}\subset\mathbb{R}^m$ and $\mathcal{S}_{u}\subset\mathbb{R}^m$ are compact sets of $\mathbb{R}^{m}$ and $\tilde{\mathcal{S}}_{x,\theta}^r$ is a compact set of $\mathcal{S}_{x,\theta}^{r}$ defined in \eqref{eq:set_Sx} and $\mathcal{S}_\gamma$ is either a compact set of $\mathbb{R}$ or is given by the entire line of reals $\mathbb{R}$. We define the following prcesses:
		\begin{equation*}
			Y({\bf d,x,u},\gamma)={{\bf u}^T{\bf G}}({\bf x}-r\mathcal{Z}({\bf x}))+\gamma{\bf u}^T{{\bf G}\mathcal{Z}({\bf x})}+\psi({\bf d,x,u},\gamma) 
		\end{equation*}
		and 
		\begin{align*}
			{Z}({\bf d,x,u},\gamma)=&\|{\bf x}-r \mathcal{Z}({\bf x})\|{\bf g}^T{\bf u}-\|{\bf u}\|{\bf h}^T({\bf x}-r\mathcal{Z}({\bf x}))+\gamma \|\mathcal{Z}({\bf x})\|(\theta{\bf g}+\nu{\bf f})^T{\bf u}-|\gamma|\|{\bf u}\|{\bf h}^T\mathcal{Z}({\bf x})+\nonumber\\&\psi({\bf d,x,u},\gamma),
		\end{align*}
		where all elements of ${\bf G}\in\mathbb{R}^{m\times n}$, ${\bf g},{\bf f}\in\mathbb{R}^m$, ${\bf h}\in\mathbb{R}^n$  independently follow the standard Gaussian distribution, 
		$\nu=\sqrt{1-\theta^2}$.
		and $\psi:({\bf d}, {\bf x},{\bf u},\gamma)\mapsto \psi({\bf d,x,u},\gamma)$ is a real-valued, continuous function that is possibly random but independent of ${\bf G}$. Then, for any $t\in \mathbb{R}$,
		\begin{equation*}
			\mathbb{P}[\min_{\substack{{\bf d}\in\mathcal{S}_{d}\\{\bf x}\in\tilde{\mathcal{S}}_{x,\theta}^r}}\max_{\substack{{\bf u}\in\mathcal{S}_{\bf u}\\\gamma\in\mathcal{S}_\gamma}}{Y}({\bf d,x,u},\gamma)\leq t]\leq4\mathbb{P}[\min_{\substack{{\bf d}\in\mathcal{S}_{d}\\{\bf x}\in\tilde{\mathcal{S}}_{x,\theta}^r}}\max_{\substack{{\bf u}\in\mathcal{S}_{\bf u}\\\gamma\in\mathcal{S}_\gamma}}{Z}({\bf d,x,u},\gamma)\leq t]. 
		\end{equation*}
        \begin{proof}
 The proof is deferred to Appendix \ref{app:proof_cgmt}.
 \end{proof}
	\end{theorem}
	
\subsection{Preparatory framework for proving convergence in \eqref{eq:essential} via Theorem \ref{1bit}} 
Theorem \ref{1bit} is central to establishing the convergence in \eqref{eq:essential}. However, a direct application is not feasible, as the optimization variable \({\bf x}\) does not meet the theorem’s conditions. To address this, in this section, we show that the probability term in \eqref{eq:essential} can be upper bounded by another probability term involving a Gaussian process that satisfies the requirements of Theorem \ref{1bit}.

\subsubsection{Reformulation and technical challenges }
To begin with,  we use the technique of Lagrange multipliers to rewrite $\mathcal{P}({\bf x})$ as:
\begin{align}
 \mathcal{P}({\bf x}) = \min_{{\bf d}\in\mathbb{R}^m} \max_{\boldsymbol{\mu}\in \mathbb{R}^m} \mathcal{V}({\bf x}, {\bf d},\boldsymbol{\mu}), \label{eq:d1}
\end{align}
where 
\begin{align*}
\mathcal{V}({\bf x},{\bf d},\boldsymbol{\mu})&=\frac{1}{n}\|{\bf H}({\bf x}-r\mathcal{Z}({\bf x}))+r{\bf d}-{\bf s}\|^2+\frac{\lambda}{n}\|{\bf x}\|^2+\frac{\boldsymbol{\mu}^{T}}{\sqrt{n}}( {\bf H}\mathcal{Z}({\bf x})-{\bf d})
\end{align*}
with \footnote{The relation in \eqref{eq:d1} holds for any real $r$. However, selecting $r=\epsilon^2$ is important for technical reasons that will be clarified later.} $r=\epsilon^2$. In this setup, the solution in ${\bf d}$ for any ${\bf x}$ is given by $d_T({\bf x})$.
Based on \eqref{eq:d1}, we obtain that for any compact set ${S}_x\subset\mathbb{R}^n$ and compact set $\mathcal{S}_d\subset \mathbb{R}^m$
\begin{equation}
\min_{\substack{{\bf x}\in \mathcal{S}_x\\ d_T({\bf x})\in \mathcal{S}_d}} \mathcal{P}({\bf x})=
\min_{{\bf x}\in \mathcal{S}_x}\min_{{\bf d}\in \mathcal{S}_d} \max_{\boldsymbol{\mu}\in \mathbb{R}^m} \mathcal{V}({\bf x}, {\bf d},\boldsymbol{\mu}) .\label{eq:lagrange}
\end{equation}
The proof of \eqref{eq:essential} amounts thus to showing
\begin{equation}
\mathbb{P} \Big[ \min_{{\bf x}\in \mathcal{S}_x^\circ,\|{\bf x}\|_\infty\leq A}\min_{\substack{{\bf d}\in \mathcal{S}_d^{\epsilon}}} \max_{\boldsymbol{\mu}} \mathcal{V}({\bf x}, {\bf d},\boldsymbol{\mu})\leq \varphi(\tau^\star, \beta^\star) + \gamma \epsilon^6 \Big]\leq \frac{C}{\epsilon^6}\exp(-cn\epsilon^{12}).\label{eq:essential_V}
\end{equation}
Next,  we exploit the following relation
\begin{equation}
\|{\bf x}\|^2= \max_{{\bf u}} {\bf u}^{T}{\bf x}-\frac{\|{\bf u}\|^2}{4} \label{eq:max_u}
\end{equation}
to rewrite $\mathcal{V}({\bf x},{\bf d},\boldsymbol{\mu})$ as:
\begin{align*}
\mathcal{V}({\bf x},{\bf d},\boldsymbol{\mu})&=\max_{{\bf u}} \frac{1}{\sqrt{n}}{\bf u}^{T}{\bf H}({\bf x}-\epsilon^2\mathcal{Z}({\bf x}))+\epsilon^2 \frac{1}{\sqrt{n}}{\bf u}^{T}{\bf d}-\frac{1}{\sqrt{n}}{\bf u}^{T}{\bf s}-\frac{\|{\bf u}\|^2}{4}+\frac{\lambda}{n}\|{\bf x}\|^2  \\
&+\frac{\boldsymbol{\mu}^{T}}{\sqrt{n}}( {\bf H}\mathcal{Z}({\bf x})-{\bf d}).
\end{align*}
Our goal is to apply the Gaussian min-max inequality in Theorem \ref{1bit} to prove \eqref{eq:essential_V}, which in turn implies both \eqref{eq:essential}
 and, as noted earlier \eqref{eq:res}. However, a direct application of Theorem \ref{1bit} is not feasible for two main issues: $(1)$ the variable ${\bf x}$ does not belong to $\mathcal{S}_{x,\theta}^{\epsilon^2}$ for some $\theta\in[-1,1]$ and $\epsilon>0$ (See \eqref{eq:set_Sx} for the definition of set $\mathcal{S}_{x,\theta}^{\epsilon^2}$), and $(2)$ the variable $\boldsymbol{\mu}$ is not aligned with ${\bf u}$. The second issue will be easily handled in section \ref{sub4}. 
However,  the first issue, related to the feasible ${\bf x}$ not satisfying the conditions of Theorem \ref{1bit}, is significantly more complex to address. This will require the introduction of non-standard techniques, which we will develop in the following section.
\subsubsection{Variable transformation} To address the issue related to the non-compliance of the feasible ${\bf x}$, we follow the approach developed in our previous work in \cite{as_arxiv}. More specifically, for any \( {\bf x} \in \mathcal{S}_x^\circ \) such that $\|{\bf x}\|_\infty\leq A$, we prove that there exists \( {\bf y} \) close to \( {\bf x} \) such that \( {\bf y} \) satisfies the requirements of Theorem \ref{1bit}. The closeness of \( {\bf x} \) to \( {\bf y} \) is defined in a way that will be convenient for the continuation of the proof.
More specifically, we prove the following result.
\begin{proposition}
With probability at least $1-C\exp(-cn)$, there exists a constant $\tilde{C}$ such that  for all ${\bf x}\in \mathcal{S}_x^\circ$ and $\|{\bf x}\|_{\infty}\leq A$, there exists ${\bf y}\in \mathbb{R}^{n}$ with $\|{\bf y}\|_{\infty}\leq A$ that satisfies the following properties:
\begin{align}
&|\mathcal{P}({\bf x})-P({\bf y})|\leq \tilde{C}\max(\epsilon^6,\frac{\epsilon^2}{\sqrt{n}}) ,\label{eq:prop1}\\
&\frac{1}{m}\|d_T({\bf x})-d_T({\bf y})\|^2\leq \tilde{C}\epsilon^2 ,\label{eq:prop2}\\
& \frac{1}{n}\|{\bf x}-{\bf y}\|^2\leq \tilde{C}\epsilon^6,\label{eq:prop4}\\
& \#\pi_{(-L\epsilon^2,L\epsilon^2)}({\bf y})=0 . \label{eq:prop3}
\end{align}
\label{prop:prin}
\end{proposition}
\begin{proof}
Denote by $\mathcal{J}$ the indexes in $\{1,\cdots,n\}$ corresponding to the positions of elements in ${\bf x}$ belonging to the interval $[-L\epsilon^2,L\epsilon^2]$. 
Let ${\bf a}_1$ and ${\bf a}_2$ be  ${{\bf H}}^{T}({\bf Hx}-{\bf s})$ and ${\bf x}$ at positions  indexed by $\mathcal{I}$ and zero otherwise. According to Lemma \ref{lem:conc_spectral_norm} in Appendix \ref{app:technical_lemmas}, with probability $1-C\exp(-cn)$, $\|{\bf H}\|\leq 3\max(1,\sqrt{\delta})$. Then, with probability $1-C\exp(-cn)$, there exists a constant $C_a$ such that $\|{\bf a}_1\|\leq C_a\sqrt{n}$. Moreover, we note that $\|{\bf a}_2\|\leq A\sqrt{n}$. Given $\boldsymbol{\sigma} = [\sigma_1, \cdots, \sigma_n]$, where $\sigma_i$ are independent and identically distributed Rademacher random variables, consider ${\bf y}(\boldsymbol{\sigma})$ defined as:
\begin{equation}
{\bf y}(\boldsymbol{\sigma}) = {\bf x} \odot {\bf 1}_{\{{\bf x} \notin [-\epsilon^2 L, \epsilon^2 L]\}} + ({\bf x} + 2\boldsymbol{\sigma}\epsilon^2 L) \odot {\bf 1}_{\{{\bf x} \in [-\epsilon^2 L, \epsilon^2 L]\}},
\label{eq:sigma}
\end{equation}
where the notation ${\bf a} \odot {\bf b}$ denotes the Hadamard (element-wise) product between vectors ${\bf a}$ and ${\bf b}$, and the expressions ${\bf 1}_{\{{\bf x} \in [-\epsilon^2 L, \epsilon^2 L]\}}$ and ${\bf 1}_{\{{\bf x} \notin [-\epsilon^2 L, \epsilon^2 L]\}}$ represent a vector where the indicator function is applied element-wise to each entry of the vector ${\bf x}$.  Specifically, for each entry $[{\bf x}]_i, i=1,\cdots,n$ of ${\bf x}$,  ${\bf 1}_{\{{\bf x}\in [-\epsilon^2L,\epsilon^2L]\}}$ returns $1$ if the condition $[{\bf x}]_i \in [-\epsilon^2 L, \epsilon^2 L]$ is true, and $0$ otherwise. Obviously, for any chosen rademacher sequence $\boldsymbol{\sigma}$, ${\bf y}(\boldsymbol{\sigma})$ satisfies $\#\pi_{(-L\epsilon^2,L\epsilon^2)}({\bf y}(\boldsymbol{\sigma}))=0$. 
It follows from the finite class lemma (Lemma \ref{lem:finite_class_lemma}), discussed in Appendix \ref{app:technical_lemmas}, that:
\begin{align}
\mathbb{E}_{\boldsymbol{\sigma}}\big[\max(\frac{1}{n}({\bf y}(\boldsymbol{\sigma})-{\bf x})^{T}{\bf x}, \frac{1}{n}({\bf y}(\boldsymbol{\sigma})-{\bf x})^{T}{\bf H}^{T}({\bf Hx}-{\bf s}))\big] &= \mathbb{E}_{\boldsymbol{\sigma}}[\max(2\epsilon^2L \sum_{i\in \mathcal{J}}\sigma_i[{\bf a}_1]_i,2\epsilon^2L \sum_{i\in \mathcal{J}}\sigma_i[{\bf a}_1]_i) ]\nonumber\\
&\leq 2L\epsilon^2\log 2 \frac{\max(C_a,A)}{\sqrt{n}}.\label{eq:sol}
\end{align}
Let ${S}_{\boldsymbol{\sigma}}$  be the set of all possible values of the random sequence $\boldsymbol{\sigma}$. Obviously, 
$$
\mathbb{E}_{\boldsymbol{\sigma}}\big[\max(\frac{1}{n}({\bf y}(\boldsymbol{\sigma})-{\bf x})^{T}{\bf x}, \frac{1}{n}({\bf y}(\boldsymbol{\sigma})-{\bf x})^{T}{\bf H}^{T}({\bf Hx}-{\bf s}))\big]\geq \min_{\boldsymbol{\sigma}\in {S}_{\boldsymbol{\sigma}}} \max(\frac{1}{n}({\bf y}(\boldsymbol{\sigma})-{\bf x})^{T}{\bf x}, \frac{1}{n}({\bf y}(\boldsymbol{\sigma})-{\bf x})^{T}{\bf H}^{T}({\bf Hx}-{\bf s})).
$$
Using \eqref{eq:sol}, we deduce that for any ${\bf x}\in \mathcal{S}_x^\circ$, there exists $\boldsymbol{\sigma}({\bf x})\in \mathcal{S}_{\boldsymbol{\sigma}}$ depending on ${\bf x}$ such that:
\begin{equation}
 \max(\frac{1}{n}({\bf y}(\boldsymbol{\sigma}({\bf x}))-{\bf x})^{T}{\bf x}, \frac{1}{n}({\bf y}(\boldsymbol{\sigma}({\bf x}))-{\bf x})^{T}{\bf H}^{T}({\bf Hx}-{\bf s}))\leq 2L\epsilon^2\log2\frac{\max(C_a, A)}{\sqrt{n}} .\label{eq:scalar_product}
\end{equation}
With this, we are now in position to show that ${\bf y}(\boldsymbol{\sigma}({\bf x}))$ satisfies the properties in \eqref{eq:prop1}-\eqref{eq:prop4}. For that, using the property $\|{\bf a}\|^2-\|{\bf b}\|^2=({\bf a}-{\bf b})^{T}({\bf a}+{\bf b})$, we can upper bound $|\mathcal{P}({\bf x})-\mathcal{P}({\bf y})|$ as:
\begin{align}
|\mathcal{P}({\bf x})-\mathcal{P}({\bf y})|&\leq \frac{1}{n}|({\bf x}^{T}-{\bf y}(\boldsymbol{\sigma}({\bf x})))^{T}{\bf H}^{T}({\bf H}{\bf x}+{\bf H}{\bf y}(\boldsymbol{\sigma}({\bf x}))-2{\bf s})|+\frac{\lambda}{n}|({\bf x}-{\bf y}(\boldsymbol{\sigma}({\bf x})))^{T}({\bf x}+{\bf y})|\nonumber\\
&\leq \frac{1}{n}\|{\bf H}({\bf x}-{\bf y}(\boldsymbol{\sigma}({\bf x})))\|^2+\frac{2}{n}|({\bf x}-{\bf y}(\boldsymbol{\sigma}({\bf x})))^{T}{\bf H}^{T}({\bf H}{\bf x}-{\bf s})|+\frac{\lambda}{n}\|{\bf x}-{\bf y}(\boldsymbol{\sigma}({\bf x}))\|^2\nonumber \\
&+\frac{2\lambda}{n}|({\bf x}-{\bf y}(\boldsymbol{\sigma}({\bf x})))^{T}{\bf x}|. \label{eq:cost_app}
\end{align}
It follows from \eqref{eq:sigma} that:
\begin{equation}
\frac{1}{n}\|{\bf x}-{\bf y}(\boldsymbol{\sigma}({\bf x}))\|^2\leq 4\epsilon^6L^2C_u, \label{eq:normx}
\end{equation}
hence, with probability $1-C\exp(-cn)$, 
\begin{equation}
\frac{1}{n}\|{\bf H}({\bf x}-{\bf y}(\boldsymbol{\sigma}({\bf x})))\|^2\leq 36\max(1,\delta)\epsilon^6L^2C_u \label{eq:p1}
\end{equation}
and 
\begin{equation}
\frac{\lambda}{n}\|{\bf x}-{\bf y}(\boldsymbol{\sigma}({\bf x}))\|^2\leq 4\lambda \epsilon^6L^2C_u. \label{eq:p2}
\end{equation}
By plugging \eqref{eq:scalar_product}, \eqref{eq:p1} and \eqref{eq:p2} into \eqref{eq:cost_app}, we show that
$$
|\mathcal{P}({\bf x})-\mathcal{P}({\bf y})|\leq \frac{\tilde{C}_1}{\sqrt{n}}+\tilde{C}_2\epsilon^6$$
where 
\begin{align*}
\tilde{C}_1&=4L(\lambda+1)\epsilon^2\log2 \max(C_a,A),\\
\tilde{C}_2&=36 \max(1,\delta)L^2 C_u+4\lambda L^2 C_u.
\end{align*}
Hence ${\bf y}(\sigma({\bf x}))$ satisfies \eqref{eq:prop1} for any $\tilde{C}\geq 2\max(\tilde{C}_1,\tilde{C}_2)$. Additionally,  from \eqref{eq:normx}, we see that it satisfies \eqref{eq:prop4} for any $\tilde{C}\geq 4L^2C_u$. Finally, by noting that:
$$
\|\mathcal{Z}({\bf x})-\mathcal{Z}({\bf y})\|_{\infty}\leq 2L
$$
and that ${\bf x}$ and ${\bf y}(\boldsymbol{\sigma}({\bf x}))$ differ in at most $C_u n\epsilon^2$ entries, we obtain:
$$
\|\mathcal{Z}({\bf x})-\mathcal{Z}({\bf y})\|^2\leq 4L^2C_u n\epsilon^2
$$
and thus with probability $1-C\exp(-cn)$,
$$
\frac{1}{m}\|d_T({\bf x})-d_T(\boldsymbol{\sigma({\bf x})})\|^2\leq \max(\frac{1}{\delta},9)4L^2C_u \epsilon^2.
$$
By choosing $\tilde{C}\geq\max(\frac{1}{\delta},9)L^2C_u$, we prove that ${\bf y}(\boldsymbol{\sigma}({\bf x}))$ satisfies \eqref{eq:prop2} and \eqref{eq:prop4}. 
Hence, by choosing $\tilde{C}\geq \max(\tilde{C}_1,\tilde{C}_2)$ we ensure that ${\bf y}(\boldsymbol{\sigma}({\bf x}))$ satisfies \eqref{eq:prop1}-\eqref{eq:prop4}. 
\end{proof}
\begin{corollary}
\label{cor:fund}
There exists a constant $\tilde{C}$ such that with probability $1-C\exp(-cn)$, for any ${\bf x}\in \mathcal{S}_x^\circ\cap \{{\bf x}| d_T({\bf x})\in \mathcal{S}_d^{\epsilon}\}$ and $\|{\bf x}\|_\infty\leq A$, there exists ${\bf y}\in \mathcal{S}_y^{\epsilon}\cap \{{\bf y}| d_T({\bf y})\in \tilde{\mathcal{S}}_d^\epsilon\}$
 such that:
 $$
 |\mathcal{P}({\bf x})-\mathcal{P}({\bf y})|\leq \tilde{C}\max(\epsilon^6,\frac{\epsilon^2}{\sqrt{n}})
 $$
 where 
 $\mathcal{S}_y^\epsilon=\mathcal{S}_{y1}\cap\mathcal{S}_{y2}\cap \mathcal{S}_{y3}$ with
 \begin{align}
 \mathcal{S}_{y1}&= \{{\bf y}\in \mathbb{R}^n\  | \ |\frac{1}{\sqrt{n}}\|{\bf y}\|-\sqrt{\delta(\tau^\star)^2-1} |\leq (\sqrt{\tilde{C}}+C_{x2})\epsilon^3\},\nonumber\\
 \mathcal{S}_{y2}&=  \{{\bf y}\in \mathbb{R}^n\  | \ |\frac{1}{n}\|{\bf y}\|_1-\mathbb{E}|X| |\leq (\sqrt{\tilde{C}}+C_{x1})\epsilon^3\},\nonumber\\ 
 \mathcal{S}_{y3}& = \{{\bf y}\in \mathbb{R}^n | \|{\bf y}\|_{\infty}\leq A \ \text{and } \#\{i=1,\cdots,n, \ \  y_i\in [-L\epsilon^2,L\epsilon^2]\}=0\}, \label{eq:S_y3}
 \end{align}
 and $\tilde{\mathcal{S}}_d^\epsilon$ being defined as:
 \begin{equation*}
 \tilde{\mathcal{S}}_d^\epsilon= \{{\bf d}\in \mathbb{R}^m \  | \ (\mathcal{W}_2(\hat{\mu}({\bf d},{\bf s}), \tilde{\mu}_{E,s}^\star)^2)\geq (\sqrt{\chi}-\sqrt{\tilde{C}})^2\epsilon^2 \}.
 \end{equation*}
 \end{corollary}
 \begin{proof}
 It follows from Proposition \ref{prop:prin} that with probability $1-C\exp(-cn)$, for any ${\bf x}\in \mathcal{S}_x^\circ$ and $\|{\bf x}\|_{\infty}\leq A$, there exists ${\bf y}$ that satisfies  \eqref{eq:prop1}-\eqref{eq:prop4}. Obviously, such ${\bf y}$ is in $\mathcal{S}_{y1}\cap\mathcal{S}_{y2}\cap\mathcal{S}_{y3}$. It remains thus to prove that for such ${\bf y}$ 
 $$
 d_{T}({\bf x})\in \mathcal{S}_d^\epsilon \Longrightarrow d_{T}({\bf y})\in \tilde{\mathcal{S}}_d^\epsilon.
 $$
 For that, we use the fact that:
 \begin{align*}
 \mathcal{W}_2(\hat{\mu}(d_T({\bf y}),{\bf s}), \tilde{\mu}_{ES}^\star)&\geq \mathcal{W}_2(\hat{\mu}(d_T({\bf x}),{\bf s}), \tilde{\mu}_{ES}^\star)-\mathcal{W}_2(\hat{\mu}(d_T({\bf x}),{\bf s}), \hat{\mu}(d_T({\bf y}),{\bf s}))\\
 &  \geq\mathcal{W}_2(\hat{\mu}(d_T({\bf x}),{\bf s}), \tilde{\mu}_{ES}^\star)-\frac{1}{\sqrt{m}}\|d_T({\bf x})-d_T({\bf y})\|_2\\
 &\geq \sqrt{\chi}\epsilon-\sqrt{\tilde{C}}\epsilon
 \end{align*}
 which shows that $d_{T}({\bf y})\in \tilde{\mathcal{S}}_d^\epsilon$.
 \end{proof}
Based on the definition of \( \mathcal{S}_y^\epsilon \), we can derive the following result for \( \rho_{\epsilon^2}({\bf y}) \) for all \( {\bf y} \in \mathcal{S}_y^\epsilon \). Recall that the function $\rho_{\epsilon^2}({\bf y})$ is defined in \eqref{eq:rho_r}.
\begin{lemma}
There exists a constant $C_{y}$ such that for all ${\bf y}\in \mathcal{S}_y^\epsilon$,
$|\rho_{\epsilon^2}({\bf y})-\theta^{\epsilon}|\leq C_y\epsilon^3$ where
$$
\theta^{\epsilon}:=\frac{\mathbb{E}[|X|]-\epsilon^2L }{\sqrt{\delta (\tau^\star)^2-1-2\epsilon^2L\mathbb{E}[|X|]}}.
$$
\end{lemma}
 \subsubsection{Reducing the convergence in \eqref{eq:essential} to the Gaussian process framework of Theorem \ref{1bit}}
According to Corollary~\ref{cor:fund}, with probability $1-C\exp(-cn)$, we have
$$
\min_{\substack{{\bf x}\in \mathcal{S}_x^\circ,\|{\bf x}\|_\infty\leq A\\ {d_T}({\bf x})\in \mathcal{S}_d^\epsilon}} \mathcal{P}({\bf x})\geq \min_{\substack{{\bf y}\in \mathcal{S}_y^\epsilon\\ {d_T}({\bf y})\in \tilde{\mathcal{S}}_d^\epsilon}} \mathcal{P}({\bf y})-  \tilde{C}\max(\epsilon^6,\frac{\epsilon^2}{\sqrt{n}})
$$
and consequently, 
\begin{align}
&\mathbb{P}\Big[\min_{\substack{{\bf x}\in \mathcal{S}_x^\circ,\|{\bf x}\|_\infty\leq A\\ {d_T}({\bf x})\in \mathcal{S}_d^\epsilon}} \mathcal{P}({\bf x})\leq \varphi(\tau^\star,\beta^\star)+\gamma\epsilon^6\Big]\nonumber\\\leq &\mathbb{P}\Big[\min_{\substack{{\bf y}\in \mathcal{S}_y^\epsilon\\ {d_T}({\bf y})\in \tilde{\mathcal{S}}_d^\epsilon}} \mathcal{P}({\bf y})\leq \varphi(\tau^\star,\beta^\star)+\gamma\epsilon^6+\tilde{C}\max(\epsilon^6,\frac{\epsilon^2}{\sqrt{n}})\Big]+C\exp(-cn).\label{eq:before}
\end{align}
Theorem \ref{1bit} allows us to handle probability inequalities involving Gaussian processes of the form:
\begin{equation}
\min_{\substack{{\bf y}\in \mathcal{S}_y\\ {d_T}({\bf y})\in {\mathcal{S} }_d\\ \rho_{\epsilon^2}({\bf y})=\theta}} \mathcal{P}({\bf y}) \label{eq:Py}
\end{equation}
where $\mathcal{S}_y$ and $\mathcal{S}_d$ are compact, $\theta$ is a fixed real in $[-1,1]$. In order to relate \eqref{eq:before} to probability inequalities involving \eqref{eq:Py}, we define for $\theta\in(\theta^\epsilon-C_y\epsilon^3, \theta^\epsilon+C_y\epsilon^3)$, the following functions: 
$$
\Upsilon(\theta):=\min_{\substack{ \rho_{\epsilon^2}({\bf y})=\theta\\ d_T({\bf y})\in \tilde{\mathcal{S}}_d^\epsilon\\ {\bf y}\in \mathcal{S}_{y3}}} \mathcal{P}({\bf y}). 
$$
where $\mathcal{S}_{y3}$ is defined in \eqref{eq:S_y3}. 
We note that 
\begin{equation*}
\min_{\substack{{\bf y}\in \mathcal{S}_y^\epsilon\\ {d_T}({\bf y})\in \tilde{\mathcal{S}}_d^\epsilon}} \mathcal{P}({\bf y})\geq \min_{\substack{{\bf y}\in \mathcal{S}_{y3}\\ {d_T}({\bf y})\in \tilde{\mathcal{S}}_d^\epsilon\\\rho_{\epsilon^2}({\bf y})\in (\theta^\epsilon-C_y\epsilon^3,\theta^\epsilon+C_y\epsilon^3)}} \mathcal{P}({\bf y}) \geq\min_{\theta\in (\theta^\epsilon-C_y\epsilon^3,\theta^\epsilon+C_y\epsilon^3)}\Upsilon(\theta) 
\end{equation*}
and hence, 
\begin{equation}
\begin{aligned}
&\mathbb{P}\Big[\min_{\substack{{\bf y}\in \mathcal{S}_y^\epsilon\\ {d_T}({\bf y})\in \tilde{\mathcal{S}}_d^\epsilon}} \mathcal{P}({\bf y})\leq \varphi(\tau^\star,\beta^\star)+\gamma\epsilon^6+\tilde{C}\max(\epsilon^6,\frac{\epsilon^2}{\sqrt{n}})\Big]\\
\leq &\mathbb{P}\Big[\min_{\theta^\epsilon-C_y\epsilon^3\leq\theta\leq \theta^\epsilon+C_y\epsilon^3} \Upsilon(\theta)\leq \varphi(\tau^\star,\beta^\star)+\gamma\epsilon^6+\tilde{C}\max(\epsilon^6,\frac{\epsilon^2}{\sqrt{n}})\Big].
\end{aligned}\label{eq:ineq1}
\end{equation}
To continue, we show that on the event
$$
\mathcal{A}:=\Big\{\min_{\theta^\epsilon-C_y\epsilon^3\leq\theta\leq \theta^\epsilon+C_y\epsilon^3} \Upsilon(\theta)\leq \varphi(\tau^\star,\beta^\star)+\gamma\epsilon^6+\tilde{C}\max(\epsilon^6,\frac{\epsilon^2}{\sqrt{n}})\Big\}
$$
there exists a constant $C_\theta$ such that with probability $1-\frac{C}{\epsilon^3}\exp(-cn\epsilon^6)$,  for any $\theta\in (\theta^\epsilon-C_y\epsilon^3, \theta^\epsilon+C_y\epsilon^3)$,
\begin{equation}
|\Upsilon(\hat{\theta})-{\Upsilon}(\theta)|\leq C_\theta|\hat{\theta}-\theta|\label{eq:Crho}
\end{equation}
where $\hat{\theta}=\arg\min_{\theta\in(\theta^\epsilon-C_y\epsilon^2,\theta^\epsilon+C_y\epsilon^3)}\Upsilon(\theta)$.

Before proving \eqref{eq:Crho}, let us show how it reduces the proof of \eqref{eq:essential} to that of analyzing probability inequalities in the form of Theorem \ref{1bit}. Indeed, consider $\mathcal{R}:=\{\theta_i\}_{i=1}^{\lceil2C_y\epsilon^{-3}+1\rceil }$ a grid of $(\theta^\epsilon-C_y\epsilon^3, \theta^\epsilon+C_y\epsilon^3)$ such that for any $\theta\in (\theta^\epsilon-C_y\epsilon^3, \theta^\epsilon+C_y\epsilon^3)$, there exists $\theta_i\in \mathcal{R}$ such that $|\theta-\theta_i|\leq \epsilon^6$. Hence,  assuming \eqref{eq:Crho}, we obtain
\begin{equation}
\Upsilon(\hat{\theta})\geq \min_{\theta\in \mathcal{R}} {\Upsilon}(\theta)-C_\theta  \epsilon^6. \label{eq:Sec_ineq}
\end{equation}
Using union bound argument, we obtain using  \eqref{eq:Sec_ineq} that:
\begin{align*}
 &\mathbb{P}\Big[\min_{\theta^\epsilon-C_y\epsilon^3\leq\theta\leq \theta^\epsilon+C_y\epsilon^3} \Upsilon(\theta)\leq \varphi(\tau^\star,\beta^\star)+\gamma\epsilon^6+\tilde{C}\max(\epsilon^6,\frac{\epsilon^2}{\sqrt{n}})\Big]  \\
 \leq& (\lceil2C_y\epsilon^{-3}+1\rceil)\max_{\theta\in(\theta^\epsilon-C_y\epsilon^3, \theta^\epsilon+C_y\epsilon^3)} \mathbb{P}[{\Upsilon}(\theta)\leq \varphi(\tau^\star,\beta^\star)+(\gamma+C_\theta)\epsilon^6+\tilde{C}\max(\epsilon^6,\frac{\epsilon^2}{\sqrt{n}})].
 \end{align*}
Proving the convergence in \eqref{eq:essential} reduces that to showing that for all $\theta\in (\theta^\epsilon-C_y\epsilon^3, \theta^\epsilon+C_y\epsilon^3)$
\begin{equation}
 \mathbb{P}[{\Upsilon}(\theta)\leq \varphi(\tau^\star,\beta^\star)+(\gamma+C_\theta)\epsilon^6+\tilde{C}\max(\epsilon^6,\frac{\epsilon^2}{\sqrt{n}})]\leq \frac{C}{\epsilon^6}\exp(-cn\epsilon^{12}).  \label{eq:Upsilon_theta}
\end{equation}
We observe that the optimization problem in $\Upsilon(\theta)$ meets the criteria outlined in Theorem \ref{1bit}, thereby making the associated probability inequality applicable. Although we will apply this theorem later, our immediate goal is to demonstrate equation \eqref{eq:Crho}. As seen previously, establishing \eqref{eq:Crho} is crucial, as it simplifies reduces the control of \eqref{eq:ineq1}  to the analysis of \eqref{eq:Upsilon_theta}. To prove \eqref{eq:Crho}, we will first establish the following intermediate result.
 \begin{proposition}
For $\epsilon>0$ chosen sufficiently small, let $\theta_1,\theta_2\in(\theta^{\epsilon}-C_y\epsilon^3, \theta^{\epsilon}+C_y\epsilon^3)$. Then, if there exists ${\bf y}_1\in \mathcal{S}_{y4}(C_0)$ such that $\rho_{\epsilon^2}({\bf y}_1)=\theta_1$, then there exists ${\bf y}_2$ with $\|{\bf y}_2\|_{\infty}\leq A$ such that:
 \begin{align}
 &\rho_{\epsilon^2}({\bf y}_2)=\theta_2, \label{eq:rho}\\
 &\mathcal{Z}({\bf y}_2)=\mathcal{Z}({\bf y}_1), \label{eq:second}\\
 &\|{\bf y}_2\|_1=\|{\bf y}_1\|_1, \label{eq:yy}\\
 &\frac{1}{\sqrt{n}}\|{\bf y}_2-{\bf y}_1\|\leq \hat{C}_y|\theta_1-\theta_2|,\label{eq:third}\\
 &{\bf y}_1{\bf 1}_{{\bf y}_1\in (-\frac{A}{10},\frac{A}{10})}={\bf y}_2{\bf 1}_{{\bf y}_2\in (-\frac{A}{10},\frac{A}{10})},\label{eq:fourth}
 \end{align}
 where $\hat{C}_y$ is some constant independent of $\theta_1$, $\theta_2$, ${\bf y}_1$ and ${\bf y}_2$, and ${\mathcal{S}_{y4}}(C_0)$ is defined as  $$\mathcal{S}_{y4}(C_0)=\{{\bf y}\ | \ \forall\ k=1,...,4,\ \#\pi_{\mathcal{I}_k}({\bf y})\geq C_0 n\}.$$
 \label{prop:loes}
 \end{proposition}
 \begin{proof}
 Without loss of generality, assume that $\theta_2\geq \theta_1$ \footnote{If $\theta_2\leq \theta_1$ we apply the same approach with ${\bf y}({\bf a}):={\bf y}_1\odot{\bf 1}_{\{( \cup_{k=1}^4\mathcal{J}_k)^c\}}+ ({\bf y}_1-a\epsilon^3)\odot{\bf 1}_{\{  \mathcal{J}_1\cup \mathcal{J}_3\}}+({\bf y}_1+a\epsilon^3)\odot{\bf 1}_{\{ \mathcal{J}_2\cup \mathcal{J}_4\}}$}. Let ${\bf y}_1\in \mathcal{S}_{y4}(C_0)$ such that $\rho_{\epsilon^2}({\bf y}_1)=\theta_1$. Then, for ${\bf y}_2$ satisfying \eqref{eq:yy} and \eqref{eq:rho} exists, we must have:
 \begin{align}
 &\theta_2^2L^2(\|{\bf y}_2\|^2-2\epsilon^2L\|{\bf y}_2\|_1+\epsilon^4L^2n )n =(L\|{\bf y}_2\|_1-\epsilon^2L^2n)^2\nonumber\\
 \Longrightarrow &(\theta_2^2-\theta_1^2)L^2(\|{\bf y}_2\|^2-2\epsilon^2L\|{\bf y}_2\|_1+\epsilon^4L^2n )n+\theta_1^2L^2n(\|{\bf y}_2\|^2-\|{\bf y}_1\|^2)=0\nonumber\\
 \Longrightarrow& \|{\bf y}_2\|^2-\|{\bf y}_1\|^2+\frac{\theta_2^2-\theta_1^2}{n\theta_1^2\theta_2^2L^2}(L\|{\bf y}_1\|_1-n\epsilon^2L^2)^2=0.\label{eq:last}
 \end{align}
 Equivalently, we can easily check that if ${\bf y}_2$ satisfies \eqref{eq:yy} and \eqref{eq:last}, then $\rho_{\epsilon^2}({\bf y}_2)=\theta_2$. Since we assumed $\theta_2\geq \theta_1$, then we must look for an ${\bf y}_2$ for which $\|{\bf y}_2\|\leq \|{\bf y}_1\|$. Let $\mathcal{J}_1,\mathcal{J}_2, \mathcal{J}_3$ and $\mathcal{J}_4$ indexes of $\frac{C_0n}{2}$ elements of ${\bf y}_1$ belonging respectively to $ \mathcal{I}_k, k=1,\cdots, 4$. These indexes exist since $\mathcal{\bf y}_1\in \mathcal{S}_{y4}(C_0)$.  For $a\geq 0$, define vector ${\bf y}(a)$ as:
 \begin{equation}
 {\bf y}(a)={\bf y}_1\odot{\bf 1}_{\{( \cup_{k=1}^4\mathcal{J}_k)^c\}}+ ({\bf y}_1+a\epsilon^3)\odot{\bf 1}_{\{  \mathcal{J}_1\cup \mathcal{J}_3\}}+({\bf y}_1-a\epsilon^3)\odot{\bf 1}_{\{ \mathcal{J}_2\cup \mathcal{J}_4\}} \label{eq:ya}
 \end{equation}
 where \(\odot\) denotes the Hadamard product (element-wise multiplication). For any set of indices \(\mathcal{J}\), the complement \((\mathcal{J})^c := \{1, \dots, n\} \setminus \mathcal{J}\) represents all indices not in \(\mathcal{J}\). The indicator vector \(\mathbf{1}_{\mathcal{J}}\) returns 1 for indices in \(\mathcal{J}\) and 0 elsewhere.
 Basically, this transformation consists of adding $a\epsilon^3$ to $\frac{C_0n}{2}$ elements of ${\bf y}_1$ lying in $\mathcal{I}_1$ and $\frac{C_0n}{2}$ elements lying in $\mathcal{I}_3$ and subtracting $a\epsilon^3$ from $\frac{C_0n}{2}$ elements of ${\bf y}_1$ that are in $\mathcal{I}_2$ and $\frac{C_0n}{2}$ elements in $\mathcal{I}_4$. The objective is to show that there exists a bounded constant $a$ such that \eqref{eq:last} holds true. Before continuing it is important to note that  $\|{\bf y}(a)\|_1=\|{\bf y}_1\|_1$. Moreover, provided that $\epsilon$ is sufficiently small, $\|{\bf y}(a)\|_{\infty}\leq A$, and satisfies \eqref{eq:second} and \eqref{eq:fourth}. 
 Using \eqref{eq:ya}, for ${\bf y}_2:={\bf y}(a)$ to satisfy \eqref{eq:last}, we must choose ${a}$ such that:
 \begin{equation}
 a^2\epsilon^62C_0n-2a\epsilon^3S + \frac{\theta_2^2-\theta_1^2}{n\theta_1^2\theta_2^2L^2}(L\|{\bf y}_1\|_1-n\epsilon^2L^2)^2=0 \label{eq:sola}
 \end{equation}
 where
 $$
 S=\sum{\bf y}_1\odot {\bf 1}_{\mathcal{J}_2\cup\mathcal{J}_4}-{\bf y}_1\odot {\bf 1}_{\mathcal{J}_1\cup\mathcal{J}_3}.
 $$
 Define $\Delta_{\epsilon}$ as:
 $$
 \Delta_\epsilon=\epsilon^6S^2-n^2\epsilon^62C_0\frac{\theta_2^2-\theta_1^2}{\theta_1^2\theta_2^2L^2}(L\frac{1}{n}\|{\bf y}_1\|_1-L^2\epsilon^2)^2.
 $$
 We note that $S\geq \frac{A}{5}C_0n$. This can be easily seen by noting that for $i\in \mathcal{J}_2$ and $j\in\mathcal{J}_1$ $[{\bf y}_1]_i-[{\bf y}_1]_j\geq \frac{A}{5}$, and similarly for $i\in \mathcal{J}_4$ and $j\in\mathcal{J}_3$ $[{\bf y}_1]_i-[{\bf y}_1]_j\geq \frac{A}{5}$. Since $|\theta_2-\theta_1|\leq 2C_y\epsilon^3$, we deduce that for all $\epsilon$ sufficiently small, $\Delta_\epsilon >0$.  
 Hence \eqref{eq:sola} admits a positive solution $a^\star(\epsilon)$ given by:
 $$
 a^{\star}(\epsilon)=\frac{\epsilon^3S-\sqrt{\Delta_\epsilon}}{\epsilon^62C_0n}.
 $$
 It remains to check that for all $\epsilon$ sufficiently small $\frac{\epsilon^3a(\epsilon)}{\theta_2-\theta_1}$ is bounded by some  $\hat{C}_y$. Inded, if this is true, then obviously: $\|{\bf y}_1-{\bf y}_2\|_{\infty}\leq \hat{C}_y(\theta_2-\theta_1)$ and thus $\frac{1}{\sqrt{n}}\|{\bf y}_1-{\bf y}_2\|\leq \hat{C}(\theta_2-\theta_1)$. Using the relation $|\sqrt{x}-\sqrt{y}|\leq \frac{x-y}{\sqrt{x}}$ for $x\geq y\geq 0$, we obtain:
 $$
\frac{\epsilon^3 a^{\star}(\epsilon)}{\theta_2-\theta_1}\leq \frac{n\epsilon^9\frac{\theta_2+\theta_1}{\theta_1^2\theta_2^2L^2}(\frac{1}{n}\|{\bf y}_1\|_1-L^2\epsilon^2)^2}{\epsilon^{9} S}\leq \frac{4(\theta^\epsilon+C_y\epsilon^3)(2A^2+2\epsilon^4)}{\frac{A}{5}C_0L^2(\theta^{\epsilon}-C_y\epsilon^3)^4}.
 $$
 For all $\epsilon^2<c_1:=\min(1,\frac{\mathbb{E}|X|}{2L}, \frac{\delta(\tau^\star)^2-1}{4L\mathbb{E}|X|})$, $\rho_M:=\frac{\sqrt{2}{\mathbb{E}|X|}}{\sqrt{\delta(\tau^\star)^2-1}}\geq \theta^\epsilon\geq  \rho_L:=\frac{\mathbb{E}|X|}{2\sqrt{\delta(\tau^\star)^2-1}}$.  Furthermore, for all $\epsilon$ satisfying the additional condition $\epsilon^3\leq \min(\frac{\rho_L}{2C_y},1)$, we obtain:
 $$
 \frac{\epsilon^3 a^{\star}(\epsilon)}{\theta_2-\theta_1}a^\star(\epsilon)\leq {C}_y:= \frac{64C_y(\rho_M+C_y)(2A^2+2)}{\frac{A}{5}C_0L^2\rho_L^4}.
 $$
 \end{proof}

\noindent\underline{Proof of \eqref{eq:Crho}.} Let $\hat{\bf y}$ be given by:
$$
\hat{\bf y}\in\arg\min_{\substack{ \rho_{\epsilon^2}({\bf y})=\hat{\theta}\\ d_T({\bf y})\in \tilde{\mathcal{S}}_d^{\epsilon}\\ {\bf y}\in \mathcal{S}_y^3}} \mathcal{P}({\bf y})
$$
Then, it follows from \eqref{eq:Sx4_repr}, that on the event $\mathcal{A}$, with probability $1-\frac{C}{\epsilon^3}\exp(-cn\epsilon^6)$, $\hat{\bf y}\in \mathcal{S}_{y4}(C_0)$ for some constant $C_0$ and  sufficiently small $\epsilon$. 
Using proposition \ref{prop:loes}, we deduce that there exists  ${\bf y}$ with $\rho_{\epsilon^2}({\bf y})=\theta$ and such that $\frac{1}{\sqrt{n}}\|\hat{\bf y}-{\bf y}\|\leq \hat{C}_y|\theta-\hat{\theta}|$ and $\mathcal{Z}(\hat{\bf y})=\mathcal{Z}({\bf y})$. Hence,
$$
\Upsilon(\theta)-{\Upsilon}(\hat{\theta})\leq \mathcal{P}({\bf y})-\mathcal{P}(\hat{\bf y}).
$$
It takes no much effort to check that function ${\bf y}\mapsto \mathcal{P}({\bf y})$ is $\frac{C_L}{\sqrt{n}}$-Lipschitz with probability $1-C\exp(-cn)$ where $C_L$ is some consant. Hence, 
$$
 \mathcal{P}({\bf y})-\mathcal{P}(\hat{\bf y})\leq \frac{C_L}{\sqrt{n}}\|{\bf y}-\hat{\bf y}\|\leq \hat{C}_yC_L|\hat{\theta}-\theta|.
$$
The inequality in \eqref{eq:Crho} is thus established for $C_\theta=\hat{C}_yC_L$.
\subsubsection{Final transformations to pave the way towards applying Theorem \ref{1bit}}\label{sub4}
Let $\theta\in(\theta^\epsilon-C_y\epsilon^3, \theta^\epsilon+C_y\epsilon^3)$. Our aim is to apply Theorem \ref{1bit} to show that
$$
\mathbb{P}[\min_{\substack{{\bf y}\in \mathcal{S}_y^3\\ \rho_{\epsilon^2}({\bf y})=\theta\\ d_T({\bf y})\in \tilde{\mathcal{S}}_d^\epsilon \\
}} \mathcal{P}({\bf y})\leq \varphi(\tau^\star,\beta^\star)+(\gamma+C_\theta)\epsilon^6+\tilde{C}\max(\epsilon^6,\frac{\epsilon^2}{\sqrt{n}})]\leq \frac{C}{\epsilon^6}\exp(-cn\epsilon^{12})
$$
 uniformly in $\theta^\epsilon\in(\theta^\epsilon-C_y\epsilon^3,\theta^\epsilon+C_y\epsilon^3)$.
 By tightness, for any $\tilde{\epsilon}>0$ there exists  a constant $K_{\tilde{\epsilon}}$ such that:
$$
\mathbb{P}[\|{\bf H}\|\leq K_{\tilde{\epsilon}} ]\geq 1-\tilde{\epsilon}.
$$
On the event $\|{\bf H}\|\leq K_{\tilde{\epsilon}}$, we thus obtain:
$$
\forall {\bf y} \ \text{with } \|{\bf y}\|_{\infty}\leq A,  \ \ \|d_T({\bf y})\|\leq LK_{\tilde{\epsilon}}\sqrt{n}.
$$
Hence,
\begin{equation}
\begin{aligned}
&\mathbb{P} \Big[ \min_{\substack{{\bf y}\in \mathcal{S}_y^3\\  d_T({\bf y})\in \tilde{\mathcal{S}}_d^\epsilon\\ \rho_{\epsilon^2}({\bf y})=\theta}}  \mathcal{P}({\bf y})\leq \varphi(\tau^\star,\beta^\star)+(\gamma+C_\theta)\epsilon^6+\tilde{C}\max(\epsilon^6,\frac{\epsilon^2}{\sqrt{n}}) \Big] \nonumber\\
&\leq \mathbb{P} \Big[ \min_{\substack{{\bf y}\in \mathcal{S}_y^3\\  d_T({\bf y})\in \tilde{\mathcal{S}}_d^\epsilon\\ \rho_{\epsilon^2}({\bf y})=\theta\\ \|d_T({\bf y})\|\leq LK_{\tilde{\epsilon}}\sqrt{n}}}  \mathcal{P}({\bf y})\leq \varphi(\tau^\star,\beta^\star)+(\gamma+C_\rho)\epsilon^6+\tilde{C}\max(\epsilon^6,\frac{\epsilon^2}{\sqrt{n}}) \Big]+\tilde{\epsilon}.
\end{aligned}
\end{equation}
From this point onward, we focus on proving the following statement\footnote{Since \(\tilde{\epsilon}\) is arbitrarily small, proving \eqref{eq:tobe_proven} will imply \eqref{eq:essential_V}.}
\begin{equation}
\mathbb{P}\Big[\min_{\substack{{\bf y}\in \mathcal{S}_y^3\\  d_T({\bf y})\in \tilde{\mathcal{S}}_d^\epsilon\\ \|{\bf d}_T({\bf y})\|\leq LK_{\tilde{\epsilon}\sqrt{n}}\\ \rho_{\epsilon^2}({\bf y})=\theta}}  \mathcal{P}({\bf y}) \leq \varphi(\tau^\star,\beta^\star)+(\gamma+C_\theta)\epsilon^6+\tilde{C}\max(\epsilon^6,\frac{\epsilon^2}{\sqrt{n}}) \Big]\leq \frac{C}{\epsilon^6}\exp(-cn\epsilon^{12}). \label{eq:tobe_proven}
\end{equation}
Using the technique of Lagrange multiplier as shown in \eqref{eq:lagrange} along with the relation in \eqref{eq:max_u}, showing \eqref{eq:tobe_proven} amounts also to proving:
\begin{align}
\mathbb{P} \Big[\min_{\substack{{\bf y}\in\mathcal{S}_y^3\\  \rho_{\epsilon^2}({\bf y})=\theta}} \min_{\substack{ {\bf d}\in \tilde{\mathcal{S}}_d^\epsilon\\ \|{\bf d}\|\leq LK_{\tilde{\epsilon}}\sqrt{n}}}  \max_{\boldsymbol{\mu}} \mathcal{V}({\bf y}, {\bf d},\boldsymbol{\mu})\leq \varphi(\tau^\star,\beta^\star)+(\gamma+C_\theta)\epsilon^6+\tilde{C}\max(\epsilon^6,\frac{\epsilon^2}{\sqrt{n}}) \Big]\to 0 \label{eq:yepsilon}
\end{align}
where  $\mathcal{V}({\bf y}, {\bf d},\boldsymbol{\mu})$ is given by:
\begin{align*}
\mathcal{V}({\bf y},{\bf d},\boldsymbol{\mu})&=\max_{{\bf u}} \frac{1}{\sqrt{n}}{\bf u}^{T}{\bf H}({\bf y}-\epsilon^2\mathcal{Z}({\bf y}))+\epsilon^2 \frac{1}{\sqrt{n}}{\bf u}^{T}{\bf d}-\frac{1}{\sqrt{n}}{\bf u}^{T}{\bf s}-\frac{\|{\bf u}\|^2}{4}+\frac{\lambda}{n}\|{\bf y}\|^2  \\
&+\frac{\boldsymbol{\mu}^{T}}{\sqrt{n}}( {\bf H}\mathcal{Z}({\bf y})-{\bf d}).
\end{align*} 
At optimum, the solution in ${\bf u}$ denoted by ${\bf u}^{\star}$ is given by:
$$
{\bf u}^{\star}=\frac{1}{\sqrt{n}}{\bf H}({\bf y}-\epsilon^2\mathcal{Z}({\bf y}))+\frac{\epsilon^2}{\sqrt{n}}{\bf d}.
$$
Again, as previously explained, by invoking the tightness argument, for any $\tilde{\epsilon}>0$, there exists $K_{\tilde{\epsilon}}$ such that with probability at least $1-\tilde{\epsilon}$,
$
\|{\bf H}\|\leq K_{\tilde{\epsilon}}
$
Hence, on the event $\{\|{\bf H}\|\leq K_{\tilde{\epsilon}}\}$,
$$
\|{\bf u}^{\star}\|\leq K_{u,\tilde{\epsilon}}:=K_{\tilde{\epsilon}}A +2LK_{\tilde{\epsilon}}\epsilon^2
$$
and thus with probability $1-\tilde{\epsilon}$,
$
\mathcal{V}({\bf y},{\bf d},\boldsymbol{\mu})=\tilde{\mathcal{V}}({\bf y},{\bf d},\boldsymbol{\mu})
$
with
\begin{align}
\tilde{\mathcal{V}}({\bf y},{\bf d},\boldsymbol{\mu})&=\max_{\substack{{\bf u}\\ \|{\bf u}\|\leq K_{u,\tilde{\epsilon}}}} \frac{1}{\sqrt{n}}{\bf u}^{T}{\bf H}({\bf y}-\epsilon^2\mathcal{Z}({\bf y}))+\epsilon^2 \frac{1}{\sqrt{n}}{\bf u}^{T}{\bf d}-\frac{1}{\sqrt{n}}{\bf u}^{T}{\bf s}-\frac{\|{\bf u}\|^2}{4}+\frac{\lambda}{n}\|{\bf y}\|^2\nonumber \\
&+\frac{\boldsymbol{\mu}^{T}}{\sqrt{n}}( {\bf H}\mathcal{Z}({\bf y})-{\bf d}) \label{V_tilde}
\end{align}
With this, we can upper bound the probability term in the left-hand side of \eqref{eq:yepsilon}
as:
\begin{align*}
&\mathbb{P} \Big[\min_{\substack{{\bf y}\in\mathcal{S}_y^3\\  \rho_{\epsilon^2}({\bf y})=\theta}} \min_{\substack{ {\bf d}\in \tilde{\mathcal{S}}_d^\epsilon\\ \|{\bf d}\|\leq LK_{\tilde{\epsilon}}\sqrt{n}}}  \max_{\boldsymbol{\mu}} \mathcal{V}({\bf y}, {\bf d},\boldsymbol{\mu})\leq \varphi(\tau^\star,\beta^\star)+(\gamma+C_\theta)\epsilon^6+\tilde{C}\max(\epsilon^6,\frac{\epsilon^2}{\sqrt{n}}) \Big]\\
&\leq \mathbb{P} \Big[\min_{\substack{{\bf y}\in \mathcal{S}_y^3\\  \rho_{\epsilon^2}({\bf y})=\theta}} \min_{\substack{ {\bf d}\in \tilde{\mathcal{S}}_d^\epsilon\\ \|{\bf d}\|\leq LK_{\tilde{\epsilon}\sqrt{n}}}}  \max_{\boldsymbol{\mu}} \tilde{\mathcal{V}}({\bf y}, {\bf d},\boldsymbol{\mu})\leq \varphi(\tau^\star,\beta^\star)+(\gamma+C_\theta)\epsilon^6+\tilde{C}\max(\epsilon^6,\frac{\epsilon^2}{\sqrt{n}}) \Big]+\tilde{\epsilon}.
\end{align*}
Since $\tilde{\epsilon}$ can be chosen arbitrarily small, the proof of \eqref{eq:yepsilon} reduces to showing:
\begin{equation}
\mathbb{P} \Big[\min_{\substack{{\bf y}\in \mathcal{S}_y^3\\  \rho_\epsilon({\bf y})=\theta}} \min_{\substack{ {\bf d}\in \tilde{\mathcal{S}}_d^\epsilon\\ \|{\bf d}\|\leq LK_{\tilde{\epsilon}}\sqrt{n}}}  \max_{\boldsymbol{\mu}} \tilde{\mathcal{V}}({\bf y}, {\bf d},\boldsymbol{\mu})\leq \varphi(\tau^\star,\beta^\star)+(\gamma+C_\theta)\epsilon^6+\tilde{C}\max(\epsilon^6,\frac{\epsilon^2}{\sqrt{n}}) \Big]\leq \frac{C}{\epsilon^6}\exp(-cn\epsilon^{12}). \label{eq:res2}
\end{equation}
The random process \eqref{V_tilde} does not satisfy the requirements of Theorem \ref{1bit} since $\boldsymbol{\mu}$ is not aligned with ${\bf u}$. To solve this issue, we note that:
$$
\max_{\boldsymbol{\mu}} \tilde{\mathcal{V}}({\bf y},{\bf d},\boldsymbol{\mu})\geq \sup_{\ell\in \mathbb{R}}\hat{\mathcal{V}}({\bf y},{\bf d},\ell)
$$
where
\begin{align*}
\hat{\mathcal{V}}({\bf y},{\bf d},\ell)&=\max_{\substack{{\bf u}\\ \|{\bf u}\|\leq K_{u,\tilde{\epsilon}}}} \frac{1}{\sqrt{n}}{\bf u}^{T}{\bf H}({\bf y}-\epsilon^2\mathcal{Z}({\bf y}))+\epsilon^2 \frac{1}{\sqrt{n}}{\bf u}^{T}{\bf d}-\frac{1}{\sqrt{n}}{\bf u}^{T}{\bf s}-\frac{\|{\bf u}\|^2}{4}+\frac{\lambda}{n}\|{\bf y}\|^2  \\
&+\frac{\ell {\bf u}^{T}}{\sqrt{n}}( {\bf H}\mathcal{Z}({\bf y})-{\bf d})
\end{align*}
and hence:
\begin{align*}
&\mathbb{P} \Big[\min_{\substack{{\bf y}\in \mathcal{S}_y^3\\  \rho_{\epsilon^2}({\bf y})=\theta}} \min_{\substack{ {\bf d}\in \tilde{\mathcal{S}}_d^\epsilon\\ \|{\bf d}\|\leq LK_{\tilde{\epsilon}}\sqrt{n}}}  \max_{\boldsymbol{\mu}} \tilde{\mathcal{V}}({\bf y}, {\bf d},\boldsymbol{\mu})\leq \varphi(\tau^\star,\beta^\star)+(\gamma+C_\theta)\epsilon^6+\tilde{C}\max(\epsilon^6,\frac{\epsilon^2}{\sqrt{n}}) \Big] \\
&\leq \mathbb{P} \Big[\min_{\substack{{\bf y}\in \mathcal{S}_y^3\\  \rho_{\epsilon^2}({\bf y})=\theta}} \min_{\substack{ {\bf d}\in \tilde{\mathcal{S}}_d^\epsilon\\ \|{\bf d}\|\leq LK_{\tilde{\epsilon}}\sqrt{n}}}  \sup_{\ell\in\mathbb{R}} \hat{\mathcal{V}}({\bf y}, {\bf d},\ell)\leq \varphi(\tau^\star,\beta^\star)+(\gamma+C_\theta)\epsilon^6+\tilde{C}\max(\epsilon^6,\frac{\epsilon^2}{\sqrt{n}}) \Big].
\end{align*}
The proof of \eqref{eq:res2} reduces to showing:
\begin{equation}
\mathbb{P} \Big[\min_{\substack{{\bf y}\in \mathcal{S}_y^3\\  \rho_{\epsilon^2}({\bf y})=\theta}} \min_{\substack{ {\bf d}\in \tilde{\mathcal{S}}_d^\epsilon\\ \|{\bf d}\|\leq LK_{\tilde{\epsilon}}\sqrt{n}}}  \sup_{\ell\in\mathbb{R}} \hat{\mathcal{V}}({\bf y}, {\bf d},\ell)\leq \varphi(\tau^\star,\beta^\star)+(\gamma+C_\theta)\epsilon^6+\tilde{C}\max(\epsilon^6,\frac{\epsilon^2}{\sqrt{n}}) \Big]\leq \frac{C}{\epsilon^6}\exp(-cn\epsilon^{12}).\label{eq:res3}
\end{equation}
We note that the Gaussian process in \eqref{eq:res3} satisfies the requirements of Theorem \ref{1bit}. We will thus apply in the next section Theorem \eqref{1bit} to show \eqref{eq:res3}. In view of all earlier developments, this will also imply \eqref{eq:essential_V}.
\subsection{Proof of \eqref{eq:res3} via Theorem \ref{1bit}}
\subsubsection{Preliminaries}
The Gaussian process in the right-hand side of \eqref{eq:res3} satisfies the requirements of Theorem \ref{1bit}. By applying Theorem \ref{1bit}, we thus obtain:
\begin{align}
& \mathbb{P} \Big[\min_{\substack{{\bf y}\in \mathcal{S}_y^3\\  \rho_{\epsilon^2}({\bf y})=\theta}} \min_{\substack{ {\bf d}\in \tilde{\mathcal{S}}_d^\epsilon\\ \|{\bf d}\|\leq LK_{\tilde{\epsilon}}}}  \sup_{\ell\in\mathbb{R}} \hat{\mathcal{V}}({\bf y}, {\bf d},\ell)\leq \varphi(\tau^\star,\beta^\star)+(\gamma+C_\theta)\epsilon^6+\tilde{C}\max(\epsilon^6,\frac{\epsilon^2}{\sqrt{n}}) \Big]\nonumber\\
&\leq \mathbb{P} \Big[\min_{\substack{{\bf y}\in \mathcal{S}_y^3\\  \rho_{\epsilon^2}({\bf y})=\theta}} \min_{\substack{ {\bf d}\in \tilde{\mathcal{S}}_d^\epsilon\\ \|{\bf d}\|\leq LK_{\tilde{\epsilon}}\sqrt{n}}}  \sup_{\ell\in\mathbb{R}} k({\bf y}, {\bf d},\ell)\leq \varphi(\tau^\star,\beta^\star)+(\gamma+C_\theta)\epsilon^6+\tilde{C}\max(\epsilon^6,\frac{\epsilon^2}{\sqrt{n}}) \Big]\label{eq:prob_ineq}
 \end{align}
 where 
 \begin{align}
 k({\bf y}, {\bf d},\ell)=\max_{\substack{{\bf u}\\ \|{\bf u}\|\leq K_{u,\tilde{\epsilon}}}}&(\frac{\|{\bf y}-\epsilon^2\mathcal{Z}({\bf y})\|{\bf g}}{n}+\epsilon^2\frac{\bf d}{\sqrt{n}}-\frac{{\bf s}}{\sqrt{n}})^{T}{\bf u}-\frac{\|{\bf u}\|{\bf h}^{T}({\bf y}-\epsilon^2\mathcal{Z}({\bf y}))}{n}-\frac{\|{\bf u}\|^2}{4}\nonumber\\
 &+\frac{\lambda}{n}\|{\bf y}\|^2+\ell(\frac{\|\mathcal{Z}({\bf y})\|(\theta{\bf g}+\nu {\bf f})}{n}-\frac{{\bf d}}{\sqrt{n}})^{T}{\bf u}-|\ell|\frac{\|{\bf u}\|}{n}{\bf h}^{T}\mathcal{Z}({\bf y})
 \end{align}
 where ${\bf g}, {\bf f}\in \mathbb{R}^m$ and ${\bf h}\in \mathbb{R}^n$ are independent standard Gaussian vectors and $\nu=\sqrt{1-\theta^2}$.
The objective of this section is to analyze the properties of \( k({\bf y}, {\bf d}, \ell) \), which will allow us to prove \eqref{eq:res3}.

Consider $\overline{\bf d}_{\infty}^{\rm AO}$ defined as:
$$
\overline{\bf d}_{\infty}^{\rm AO}:=L(\theta^\star {\bf g}+\nu^\star {\bf f})-L\frac{\sqrt{\delta (\tau^\star)^2-1}}{\delta \tau^\star}\mathbb{E}[|H|]{\bf g}+L\frac{\mathbb{E}[|H|]}{\delta \tau^\star}{\bf s},
$$where $\theta^\star=\frac{\mathbb{E}[|X|]}{\sqrt{\delta(\tau^\star)^2-1}}$.
Let $\hat{\mu}(\overline{\bf d}_{\infty}^{\rm AO},{\bf s})$ be the joint empirical distribution of ${\overline{\bf d}}_{\infty}^{\rm AO}$ and ${\bf s}$. In the following Lemma, we prove that with overwhelming probability, the empirical distribution $\hat{\mu}(\overline{\bf d}_{\infty}^{\rm AO},{\bf s})$ is close to $\tilde{\mu}_{ES}^\star$:
\begin{lemma} 
There exists constants $C$ and $c$ such that for  any $\epsilon>0$, 
$$
\mathbb{P}\Big[\Big(\mathcal{W}_2(\hat{\mu}(\overline{\bf d}_{\infty}^{\rm AO},{\bf s}),\tilde{\mu}_{ES}^\star))^2\geq (\sqrt{\chi}-\sqrt{\tilde{C}})^2\frac{\epsilon^2}{4}\Big]\leq C\exp(-cn\epsilon^4).
$$
\label{lem:wassertein}
\end{lemma}
 \begin{proof}
 The proof relies on Lemma \ref{lem:convergence_empirical_rate} and follows the same approach as the proof of Lemma 11 in \cite{as_arxiv}. Therefore, the details are omitted for brevity.
 \end{proof}
 Recall that the set $\tilde{S}_d^\epsilon$ is defined as:
 $$
 \tilde{S}_d^\epsilon=\{{\bf d}\in \mathbb{R}^m\ | \  \mathcal{W}_2(\hat{\mu}({\bf d},{\bf s}),\tilde{\mu}_{ES}^\star)\geq (\sqrt{\chi}-\sqrt{\tilde{C}})\epsilon  \}.
 $$
Using the triangular inequality:
$$
\mathcal{W}_2(\hat{\mu}({\bf d},{\bf s}),\hat{\mu}(\overline{\bf d}_{\infty}^{\rm AO},{\bf s}))\geq \mathcal{W}_2(\hat{\mu}({\bf d},{\bf s}),\tilde{\mu}_{ES}^\star)-\mathcal{W}_2(\hat{\mu}(\overline{\bf d}_{\infty}^{\rm AO},{\bf s}),\tilde{\mu}_{ES}^\star)
$$
along with the result of Lemma \ref{lem:wassertein}, we conclude that
with probability $1-C\exp(-cn\epsilon^4)$
 $$
 {\bf d}\in \tilde{\mathcal{S}}_d^\epsilon  \Longrightarrow \mathcal{W}_2(\hat{\mu}({\bf d},{\bf s}),\hat{\mu}(\overline{\bf d}_{\infty}^{\rm AO},{\bf s}))\geq (\sqrt{\chi}-\sqrt{\tilde{C}})\frac{\epsilon}{2}\Longrightarrow \frac{1}{\sqrt{m}} \|{\bf d}-\overline{\bf d}_\infty^{\rm AO}\|\geq (\sqrt{\chi}-\sqrt{\tilde{C}})\frac{\epsilon}{2}.
 $$
 Hence, with probability $1-C\exp(-cn\epsilon^4)$, 
 \begin{equation*}
 \tilde{\mathcal{S}}_d^\epsilon\subset \hat{\mathcal{S}}_d^\epsilon:=\{{\bf d}, \frac{1}{\sqrt{m}} \|{\bf d}-\overline{\bf d}_\infty^{\rm AO}\|\geq (\sqrt{\chi}-\sqrt{\tilde{C}})\frac{\epsilon}{2}\}.
 \end{equation*}
Starting from the probability inequality in \eqref{eq:prob_ineq}, we thus obtain:
\begin{align*}
&\mathbb{P} \Big[\min_{\substack{{\bf y}\in \mathcal{S}_y^3\\  \rho_{\epsilon^2}({\bf y})=\theta}} \min_{\substack{ {\bf d}\in \tilde{\mathcal{S}}_d^\epsilon\\ \|{\bf d}\|\leq LK_{\tilde{\epsilon}}\sqrt{n}}}  \sup_{\ell\in\mathbb{R}} k({\bf y}, {\bf d},\ell)\leq \varphi(\tau^\star,\beta^\star)+(\gamma+C_\theta)\epsilon^6+\tilde{C}\max(\epsilon^6,\frac{\epsilon^2}{\sqrt{n}}) \Big]\\
&\leq \mathbb{P} \Big[\min_{\substack{{\bf y}\in \mathcal{S}_y^3\\  \rho_{\epsilon^2}({\bf y})=\theta}} \min_{\substack{ {\bf d}\in \hat{\mathcal{S}}_d^\epsilon\\ \|{\bf d}\|\leq LK_{\tilde{\epsilon}}\sqrt{n}}}  \sup_{\ell\in\mathbb{R}} k({\bf y}, {\bf d},\ell)\leq \varphi(\tau^\star,\beta^\star)+(\gamma+C_\theta)\epsilon^6+\tilde{C}\max(\epsilon^6,\frac{\epsilon^2}{\sqrt{n}}) \Big]+C\exp(-cn\epsilon^4)
\end{align*}
We shift our focus then to proving the probability inequality:
\begin{equation}
\mathbb{P} \Big[\min_{\substack{{\bf y}\in \mathcal{S}_y^3\\  \rho_{\epsilon^2}({\bf y})=\theta}} \min_{\substack{ {\bf d}\in \hat{\mathcal{S}}_d^\epsilon\\ \|{\bf d}\|\leq LK_{\tilde{\epsilon}}\sqrt{n}}}  \sup_{\ell\in\mathbb{R}} k({\bf y}, {\bf d},\ell)\leq \varphi(\tau^\star,\beta^\star)+(\gamma+C_\theta)\epsilon^6+\tilde{C}\max(\epsilon^6,\frac{\epsilon^2}{\sqrt{n}}) \Big]\leq \frac{C}{\epsilon^6}\exp(-cn\epsilon^{12}) \label{eq:final}
\end{equation}
for some constants $C$ and $c$ and for $n\geq \frac{1}{\epsilon^8}$. 
\subsubsection{Methodology of the proof.} To prove \eqref{eq:final}, we proceed into the following steps.

\noindent{\underline {\it Step 1: Characterization of the set containing the minimum in ${\bf d}$}.} For ${\bf y}\in \mathbb{R}^n$, we define the set 
 $\mathcal{S}_d^\circ({\bf y})$ as:
 $$
 \mathcal{S}_d^\circ({\bf  y}):=\{{\bf d}\in \mathbb{R}^m, \|\frac{\|\mathcal{Z}({\bf y})\|(\theta{\bf g}+\nu {\bf f})}{n}-\frac{{\bf d}}{\sqrt{n}}\|\leq \frac{1}{n}{\bf h}^{T}\mathcal{Z}({\bf y}) \}.
 $$
 Then, noting that:
  $$
 {\bf d}\notin \mathcal{S}_d^\circ({\bf  y}) \Longrightarrow  \sup_{\ell\in \mathbb{R}}k({\bf y}, {\bf d},\ell)=\infty 
 $$
 we obtain:
 \begin{equation}
 \min_{\substack{ {\bf d}\in \hat{\mathcal{S}}_d^\epsilon\\ \|{\bf d}\|\leq LK_{\tilde{\epsilon}}\sqrt{n}}}  \sup_{\ell\in\mathbb{R}} k({\bf y}, {\bf d},\ell)= \min_{\substack{ {\bf d}\in \hat{\mathcal{S}}_d^\epsilon\\ \|{\bf d}\|\leq LK_{\tilde{\epsilon}}\sqrt{n}\\ {\bf d}\in \mathcal{S}_d^\circ({\bf y})}}  \sup_{\ell\in\mathbb{R}} k({\bf y}, {\bf d},\ell) .\label{eq:note}
 \end{equation}
 Using this, we show that for each ${\bf y}\in \mathbb{R}^n$, ${\bf d}\mapsto \sup_{\ell}k({\bf y}, {\bf d}, \ell )$ is minimized at ${\bf d}_\epsilon({\bf y})$ defined below in \eqref{eq:depsilon}, and we check that  with probability $1-\frac{C}{\epsilon^6}\exp(-cn\epsilon^{12})$:
 \begin{equation}
 \min_{{\bf d}\in \mathbb{R}^n}  \sup_{\ell\in\mathbb{R}} k({\bf y}, {\bf d},\ell) \geq \varphi(\tau^\star,\beta^\star)-(\gamma+C_r+K) \epsilon^6 \label{eq:optimizer_d}
 \end{equation}
 for some constant $C_r$ positive. 

\noindent{\underline {\it Step 2: Characterization of the set containing the minimum in ${\bf y}$}.}  Let $\overline{\bf y}^{\rm AO}$ be given by:
  $$
  \overline{\bf y}^{\rm AO}:=X({\bf h})
  $$
  where $X$ as defined in  \eqref{eq:XH}, is applied element-wise to the vector ${\bf h}$. 
 By replacing ${\bf d}$ by ${\bf d}_{\epsilon}({\bf y})$, we prove that there exists a constant $C_B$ such that with probability $1-\frac{C}{\epsilon^6}\exp(-cn\epsilon^{12})$, 
 \begin{equation}
 \frac{1}{n}\|{\bf y}-\overline{\bf y}^{\rm AO}\|^2\geq C_B\epsilon^6 \ \Longrightarrow \min_{{\bf d}\in \mathcal{S}_d^\circ({\bf y})} \sup_{\ell} k({\bf y},{\bf d},\ell)\geq \varphi(\tau^\star,\beta^\star) +(\gamma+C_\theta+\tilde{C})\epsilon^6 .\label{eq:center_yao}
 \end{equation}
 Denoting by $\mathcal{B}_\epsilon(\overline{\bf y}^{\rm AO})=\{{\bf y} \  | \  \|{\bf y}\|_{\infty}\leq A \ \ \text{and } \ \frac{1}{n}\|{\bf y}-\overline{\bf y}^{\rm AO}\|^2\leq C_B\epsilon^6\}$, we thus conclude that for $n\geq \frac{1}{\epsilon^8}$, we obtain:
 \begin{align}
 &\mathbb{P} \Big[\min_{\substack{{\bf y}\in \mathcal{S}_y^3\\  \rho_{\epsilon^2}({\bf y})=\theta}} \min_{\substack{ {\bf d}\in \hat{\mathcal{S}}_d^\epsilon\\ \|{\bf d}\|\leq LK_{\tilde{\epsilon}}\sqrt{n}}}  \sup_{\ell\in\mathbb{R}} k({\bf y}, {\bf d},\ell)\leq \varphi(\tau^\star,\beta^\star)+(\gamma+C_\theta)\epsilon^6+\tilde{C}\max(\epsilon^6,\frac{\epsilon^2}{\sqrt{n}}) \Big]\nonumber\\
 &\leq  \mathbb{P} \Big[\min_{\substack{{\bf y}\in \mathcal{B}_\epsilon(\overline{\bf y}^{\rm AO})\\  \rho_{\epsilon^2}({\bf y})=\theta}} \min_{\substack{ {\bf d}\in \hat{\mathcal{S}}_d^\epsilon\\ \|{\bf d}\|\leq LK_{\tilde{\epsilon}}\sqrt{n}}}  \sup_{\ell\in\mathbb{R}} k({\bf y}, {\bf d},\ell)\leq \varphi(\tau^\star,\beta^\star)+(\gamma+C_\theta)\epsilon^6+\tilde{C}\max(\epsilon^6,\frac{\epsilon^2}{\sqrt{n}}) \Big] +\frac{C}{\epsilon^6}\exp(-cn\epsilon^{12}).\label{eq:last_equation1}
 \end{align}
 \underline{\it Step 3: Deviation argument:} We prove that function $\mathcal{H}:{\bf d}\mapsto \displaystyle{\min_{\substack{\rho_{\epsilon^2}({\bf y})=\theta\\ {\bf y}\in \mathcal{B}(\overline{\bf y}^{\rm AO})}}\sup_{\ell}k({\bf y},{\bf d},\ell)}$ satisfies with probability at least $1-C\exp(-cn\epsilon^{12})$
 \begin{equation}
 \forall {\bf d} \  \ \frac{1}{\sqrt{n}}\|{\bf d}\|\leq LK_{\tilde{\epsilon}},  \ \ \mathcal{H}({\bf d})\geq  \min_{{\bf d}\in \mathbb{R}^n}\displaystyle{\min_{\substack{\rho_{\epsilon^2}({\bf y})=\theta\\ {\bf y}\in \mathcal{B}(\overline{\bf y}^{\rm AO})}}\sup_{\ell}k({\bf y},{\bf d},\ell)} +\frac{\mathcal{G}_e\epsilon^4}{m}\|{\bf d}-{\bf d}_{\epsilon}({\bf y})\|^2,  \ \ \label{eq:strong_convex}
 \end{equation}
 where ${\bf d}_{\epsilon}({\bf y})$ is defined below in \eqref{eq:depsilon} and $\mathcal{G}_e$ is some constant. 
 This directly will allow us to prove that the probability term in \eqref{eq:last_equation1} converges to zero. Indeed, it suffices to  note that there exists a constant $C_d$ such that with probability $1-C\exp(-cn\epsilon^2)$
 $$
\forall {\bf y}\in \mathcal{B}_\epsilon(\overline{\bf y}^{\rm AO}), \ \  \frac{1}{\sqrt{m}}\|{\bf d}_\epsilon({\bf y})-\overline{\bf d}_{\infty}^{\rm AO}\|\leq C_d\epsilon.
 $$
 Hence, starting from \eqref{eq:last_equation1}, we obtain:
 $$
  \mathcal{H}({\bf d})\geq  \min_{{\bf d}\in \mathbb{R}^n}\displaystyle{\min_{\substack{\rho_{\epsilon}({\bf y})=\theta\\ {\bf y}\in \mathcal{B}(\overline{\bf y}^{\rm AO})}}\sup_{\ell}k({\bf y},{\bf d},\ell)} +{\mathcal{G}_e\epsilon^4}\Big(\frac{1}{\sqrt{m}}\|{\bf d}-\overline{\bf d}_{\infty}^{\rm AO}\|-C_d\epsilon\Big)^2
 $$
 which by using \eqref{eq:optimizer_d}, \eqref{eq:strong_convex} implies that with probability $1-\frac{C}{\epsilon^6}\exp(-cn\epsilon^{12})$
 $$
 \mathcal{H}({\bf d})\geq \varphi(\tau^\star,\beta^\star)-(\gamma+{C}_r+K)\epsilon^6+\mathcal{G}_e\epsilon^4(\frac{1}{\sqrt{m}}\|{\bf d}-\overline{\bf d}_\infty^{\rm AO}\|-C_d\epsilon)^2 .
 $$
 Choosing constant $\chi$ sufficiently large such that:
 $$
 \mathcal{G}_e(\frac{\sqrt{\chi}-\sqrt{\tilde{C}}}{2}-C_d)^2> 2\gamma+C_\theta+\tilde{C}+C_r+K
 $$
 we thus obtain:
 $$
 \mathbb{P} \Big[\min_{\substack{{\bf y}\in \mathcal{S}_y^3\\  \rho_{\epsilon^2}({\bf y})=\theta}} \min_{\substack{ {\bf d}\in \hat{\mathcal{S}}_d^\epsilon\\ \|{\bf d}\|\leq LK_{\tilde{\epsilon}}\sqrt{n}}}  \sup_{\ell\in\mathbb{R}} k({\bf y}, {\bf d},\ell)\leq \varphi(\tau^\star,\beta^\star)+(\gamma+C_\theta)\epsilon^6+\tilde{C}\max(\epsilon^6,\frac{\epsilon^2}{\sqrt{n}}) \Big] \leq \frac{C}{\epsilon^6}\exp(-cn\epsilon^{12}).
 $$
 which shows the desired. In the sequel, we shall provide the details for each step. 
 \subsubsection{Elaboration on proof steps.}

\paragraph{{\underline{Step 1. Proof of \eqref{eq:optimizer_d}}}}
In \eqref{eq:note}, by optimizing over ${\bf u}$ and $\ell$, we derive:
 \begin{align}
  \sup_{\ell\in\mathbb{R}} k({\bf y}, {\bf d},\ell)=&\Big(\Big\|\frac{\|{\bf y}-\epsilon^2\mathcal{Z}({\bf y})\|{\bf g}}{n}+\frac{\epsilon^2}{\sqrt{n}}{\bf d}-\frac{{\bf s}}{\sqrt{n}}\Big\|-\frac{{\bf h}^{T}({\bf y}-\epsilon^2\mathcal{Z}({\bf y}))}{n}\Big)_{+}^2+\frac{\lambda}{n}\|{\bf y}\|^2 \label{eq:k}\\
  {\rm s.t.}&\ \left\|\frac{\|\mathcal{Z}({\bf y})\|(\theta {\bf g}+\nu{\bf f})}{n}-\frac{\bf d}{\sqrt{n}}\right\|\leq\frac{{\bf h}^T\mathcal{Z}({\bf y})}{n}.\nonumber
  \end{align}
  For ${\bf y}\in \mathbb{R}^n$, define:
\begin{align*}
{\bf c}_\epsilon({\bf y})&=\|{\bf y}-\epsilon^2\mathcal{Z}({\bf y})\|\frac{{\bf g}}{ \sqrt{n}}-{\bf s},\\
{\bf e}_\epsilon({\bf y}) &= -\frac{\|\mathcal{Z}({\bf y})\|}{\sqrt{n}}(\theta {\bf g}+\nu {\bf f}).
\end{align*}
With probability $1-C\exp(-cn)$, for all $   {\bf y}$ such that $\|{\bf y}\|_{\infty}\leq A$ and all $\epsilon$ sufficiently small, it holds:
\begin{equation}
\frac{1}{\sqrt{n}}\|{\bf c}_{\epsilon}({\bf y})-\epsilon^2{\bf e}_\epsilon({\bf y})\|\geq \epsilon^2\frac{1}{n}{\bf h}^{T}\mathcal{Z}({\bf y}) .\label{eq:ineq}
\end{equation}
Using Lemma \ref{lem:KKT} for any ${\bf y}\in \mathcal{S}_d^{\circ}({\bf y})$, the minimum in ${\bf d}$ of $\sup_{\ell\in \mathbb{R}}k({\bf y},{\bf d},\ell)$ is:
\begin{equation}
{\bf d}_\epsilon({\bf y})=-{\bf e}_\epsilon({\bf y})-\frac{1}{\sqrt{n}}{\bf h}^{T}\mathcal{Z}({\bf y})\frac{{\bf c}_\epsilon({\bf y})-\epsilon^2{\bf e}_\epsilon({\bf y})}{\|{\bf c}_\epsilon({\bf y})-\epsilon^2{\bf e}({\bf y}\|}. \label{eq:depsilon}
\end{equation}
Replacing ${\bf d}$ by ${\bf d}_{\epsilon({\bf y})}$ and using Lemma \ref{lem:KKT}, we have
$$
\min_{{\bf d}\in \mathbb{R}^n}\sup_{\ell\in\mathbb{R}} k({\bf y}, {\bf d},\ell) = \Big(\sqrt{\frac{1}{n}\Big\|{\bf c}_\epsilon({\bf y})-\epsilon^2{\bf e}_\epsilon({\bf y})\Big\|^2}-\frac{1}{n}{\bf h}^{T}{\bf y}\Big)_{+}^2+\frac{\lambda}{n}\|{\bf y}\|^2.
$$
On the event $\mathcal{E}_t$
\begin{align*}
\mathcal{E}_t:=\Big\{|\frac{{\bf g}^{T}{\bf g}}{n}-\delta|\leq \epsilon^6 \Big\}\cap \Big\{|\frac{{\bf f}^{T}{\bf f}}{n}-\delta|\leq \epsilon^6 \Big\}\cap \Big\{|\frac{{\bf g}^{T}{\bf f}}{n}|\leq \epsilon^6\Big\} \cap \Big\{|\frac{{\bf f}^{T}{\bf s}}{n}|\leq \epsilon^6\Big\} \cap \Big\{|\frac{{\bf g}^{T}{\bf s}}{n}|\leq \epsilon^6\Big\} 
\end{align*}
which occurs with probability $1-C\exp(-cn\epsilon^{12})$, for all ${\bf y}$ such that $\rho_{\epsilon^2}({\bf y})=\theta$
\begin{align}
&\frac{1}{n}\Big\|{\bf c}_\epsilon({\bf y})-\epsilon^2{\bf e}_\epsilon({\bf y})\Big\|^2\nonumber\\
=&\Big\|\frac{1}{n}\|{\bf y}-\epsilon^2\mathcal{Z}({\bf y})\|{\bf g}+\epsilon^2\theta{\bf g}\frac{\|\mathcal{Z}({\bf y})\|}{n}-\frac{{\bf s}}{\sqrt{n}}+\epsilon^2\nu{\bf f}\frac{1}{n}\|{\mathcal{Z}}({\bf y})\|\Big\|^2\nonumber\\ 
= &\frac{\delta}{n }\|{\bf y}-\epsilon^2\mathcal{Z}({\bf y})\|^2+2\epsilon^2\theta \frac{\delta}{n} \|\mathcal{Z}({\bf y})\|\|{\bf y}-\epsilon^2\mathcal{Z}({\bf y})\|+\delta +\frac{\epsilon^4}{n} \|\mathcal{Z}({\bf y})\|^2\delta +\varepsilon({\bf y})\nonumber\\
=&\delta +\delta \frac{\|{\bf y}\|^2}{n}+\varepsilon({\bf y}) \label{eq:last_equation}
\end{align}
where with probability $1-C\exp(-cn\epsilon^{12})$
$$
\sup_{\|{\bf y}\|_\infty\leq A}|\varepsilon({\bf y})|\leq C_a\epsilon^6
$$
for some constant $C_a$. In \eqref{eq:last_equation}, we used the fact that:
$$
\theta=\rho_{\epsilon^2}({\bf y}) =\frac{\mathcal{Z}({\bf y})^{T}({\bf y}-\epsilon^2\mathcal{Z}({\bf y}))}{\|\mathcal{Z}({\bf y})\|\|{\bf y}-\epsilon^2\mathcal{Z}({\bf y})\|}.
$$
Since the function \(x \mapsto (x)_{+}^2\) is Lipschitz over compact sets, there exists a constant \(C_r\) such that:
\begin{equation}
\min_{{\bf d} \in \mathbb{R}^n} \sup_{\ell \in \mathbb{R}} k({\bf y}, {\bf d}, \ell) \geq \left( \sqrt{\delta + \delta \frac{\|{\bf y}\|^2}{n}} - \frac{1}{n} {\bf h}^{T} {\bf y} \right)_{+}^2 + \frac{\lambda}{n} \|{\bf y}\|^2 - C_r \epsilon^6. \label{eq:final_test}
\end{equation}
Based on the results \eqref{eq:f1} and \eqref{eq:f2} reviewed in Section \ref{sec:review}, we conclude that the following inequality holds with probability at least \(1 - \frac{C}{\epsilon^6} \exp(-cn\epsilon^{12})\):
\[
\min_{{\bf d} \in \mathbb{R}^n} \sup_{\ell \in \mathbb{R}} k({\bf y}, {\bf d}, \ell) \geq \varphi(\tau^\star, \beta^\star) - (C_r + \gamma + K) \epsilon^6.
\]
\paragraph{\underline{Step 2: Proof of \eqref{eq:center_yao}}}
Starting from \eqref{eq:final_test} and utilizing \eqref{eq:useful}, we can demonstrate that there exists a constant \(C_B\) such that, for all sufficiently small \(\epsilon\), the following holds with probability at least \(1 - C \exp(-cn\epsilon^{12})\) for all \({\bf y}\) satisfying \(\|{\bf y}\|_{\infty} \leq A\):

\[
\frac{1}{n} \|{\bf y} - \overline{\bf y}^{\rm AO}\|^2 \geq C_B \epsilon^6 \implies \min_{\|{\bf y}\|_{\infty} \leq A} \left( \sqrt{\delta + \delta \frac{\|{\bf y}\|^2}{n}} - \frac{1}{n} {\bf h}^{T} {\bf y} \right)_{+}^2 + \frac{\lambda}{n} \|{\bf y}\|^2 \geq \varphi(\tau^\star, \beta^\star) + (\gamma + C_\theta + \tilde{C} + C_r) \epsilon^6.
\]

Consequently, we obtain:
\begin{equation*}
\frac{1}{n} \|{\bf y} - \overline{\bf y}^{\rm AO}\|^2 \geq C_B \epsilon^6 \implies \min_{{\bf d} \in \mathbb{R}^n} \sup_{\ell \in \mathbb{R}} k({\bf y}, {\bf d}, \ell) \geq \varphi(\tau^\star, \beta^\star) + (\gamma + C_\theta + \tilde{C}) \epsilon^6. 
\end{equation*}
\paragraph{\underline{Step 3: Proof of \eqref{eq:strong_convex}}}
Using Lemma \ref{lem:KKT}, for any ${\bf d}\in \mathcal{S}_d^\circ({\bf y})$, we obtain the following inequality:
\begin{align*}
\|\frac{\|{\bf y}-\epsilon^2\mathcal{Z}({\bf y})\|{\bf g}}{n}+\frac{\epsilon^2{\bf d}}{\sqrt{n}}-\frac{{\bf s}}{\sqrt{n}}\|&\geq \sqrt{\Big(\frac{1}{\sqrt{n}}\|{\bf c}_\epsilon({\bf y})-\epsilon^2{\bf e}_{\epsilon}({\bf y})\|-\frac{\epsilon^2}{n}{\bf h}^{T}\mathcal{Z}({\bf y})\Big)^2+\frac{\epsilon^4}{n}\|{\bf d}-{\bf d}_\epsilon({\bf y})\|^2}\\
&=m_\epsilon({\bf y})+\frac{\frac{\epsilon^4}{n}\|{\bf d}-{\bf d}_\epsilon({\bf y})\|^2}{m_\epsilon({\bf y})(1+\sqrt{1+\frac{\frac{\epsilon^4}{n}\|{\bf d}-{\bf d}_\epsilon({\bf y})\|^2}{(m_\epsilon({\bf y}))^2}})}
\end{align*}
where $$m_\epsilon({\bf y}):=|\frac{1}{\sqrt{n}}\|{\bf c}_\epsilon({\bf y})-\epsilon^2{\bf e}_{\epsilon}({\bf y})\|-\frac{\epsilon^2}{n}{\bf h}^{T}\mathcal{Z}({\bf y})|.$$
From the inequality in \eqref{eq:ineq}, we know that with probability  $1-C\exp(-cn)$ for  sufficiently small $\epsilon$:
\begin{equation*}
m_\epsilon({\bf y}):=\frac{1}{\sqrt{n}}\|{\bf c}_\epsilon({\bf y})-\epsilon^2{\bf e}_{\epsilon}({\bf y})\|-\frac{\epsilon^2}{n}{\bf h}^{T}\mathcal{Z}({\bf y}).
\end{equation*}
Thus, starting from \eqref{eq:k}, we can conclude that for all ${\bf y}$ such that  such that $\rho_{\epsilon^2}({\bf y})=\theta$, we have:
\begin{align*}
\sup_{\ell} k({\bf y},{\bf d},\ell)&\geq \Big(\frac{1}{\sqrt{n}}\|{\bf c}_\epsilon({\bf y})-\epsilon^2{\bf e}_{\epsilon}({\bf y})\|-\frac{1}{n}{\bf h}^{T}{\bf y}+\frac{\frac{\epsilon^4}{n}\|{\bf d}-{\bf d}_\epsilon({\bf y})\|^2}{m_\epsilon({\bf y})(1+\sqrt{1+\frac{\frac{\epsilon^4}{n}\|{\bf d}-{\bf d}_\epsilon({\bf y})\|^2}{m_\epsilon({\bf y})}})}\Big)_{+}^2 +\frac{\lambda}{n}\|{\bf y}\|^2\\
&= \Big(\sqrt{\delta +\delta \frac{\|{\bf y}\|^2}{n}}-\frac{1}{n}{\bf h}^{T}{\bf y}+\epsilon({\bf y})+\frac{\frac{\epsilon^4}{n}\|{\bf d}-{\bf d}_\epsilon({\bf y})\|^2}{m_\epsilon({\bf y})(1+\sqrt{1+\frac{\frac{\epsilon^4}{n}\|{\bf d}-{\bf d}_\epsilon({\bf y})\|^2}{m_\epsilon({\bf y})}})}\Big)_{+}^2 +\frac{\lambda}{n}\|{\bf y}\|^2
\end{align*}
where with probability $1-C\exp(-cn\epsilon^{12})$, $\sup_{\|{\bf y}\|_{\infty}\leq A}\epsilon({\bf y})\leq C_a\epsilon^{6}$.
Using the concentration inequalities in \eqref{eq:conc1} and \eqref{eq:conc2}, it follows that with probability $1-C\exp(-cn\epsilon^{12})$ for  sufficiently small $\epsilon$:
$$
\forall {\bf y}\in \mathcal{B}_\epsilon(\overline{\bf y}^{\rm AO}) \ \ \ \sqrt{\delta +\delta \frac{\|{\bf y}\|^2}{n}}-\frac{1}{n}{\bf h}^{T}{\bf y}+\epsilon({\bf y})\geq \frac{\beta^\star}{4} .
$$
Thus, we obtain:
\begin{align}
\min_{\substack{\rho_{\epsilon}({\bf y})=\theta\\ {\bf y}\in \mathcal{B}_\epsilon(\overline{\bf y}^{\rm AO}) \\ {\bf y}\in \mathcal{S}_y^3}} \sup_{\ell} k({\bf y},{\bf d},\ell)&\geq \min_{\substack{\rho_{\epsilon}({\bf y})=\theta\\ {\bf y}\in \mathcal{B}_\epsilon(\overline{\bf y}^{\rm AO})\\ {\bf y}\in \mathcal{S}_y^3} }\Big(\sqrt{\delta +\delta \frac{\|{\bf y}\|^2}{n}}-\frac{1}{n}{\bf h}^{T}{\bf y}+\epsilon({\bf y})\Big)_{+}^2+\frac{\lambda}{n}\|{\bf y}\|^2+\frac{\frac{\epsilon^4}{n}\|{\bf d}-{\bf d}_\epsilon({\bf y})\|^2\beta^\star}{2m_\epsilon({\bf y})(1+\sqrt{1+\frac{\frac{\epsilon^4}{n}\|{\bf d}-{\bf d}_\epsilon({\bf y})\|^2}{m_\epsilon({\bf y})}})}\nonumber\\
&\geq \varphi(\tau^\star,\beta^\star)-(C_r+\gamma+K)\epsilon^6+\min_{\substack{\rho_{\epsilon^2}({\bf y})=\theta\\ {\bf y}\in \mathcal{B}_\epsilon(\overline{\bf y}^{\rm AO})\\ {\bf y}\in \mathcal{S}_y^3} }\frac{\frac{\epsilon^4}{n}\|{\bf d}-{\bf d}_\epsilon({\bf y})\|^2\beta^\star}{2m_\epsilon({\bf y})(1+\sqrt{1+\frac{\frac{\epsilon^4}{n}\|{\bf d}-{\bf d}_\epsilon({\bf y})\|^2}{m_\epsilon({\bf y})}})}.
\label{eq:s3}\end{align}
Next, we prove the following Lemma
\begin{lemma} The following result holds true:
$$
\forall {\bf y}\in \mathcal{S}_y^3,  \ \ \frac{1}{n}\|{\bf y}-\overline{\bf y}^{\rm AO}\|^2\leq C_B \epsilon^6 \ \ \Longrightarrow \ \ \frac{1}{n}\|\mathcal{Z}({\bf y})-{\mathcal{Z}}(\overline{\bf y}^{\rm AO})\|^2\leq 4L^2C_B\epsilon^2.
$$
\label{lem:ZY}
\end{lemma}
\begin{proof}
Since ${\bf y}\in \mathcal{S}_y^3$,
\begin{align}
\frac{1}{n}\|\mathcal{Z}({\bf y})-\mathcal{Z}(\overline{\bf y}^{\rm AO})\|^2&=\frac{4L^2}{n}\sum_{i=1}^n {\bf 1}_{\{[\bf y]_i\geq \epsilon^2\}}{\bf 1}_{\{[\overline{\bf y}^{\rm AO}]_i< 0\}}+\frac{4L^2}{n}\sum_{i=1}^n {\bf 1}_{\{[\bf y]_i\geq -\epsilon^2\}}{\bf 1}_{\{[\overline{\bf y}^{\rm AO}]_i> 0\}}. \label{eq:ZY}
\end{align}
On the other hand, we have:
\begin{align*}
C_B\epsilon^6\geq \frac{1}{n}\|{\bf y}-\overline{{\bf y}}^{\rm AO}\|^2&\geq \frac{1}{n}\sum_{i=1}^n |[{\bf y}]_i-[\overline{{\bf y}}^{\rm AO}]_i|^2{\bf 1}_{\{[{\bf y}]_i\geq \epsilon^2\}}{\bf 1}_{\{[\overline{{\bf y}}^{\rm AO}]_i< 0\}}\\
&+\frac{1}{n}\sum_{i=1}^n |[{\bf y}]_i-[\overline{{\bf y}}^{\rm AO}]_i|^2{\bf 1}_{\{[{\bf y}]_i\leq -\epsilon^2\}}{\bf 1}_{\{[\overline{{\bf y}}^{\rm AO}]_i> 0\}}\\
&\geq \epsilon^4\Big(\frac{1}{n}\sum_{i=1}^n {\bf 1}_{\{[{\bf y}]_i\geq \epsilon^2\}}{\bf 1}_{\{[\overline{\bf y}^{\rm AO}]_i< 0\}}+\frac{1}{n}\sum_{i=1}^n {\bf 1}_{\{[{\bf y}]_i\geq -\epsilon^2\}}{\bf 1}_{\{[\overline{\bf y}^{\rm AO}]_i> 0\}}\Big).
\end{align*}
Using \eqref{eq:ZY}, we prove the desired. 
\end{proof}
Based on the concentration inequality in \eqref{eq:conc1}, we deduce
\begin{lemma}
With probability $1-C\exp(-cn\epsilon^2)$ 
$$
|\frac{1}{n}\|\overline{\bf y}^{\rm AO}\|^2-(\delta (\tau^\star)^2-1)|\leq \epsilon.
$$
\label{lem:recall}
\end{lemma}

Based on Lemma \ref{lem:ZY} and \ref{lem:recall}, we can easily see that with probability $1-C\exp(-cn\epsilon^2)$
$$
\sup_{\substack{{\bf y}\in\mathcal{B}_{\epsilon}(\overline{\bf y}^{\rm AO})\\ {\bf y}\in \mathcal{S}_{y}^3}} \frac{1}{\sqrt{m}}\|{\bf d}_{\epsilon}({\bf y})-\overline{\bf d}_{\infty}^{\rm AO}\|\leq C_d \epsilon
$$
for some constant $C_d$, 
and thus:
$$
\frac{1}{\sqrt{m}}\|{\bf d}-{\bf d}_{\epsilon}({\bf y})\|\geq \frac{1}{\sqrt{m}}\|{\bf d}-{\bf d}_{\infty}^{\rm AO}\|-C_d\epsilon.
$$
To continue, we use the fact that for all ${\bf d}$ such that $\|{\bf d}\|\le LK_{\tilde{\epsilon}}\sqrt{n}$, ${\bf y}\in \mathcal{B}(\overline{\bf y}^{\rm AO})$, there exists a constant $\mathcal{M}$ such that with probability $1-C\exp(-cn\epsilon^2)$, 
$$
2m_\epsilon({\bf y})(1+\sqrt{1+\frac{\frac{\epsilon^4}{n}\|{\bf d}-{\bf d}_\epsilon({\bf y})\|^2}{m_\epsilon({\bf y})}})\geq \mathcal{M}.
$$
In view of \eqref{eq:s3},  this proves that there exists a constant $\mathcal{G}_e$ such that:
$$
\min_{\substack{\rho_{\epsilon^2}({\bf y})=\theta\\ {\bf y}\in \mathcal{B}_\epsilon(\overline{\bf y}^{\rm AO}) \\ {\bf y}\in \mathcal{S}_y^3}} \sup_{\ell} k({\bf y},{\bf d},\ell)\geq \varphi(\tau^\star,\beta^\star)-(C_r+\gamma+K)\epsilon^6+\mathcal{G}_e\epsilon^4(\frac{1}{\sqrt{m}}\|{\bf d}-\overline{\bf d}_{\infty}^{\rm AO}\|-C_d\epsilon).
$$

\section{Proof of Theorem \ref{1bit}}
To pave the way for the proof of Theorem \ref{1bit} we shall prove the following result involving two Gaussian processes defined on finite sets. 
\label{app:proof_cgmt}
\begin{theorem} 
\label{th:discrete_sets}
For $\theta\in[-1.1]$ and $r>0$, let $I_x$ be a finite set of vectors in $\mathbb{R}^n$ satisfying:
$$\left\{\begin{array}{ll}
\frac{({\bf x}-r\mathcal{Z}({\bf x}))^{T}\mathcal{Z}({\bf x})}{\|{\bf x}-r\mathcal{Z}({\bf x})\|\|\mathcal{Z}({\bf x})\|}&=\theta\\
\#\pi_{(-Lr,Lr)}({\bf x})&=0
\end{array}\right..
$$
Let  $I_u$ a finite set of vectors in $\mathbb{R}^m$. Let $I_\gamma$ be a finite set of reals. Define the following two Gaussian processes:
\begin{align*}
\tilde{Y}({\bf x},{\bf u},\gamma)&:={\bf u}^{T}{\bf G}({\bf x}-r\mathcal{Z}({\bf x}))+\gamma {\bf u}^{T}{\bf G}\mathcal{Z}({\bf x}) +z\|{\bf u}\|\|{\bf x}-r\mathcal{Z}({\bf x})\|+|\gamma|(\theta z + \nu \tilde{z})\|{\bf u}\|\|\mathcal{Z}({\bf x})\|,\\
\tilde{Z}({\bf x},{\bf u},\gamma)&:=\|{\bf x}-r\mathcal{Z}({\bf x})\| {\bf g}^{T}{\bf u}-\|{\bf u}\|{\bf h}^{T}({\bf x}-r\mathcal{Z}({\bf x}))+\gamma \|\mathcal{Z}({\bf x})\|(\theta {\bf g}+\nu {\bf f})^{T}{\bf u}-|\gamma|\|{\bf u}\|{\bf h}^{T}\mathcal{Z}({\bf x}),
\end{align*}
where all elements of ${\bf G}\in \mathbb{R}^{m\times n}$, $z,\tilde{z}\in \mathbb{R}$, ${\bf g}, {\bf f}\in \mathbb{R}^m$  and ${\bf h}\in \mathbb{R}^n$ are independent standard Gaussian random variables. Furthermore, set $\nu=\sqrt{1-\theta^2}$ and let $I_d$ be a finite of vectors in $\mathbb{R}^m$ and  $\psi$ be a finite function defined on $I_d\times I_x \times I_u\times I_\gamma$. Then, for any $t\in \mathbb{R}$,
$$
\mathbb{P}\Big[\min_{\substack{ {\bf x}\in {I}_x\\ {\bf d}\in I_d}}\max_{\substack{{\bf u}\in I_u\\ \gamma \in I_\gamma}} \tilde{Y}({\bf x},{\bf u},\gamma) +\psi({\bf d},{\bf x},{\bf u},\gamma)\leq t \Big]\leq \mathbb{P}\Big[\min_{\substack{ {\bf x}\in {I}_x\\ {\bf d}\in I_d}}\max_{\substack{{\bf u}\in I_u\\ \gamma \in I_\gamma}} \tilde{Z}({\bf x},{\bf u},\gamma)+\psi({\bf d},{\bf x},{\bf u},\gamma)\leq t\Big] .
$$
\end{theorem}
\begin{proof}
For all ${\bf x}, {\bf x}^{'}\in  I_x $, $({\bf u},\gamma), ({\bf u}^{'},\gamma^{'})\in {I}_u\times I_\gamma$, we obtain:
\begin{align*}
&\mathbb{E}\Big[\tilde{Y}({\bf x},{\bf u},\gamma)\tilde{Y}({\bf x}^{'},{\bf u}^{'},\gamma^{'})\Big] \\
=&{\bf u}^{T}{\bf u}^{'}({\bf x}-r\mathcal{Z}({\bf x}))^{T}({\bf x}^{'}-r\mathcal{Z}({\bf x}^{'})) + \|{\bf u}\|\|{\bf u}^{'}\| \|{\bf x}-r\mathcal{Z}({\bf x})\| \|{\bf x}^{'}-r\mathcal{Z}({\bf x}^{'})\|\\
&+\gamma \gamma^{'}{\bf u}^{T}{\bf u}^{'}\mathcal{Z}({\bf x})^{T}\mathcal{Z}({\bf x}^{'})+|\gamma\gamma^{'}| \|{\bf u}\|\|{\bf u}^{'}\| \|\mathcal{Z}({\bf x})\|\|\mathcal{Z}({\bf x}^{'})\| \\
&+\gamma^{'}{\bf u}^{T}{\bf u}^{'}({\bf x}-r\mathcal{Z}({\bf x})^{T}\mathcal{Z}({\bf x}^{'})+\theta \|{\bf u}\|\|{\bf u}^{'}\| |\gamma^{'}| \|{\bf x}-r\mathcal{Z}({\bf x})\| \|\mathcal{Z}({\bf x}^{'})\| \\
&+\gamma{\bf u}^{T}{\bf u}^{'}({\bf x}^{'}-r\mathcal{Z}({\bf x}^{'})^{T}\mathcal{Z}({\bf x})+\theta \|{\bf u}\|\|{\bf u}^{'}\||\gamma| \|{\bf x}^{'}-r\mathcal{Z}({\bf x}^{'})\| \|\mathcal{Z}({\bf x})\|
\end{align*}
and 
\begin{align*}
&\mathbb{E}\Big[\tilde{Z}({\bf x},{\bf u},\gamma)\tilde{Z}({\bf x}^{'},{\bf u}^{'},\gamma^{'})\Big] \\
=&{\bf u}^{T}{\bf u}^{'}\|{\bf x}-r\mathcal{Z}({\bf x})\|\|{\bf x}^{'}-r\mathcal{Z}({\bf x}^{'})\| + \|{\bf u}\|\|{\bf u}^{'}\| ({\bf x}-r\mathcal{Z}({\bf x}))^{T} ({\bf x}^{'}-r\mathcal{Z}({\bf x}^{'})) \\
&+\gamma \gamma^{'}{\bf u}^{T}{\bf u}^{'}\|\mathcal{Z}({\bf x})\|\|\mathcal{Z}({\bf x}^{'})\|+|\gamma\gamma^{'}| \|{\bf u}\|\|{\bf u}^{'}\| \mathcal{Z}({\bf x})^{T}\mathcal{Z}({\bf x}^{'}) \\
&+\theta\gamma^{'}{\bf u}^{T}{\bf u}^{'}\|{\bf x}-r\mathcal{Z}({\bf x})\|\|\mathcal{Z}({\bf x}^{'})\|+\|{\bf u}\|\|{\bf u}^{'}\| |\gamma^{'}| ({\bf x}-r\mathcal{Z}({\bf x}))^{T} \mathcal{Z}({\bf x}^{'}) \\
&+\theta\gamma{\bf u}^{T}{\bf u}^{'}\|({\bf x}^{'}-r\mathcal{Z}({\bf x}^{'})\|\|\mathcal{Z}({\bf x})\|+ \|{\bf u}\|\|{\bf u}^{'}\||\gamma| ({\bf x}^{'}-r\mathcal{Z}({\bf x}^{'}))^{T} \mathcal{Z}({\bf x}).
\end{align*}
We thus have:
\begin{align*}
&\mathbb{E}\Big[\tilde{Z}({\bf x},{\bf u},\gamma)\tilde{Z}({\bf x}^{'},{\bf u}^{'},\gamma^{'})\Big]-\mathbb{E}\Big[\tilde{Y}({\bf d},{\bf x},{\bf u},\gamma)\tilde{Y}({\bf d}^{'},{\bf x}^{'},{\bf u}^{'},\gamma^{'})\Big] \\
=&\big({\bf u}^{T}{\bf u}^{'}-\|{\bf u}\|\|{\bf u}^{'}\|\big)\big(\|{\bf x}-r\mathcal{Z}({\bf x})\|\|{\bf x}^{'}-r\mathcal{Z}({\bf x}^{'})\|-({\bf x}-r\mathcal{Z}({\bf x}))^{T}({\bf x}^{'}-r\mathcal{Z}({\bf x}^{'}))\big) \\
&+(\gamma \gamma^{'}{\bf u}^{T}{\bf u}^{'}-\|\gamma {\bf u}\|\|\gamma^{'}{\bf u}^{'}\|)\big(\|\mathcal{Z}({\bf x})\|\|\mathcal{Z}({\bf x}^{'})\|-\mathcal{Z}({\bf x})^{T}\mathcal{Z}({\bf x}^{'})\big)\\
&+\big(\gamma^{'}{\bf u}^{T}{\bf u}^{'}-\| {\bf u}\|\|\gamma^{'}{\bf u}^{'}\|\big)\big(\theta \|{\bf x}-r\mathcal{Z}({\bf x})\|\|\mathcal{Z}({\bf x}^{'})\|-\big({\bf x}-r\mathcal{Z}({\bf x})\big)^{T}\mathcal{Z}({\bf x}^{'})\big)  \\
&+\big(\gamma{\bf u}^{T}{\bf u}^{'}-\|\gamma {\bf u}\|\|{\bf u}^{'}\|\big)\big(\theta \|{\bf x}^{'}-r\mathcal{Z}({\bf x}^{'})\|\|\mathcal{Z}({\bf x})\|-\big({\bf x}^{'}-r\mathcal{Z}({\bf x}^{'})\big)^{T}\mathcal{Z}({\bf x})\big)  .
\end{align*}
Now, using the fact that $\|\mathcal{Z}({\bf x})\|=\|\mathcal{Z}({\bf x}^{'})\|$ together with the following equations:
$$
\theta\|{\bf x}-r\mathcal{Z}({\bf x})\|\|\mathcal{Z}({\bf x})\|= ({\bf x}-r\mathcal{Z}({\bf x}))^{T}\mathcal{Z}({\bf x}) \ \ \text{and }\theta\|{\bf x}^{'}-r\mathcal{Z}({\bf x}^{'})\|\|\mathcal{Z}({\bf x}^{'})\|= ({\bf x}^{'}-r\mathcal{Z}({\bf x}^{'}))^{T}\mathcal{Z}({\bf x}^{'})
$$
we obtain:
\begin{align*}
&\mathbb{E}\Big[\tilde{Z}({\bf x},{\bf u},\gamma)\tilde{Z}({\bf x}^{'},{\bf u}^{'},\gamma^{'})\Big]-\mathbb{E}\Big[\tilde{Y}({\bf x},{\bf u},\gamma)\tilde{Y}({\bf x}^{'},{\bf u}^{'},\gamma^{'})\Big]\\
&=\big({\bf u}^{T}{\bf u}^{'}-\|{\bf u}\|\|{\bf u}^{'}\|\big)\big(\|{\bf x}-r\mathcal{Z}({\bf x})\|\|{\bf x}^{'}-r\mathcal{Z}({\bf x}^{'})\|-({\bf x}-r\mathcal{Z}({\bf x}))^{T}({\bf x}^{'}-r\mathcal{Z}({\bf x}^{'}))\big)  \\
&+(\gamma \gamma^{'}{\bf u}^{T}{\bf u}^{'}-\|\gamma {\bf u}\|\|\gamma^{'}{\bf u}^{'}\|)\big(\|\mathcal{Z}({\bf x})\|\|\mathcal{Z}({\bf x}^{'})\|-\mathcal{Z}({\bf x})^{T}\mathcal{Z}({\bf x}^{'})\big)\\
&+\big(\gamma^{'}{\bf u}^{T}{\bf u}^{'}-\| {\bf u}\|\|\gamma^{'}{\bf u}^{'}\|\big)\big( ({\bf x}-r\mathcal{Z}({\bf x}))\mathcal{Z}({\bf x})-\big({\bf x}-r\mathcal{Z}({\bf x})\big)^{T}\mathcal{Z}({\bf x}^{'})\big) \\
&+\big(\gamma{\bf u}^{T}{\bf u}^{'}-\|\gamma {\bf u}\|\|{\bf u}^{'}\|\big)\big( ({\bf x}^{'}-r\mathcal{Z}({\bf x}^{'}))\mathcal{Z}({\bf x}^{'})-\big({\bf x}^{'}-r\mathcal{Z}({\bf x}^{'})\big)^{T}\mathcal{Z}({\bf x})\big) 
\end{align*}
It is easy to check that when ${\bf x}={\bf x}^{'}$
$$
\mathbb{E}\Big[\tilde{Z}({\bf x},{\bf u},\gamma)\tilde{Z}({\bf x},{\bf u}^{'},\gamma^{'})\Big]-\mathbb{E}\Big[\tilde{Y}({\bf x},{\bf u},\gamma)\tilde{Y}({\bf x},{\bf u}^{'},\gamma^{'})\Big]=0.
$$
For ${\bf x}\neq {\bf x}^{'}$, we can show that:
\begin{equation}
\mathbb{E}\Big[\tilde{Z}({\bf x},{\bf u},\gamma)\tilde{Z}({\bf x}^{'},{\bf u}^{'},\gamma^{'})\Big]-\mathbb{E}\Big[\tilde{Y}({\bf x},{\bf u},\gamma)\tilde{Y}({\bf x}^{'},{\bf u}^{'},\gamma^{'})\Big]\leq 0. \label{eq:neg}
\end{equation}
Indeed, since ${\bf x}$ does not contain any elements in $(-Lr,Lr)$,
$$
{\rm sign}({\bf x}-r\mathcal{Z}({\bf x}))={\rm sign}({\bf x})
$$
and thus:
$$
\mathcal{Z}({\bf x})=\mathcal{Z}({\bf x}-r\mathcal{Z}({\bf x})).
$$
We thus obtain:
$$
({\bf x}-r\mathcal{Z}({\bf x}))\mathcal{Z}({\bf x})-\big({\bf x}-r\mathcal{Z}({\bf x})\big)^{T}\mathcal{Z}({\bf x}^{'})\geq 0 \  \ \text{and }({\bf x}^{'}-r\mathcal{Z}({\bf x}^{'}))\mathcal{Z}({\bf x}^{'})-\big({\bf x}^{'}-r\mathcal{Z}({\bf x}^{'})\big)^{T}\mathcal{Z}({\bf x})\geq 0
$$
thus implying \eqref{eq:neg}.
Then, based on Gordon's inequality in Theorem \ref{lem:gor}, Theorem \ref{th:discrete_sets} is established.\end{proof} We are now in position to show Theorem \ref{1bit}. We first consider the Gaussian processes:
\begin{align*}
\hat{Y}({\bf d},{\bf x},{\bf u},\gamma)&={\bf u}^{T}{\bf G}({\bf x}-r\mathcal{Z}({\bf x}))+\gamma {\bf u}^{T}{\bf G}\mathcal{Z}({\bf x})+z\|{\bf u}\|\|{\bf x}-r\mathcal{Z}({\bf x})\|+|\gamma|(\theta z+\nu \tilde{z})\|{\bf u}\|\|\mathcal{Z}({\bf x})\| \\
&+\psi({\bf d},{\bf x},{\bf u},\gamma), \\
\hat{Z}({\bf d},{\bf x},{\bf u},\gamma)&=\|{\bf x}-r\mathcal{Z}({\bf x})\| {\bf g}^{T}{\bf u}-\|{\bf u}\|{\bf h}^{T}({\bf x}-r\mathcal{Z}({\bf x}))+\gamma \|\mathcal{Z}({\bf x})\|(\theta {\bf g}+\nu {\bf f})^{T}{\bf u}-|\gamma|\|{\bf u}\|{\bf h}^{T}\mathcal{Z}({\bf x})  \\
&+\psi({\bf d},{\bf x},{\bf u},\gamma).
\end{align*}
defined on compact sets.
Particularly, we aim to show for any $t\in\mathbb{R}$,
$$
\mathbb{P}\Big[\min_{\substack{{\bf d}\in \mathcal{S}_d\\ {\bf x}\in \tilde{\mathcal{S}}_{x,\theta}^r}}\max_{\substack{{\bf u}\in \mathcal{S}_u\\ \gamma\in \mathcal{S}_\gamma}}\hat{Y}({\bf d},{\bf x},{\bf u},\gamma)\leq t\Big]\leq \mathbb{P}\Big[\min_{\substack{{\bf d}\in \mathcal{S}_d\\ {\bf x}\in \tilde{\mathcal{S}}_{x,\theta}^r}}\max_{\substack{{\bf u}\in \mathcal{S}_u\\ \gamma\in \mathcal{S}_\gamma}}\hat{Z}({\bf d},{\bf x},{\bf u},\gamma)\leq t\Big],
$$where the feasible sets are defined in Theorem \ref{1bit}.

\noindent{\underline{Case 1. $\mathcal{S}_\gamma$ is compact.}} Let $R=\sup_{\gamma \in \mathcal{S}_\gamma }|\gamma|$. Since $\psi$ is continuous, for any $\epsilon>0$, there exists $\delta_1(\epsilon)>0$ such that for any  ${\bf x},\tilde{\bf x} \in \tilde{\mathcal{S}}_{x,\theta}^r$, ${\bf d},\tilde{\bf d}\in \mathcal{S}_d$, ${\bf u}, \tilde{\bf u}\in \mathcal{S}_u$ and $\gamma, \tilde{\gamma}\in \mathcal{S}_\gamma$ satisfying:
$$
\max(\|{\bf x}-\tilde{\bf x}\|, \|{\bf d}-\tilde{\bf d}\|, \|{\bf u}-\tilde{\bf u}\|, |\gamma-\tilde{\gamma}|)\leq \delta_1(\epsilon) \Longrightarrow |\psi({\bf d},{\bf x},{\bf u},\gamma)-\psi(\tilde{\bf d},\tilde{\bf x},\tilde{\bf u},\tilde{\gamma})|\leq \epsilon.
$$
Similarly,  the processes $\tilde{Y}( {\bf x},{\bf u},\gamma)$ and $\tilde{Z}( {\bf x},{\bf u},\gamma)$ are uniformly continuous on the set $  \tilde{\mathcal{S}}_{x,\theta}^r\times \mathcal{S}_u\times\mathcal{S}_\gamma$. For any $\epsilon>0$, there exists thus $\delta_2(\epsilon)$ such that:
$$
\max(\|{\bf x}-\tilde{\bf x}\|,  \|{\bf u}-\tilde{\bf u}\|, |\gamma-\tilde{\gamma}|)\leq \delta_2(\epsilon) \Longrightarrow |\tilde{Y}( {\bf x},{\bf u},\gamma)-\tilde{Y}( \tilde{\bf x},\tilde{\bf u},\tilde{\gamma})|\leq \epsilon \  \text{ and }\ |\tilde{Z}( {\bf x},{\bf u},\gamma)-\tilde{Z}( \tilde{\bf x},\tilde{\bf u},\tilde{\gamma})|\leq \epsilon.
$$
Given $\epsilon>0$, there exists $K_\epsilon$ such that with probability $1-\epsilon$, 
$$
\max( {\|{\bf G}\|},|z|,|\tilde{z}|, {\|{\bf g}\|}, {\|\tilde{\bf g}\|},{\|\tilde{\bf f}\|},{\|\tilde{\bf h}\|})\leq K_\epsilon
$$
we conclude that there exists a constant $d_\epsilon$ such that with probability $1-\epsilon$, $\delta_2(\epsilon)$ may be chosen greater than a certain $d_\epsilon$. Now let $\delta(\epsilon)=\min(\delta_1(\epsilon), d_\epsilon)$. Then, for all $\tilde{\delta}\in(0,\delta(\epsilon)]$, 
\begin{align}
&\max(\|{\bf x}-\tilde{\bf x}\|, \|{\bf d}-\tilde{\bf d}\|, \|{\bf u}-\tilde{\bf u}\|, |\gamma-\tilde{\gamma}|)\leq \tilde{\delta}\nonumber \\\Longrightarrow &|\hat{Y}({\bf d},{\bf x},{\bf u},\gamma)-\hat{Y}(\tilde{\bf d},\tilde{\bf x},\tilde{\bf u},\tilde{\gamma})|\leq 2\epsilon \text{ and } |\hat{Z}({\bf d},{\bf x},{\bf u},\gamma)-\hat{Z}(\tilde{\bf d},\tilde{\bf x},\tilde{\bf u},\tilde{\gamma})|\leq 2\epsilon  \label{eq:s}
\end{align}
Let $\tilde{\delta}\in (0,\delta(\epsilon))$, and consider $\mathcal{S}_d^\delta, \mathcal{S}_u^\delta$, $\mathcal{S}_\gamma^\delta$ and $\mathcal{S}_x^\delta$, $\tilde{\delta}$-nets of $\mathcal{S}_d$, $\mathcal{S}_u$, $\mathcal{S}_\gamma$ and $\tilde{\mathcal{S}}_{x,\theta}^r$. Then, based on \eqref{eq:s}, we obtain, with probability $1-\epsilon$
$$
\Big|\min_{\substack{{\bf d}\in \mathcal{S}_d\\ {\bf x}\in \tilde{\mathcal{S}}_{x,\theta}^r}}\max_{\substack{{\bf u}\in \mathcal{S}_u\\ \gamma\in \mathcal{S}_\gamma}}\hat{Y}({\bf d},{\bf x},{\bf u},\gamma) - \min_{\substack{{\bf d}\in \mathcal{S}_d^\delta\\ {\bf x}\in {\mathcal{S}}_{x}^\delta}}\max_{\substack{{\bf u}\in \mathcal{S}_u^\delta\\ \gamma\in \mathcal{S}_\gamma^\delta}}\hat{Y}({\bf d},{\bf x},{\bf u},\gamma)\Big|\leq 2\epsilon \ \ \text{and }\Big|\min_{\substack{{\bf d}\in \mathcal{S}_d\\ {\bf x}\in \tilde{\mathcal{S}}_{x,\theta}^r}}\max_{\substack{{\bf u}\in \mathcal{S}_u\\ \gamma\in \mathcal{S}_\gamma}}\hat{Z}({\bf d},{\bf x},{\bf u},\gamma) - \min_{\substack{{\bf d}\in \mathcal{S}_d^\delta\\ {\bf x}\in {\mathcal{S}}_{x}^\delta}}\max_{\substack{{\bf u}\in \mathcal{S}_u^\delta\\ \gamma\in \mathcal{S}_\gamma^\delta}}\hat{Z}({\bf d},{\bf x},{\bf u},\gamma)\Big|\leq 2\epsilon.
$$
Furthermore, from Theorem \ref{th:discrete_sets}, we have for any $t\in \mathbb{R}$
$$
\mathbb{P}\Big[\min_{\substack{{\bf d}\in \mathcal{S}_d^\delta\\ {\bf x}\in {\mathcal{S}}_{x}^\delta}}\max_{\substack{{\bf u}\in \mathcal{S}_u^\delta\\ \gamma\in \mathcal{S}_\gamma^\delta}}\hat{Y}({\bf d},{\bf x},{\bf u},\gamma)\leq t\big]\leq \mathbb{P}\Big[\min_{\substack{{\bf d}\in \mathcal{S}_d^\delta\\ {\bf x}\in {\mathcal{S}}_{x}^\delta}}\max_{\substack{{\bf u}\in \mathcal{S}_u^\delta\\ \gamma\in \mathcal{S}_\gamma^\delta}}\hat{Z}({\bf d},{\bf x},{\bf u},\gamma)\leq t\big].
$$
Hence,
\begin{align*}
\mathbb{P}\Big[\min_{\substack{{\bf d}\in \mathcal{S}_d\\ {\bf x}\in \tilde{\mathcal{S}}_{x,\theta}^r}}\max_{\substack{{\bf u}\in \mathcal{S}_u\\ \gamma\in \mathcal{S}_\gamma}}\hat{Y}({\bf d},{\bf x},{\bf u},\gamma)\leq t\Big]&\leq \mathbb{P}\Big[\min_{\substack{{\bf d}\in \mathcal{S}_d^\delta\\ {\bf x}\in {\mathcal{S}}_{x}^\delta}}\max_{\substack{{\bf u}\in \mathcal{S}_u^\delta\\ \gamma\in \mathcal{S}_\gamma^\delta}}\hat{Y}({\bf d},{\bf x},{\bf u},\gamma)\leq t+2\epsilon\big]+\epsilon\\
&\leq \mathbb{P}\Big[\min_{\substack{{\bf d}\in \mathcal{S}_d^\delta\\ {\bf x}\in {\mathcal{S}}_{x}^\delta}}\max_{\substack{{\bf u}\in \mathcal{S}_u^\delta\\ \gamma\in \mathcal{S}_\gamma^\delta}}\hat{Z}({\bf d},{\bf x},{\bf u},\gamma)\leq t+2\epsilon\big]+\epsilon\\
&\leq \mathbb{P}\Big[\min_{\substack{{\bf d}\in \mathcal{S}_d\\ {\bf x}\in \tilde{\mathcal{S}}_{x,\theta}^r}}\max_{\substack{{\bf u}\in \mathcal{S}_u\\ \gamma\in \mathcal{S}_\gamma}}\hat{Z}({\bf d},{\bf x},{\bf u},\gamma)\leq t+4\epsilon\big]+\epsilon.
\end{align*}
By taking $\epsilon$ to zero, we prove thus the desired which is
\begin{align}
\mathbb{P}\Big[\min_{\substack{{\bf d}\in \mathcal{S}_d\\ {\bf x}\in \tilde{\mathcal{S}}_{x,\theta}^r}}\max_{\substack{{\bf u}\in \mathcal{S}_u\\ \gamma\in \mathcal{S}_\gamma}}\hat{Y}({\bf d},{\bf x},{\bf u},\gamma)\leq t\Big]  \leq \mathbb{P}\Big[\min_{\substack{{\bf d}\in \mathcal{S}_d\\ {\bf x}\in \tilde{\mathcal{S}}_{x,\theta}^r}}\max_{\substack{{\bf u}\in \mathcal{S}_u\\ \gamma\in \mathcal{S}_\gamma}}\hat{Z}({\bf d},{\bf x},{\bf u},\gamma)\leq t \big]\label{res_case1}.
\end{align}

\noindent{\underline{Case 2. $\mathcal{S}_\gamma=\mathbb{R}$.}} Let $\mathcal{S}_\gamma^R=\{\gamma|\ |\gamma|<R\}$. We can write:
$$
\max_{\substack{{\bf u}\in \mathcal{S}_u\\ \gamma\in \mathcal{S}_\gamma}} \hat{Y}({\bf d},{\bf x},{\bf u},\gamma)=\sup_{R\geq 0} \max_{\substack{{\bf u}\in \mathcal{S}_u\\ \gamma\in \mathcal{S}_\gamma^R}}\hat{Y}({\bf d},{\bf x},{\bf u},\gamma)=\lim_{R\to\infty} \max_{\substack{{\bf u}\in \mathcal{S}_u\\ \gamma\in \mathcal{S}_\gamma^R}}\hat{Y}({\bf d},{\bf x},{\bf u},\gamma).
$$
Similarly, 
$$
\max_{\substack{{\bf u}\in \mathcal{S}_u\\ \gamma\in \mathcal{S}_\gamma}} \hat{Z}({\bf d},{\bf x},{\bf u},\gamma)=\sup_{R\geq 0} \max_{\substack{{\bf u}\in \mathcal{S}_u\\ \gamma\in \mathcal{S}_\gamma^R}}\hat{Z}({\bf d},{\bf x},{\bf u},\gamma)=\lim_{R\to\infty} \max_{\substack{{\bf u}\in \mathcal{S}_u\\ \gamma\in \mathcal{S}_\gamma^R}}\hat{Z}({\bf d},{\bf x},{\bf u},\gamma).
$$
For ${\bf d}$ and ${\bf x}$ fixed, functions ${R}\mapsto \max_{\substack{{\bf u}\in \mathcal{S}_u\\ \gamma\in \mathcal{S}_\gamma^R}}\hat{Y}({\bf d},{\bf x},{\bf u},\gamma)$ and ${R}\mapsto \max_{\substack{{\bf u}\in \mathcal{S}_u\\ \gamma\in \mathcal{S}_\gamma^R}}\hat{Z}({\bf d},{\bf x},{\bf u},\gamma)$  are non-decreasing. Using the fact that the minimum of the non-decreasing limit of a continuous function defined over compact sets is equal to the limit of minimum of this function, we obtain:
$$
\min_{\substack{{\bf d}\in \mathcal{S}_d\\ {\bf x}\in \tilde{\mathcal{S}}_{x,\theta}^r}}\max_{\substack{{\bf u}\in \mathcal{S}_u\\ \gamma\in \mathcal{S}_\gamma}}\hat{Y}({\bf d},{\bf x},{\bf u},\gamma)=\lim_{R\to\infty} \min_{\substack{{\bf d}\in \mathcal{S}_d\\ {\bf x}\in \tilde{\mathcal{S}}_{x,\theta}^r}}\max_{\substack{{\bf u}\in \mathcal{S}_u\\ \gamma\in \mathcal{S}_\gamma^R}}\hat{Y}({\bf d},{\bf x},{\bf u},\gamma)
$$
and
$$
\min_{\substack{{\bf d}\in \mathcal{S}_d\\ {\bf x}\in \tilde{\mathcal{S}}_{x,\theta}^r}}\max_{\substack{{\bf u}\in \mathcal{S}_u\\ \gamma\in \mathcal{S}_\gamma}}\hat{Z}({\bf d},{\bf x},{\bf u},\gamma)=\lim_{R\to\infty} \min_{\substack{{\bf d}\in \mathcal{S}_d\\ {\bf x}\in \tilde{\mathcal{S}}_{x,\theta}^r}}\max_{\substack{{\bf u}\in \mathcal{S}_u\\ \gamma\in \mathcal{S}_\gamma^R}}\hat{Z}({\bf d},{\bf x},{\bf u},\gamma).
$$
Then \eqref{res_case1} also holds for $\mathcal{S}_\gamma=\mathbb{R}$.

\noindent{\underline{Concluding.}} As the final step to prove Theorem \ref{1bit}, note that if $z$ and $\tilde{z}$ are negative then,
$$
Y({\bf d},{\bf x},{\bf u},\gamma)\geq \hat{Y}({\bf d},{\bf x},{\bf u},\gamma).
$$
Hence, 
$$
\mathbb{P}\Big[\min_{\substack{{\bf d}\in \mathcal{S}_d\\ {\bf x}\in \tilde{\mathcal{S}}_{x,\theta}^r}}\max_{\substack{{\bf u}\in \mathcal{S}_u\\ \gamma\in \mathcal{S}_\gamma}}{Y}({\bf d},{\bf x},{\bf u},\gamma)\leq t\Big]\leq \mathbb{P}\Big[\min_{\substack{{\bf d}\in \mathcal{S}_d\\ {\bf x}\in \tilde{\mathcal{S}}_{x,\theta}^r}}\max_{\substack{{\bf u}\in \mathcal{S}_u\\ \gamma\in \mathcal{S}_\gamma}}\hat{Y}({\bf d},{\bf x},{\bf u},\gamma)\leq t \  | z\leq 0, \tilde{z}\leq 0\Big].
$$
Using the fact that: 
$$
\mathbb{P}\Big[\min_{\substack{{\bf d}\in \mathcal{S}_d\\ {\bf x}\in \tilde{\mathcal{S}}_{x,\theta}^r}}\max_{\substack{{\bf u}\in \mathcal{S}_u\\ \gamma\in \mathcal{S}_\gamma}}\hat{Y}({\bf d},{\bf x},{\bf u},\gamma)\leq t  \Big]= \frac{1}{4}\mathbb{P}\Big[\min_{\substack{{\bf d}\in \mathcal{S}_d\\ {\bf x}\in \tilde{\mathcal{S}}_{x,\theta}^r}}\max_{\substack{{\bf u}\in \mathcal{S}_u\\ \gamma\in \mathcal{S}_\gamma}}\hat{Y}({\bf d},{\bf x},{\bf u},\gamma)\leq t\  | \ z\leq 0, \tilde{z}\leq 0  \Big],
$$
we finally obtain:
\begin{align*}
\mathbb{P}\Big[\min_{\substack{{\bf d}\in \mathcal{S}_d\\ {\bf x}\in \tilde{\mathcal{S}}_{x,\theta}^r}}\max_{\substack{{\bf u}\in \mathcal{S}_u\\ \gamma\in \mathcal{S}_\gamma}}{Y}({\bf d},{\bf x},{\bf u},\gamma)\leq t\Big]&\leq 4 \mathbb{P}\Big[\min_{\substack{{\bf d}\in \mathcal{S}_d\\ {\bf x}\in \tilde{\mathcal{S}}_{x,\theta}^r}}\max_{\substack{{\bf u}\in \mathcal{S}_u\\ \gamma\in \mathcal{S}_\gamma}}\hat{Y}({\bf d},{\bf x},{\bf u},\gamma)\leq t  \Big]\\&\leq 4 \mathbb{P}\Big[\min_{\substack{{\bf d}\in \mathcal{S}_d\\ {\bf x}\in \tilde{\mathcal{S}}_{x,\theta}^r}}\max_{\substack{{\bf u}\in \mathcal{S}_u\\ \gamma\in \mathcal{S}_\gamma}}\hat{Z}({\bf d},{\bf x},{\bf u},\gamma)\leq t  \Big]\\&= 4 \mathbb{P}\Big[\min_{\substack{{\bf d}\in \mathcal{S}_d\\ {\bf x}\in \tilde{\mathcal{S}}_{x,\theta}^r}}\max_{\substack{{\bf u}\in \mathcal{S}_u\\ \gamma\in \mathcal{S}_\gamma}} {Z}({\bf d},{\bf x},{\bf u},\gamma)\leq t  \Big].
\end{align*} 
\section{Technical lemmas}
\label{app:technical_lemmas}
\begin{theorem}[Theorem 1.1 in \cite{Gordon}]
	Let $X_{i,j}$ and $Y_{i,j}$, $i=1,\cdots, I$, $j=1,\cdots,J$ be centered Gaussian processes such that:
	$$
	\left\{ \begin{array}{ll}&\mathbb{E}X_{ij}^2=\mathbb{E}Y_{ij}^2, \ \forall i,j\\
		&\mathbb{E}X_{ij}X_{ik}\geq \mathbb{E}Y_{ij}Y_{ik}, \forall i,j,k\\
		&\mathbb{E}X_{ij}X_{lk}\leq \mathbb{E}Y_{ij}Y_{lk} \  \forall i\neq l \ \text{and} \  j,k
	\end{array}\right..
	$$ 
	Then, for all $\lambda_{ij}\in\mathbb{R}$,
	$$
\displaystyle\mathbb{P}\Big[\displaystyle{\cap_{i=1}^{I}}\cup_{j=1}^J \Big\{Y_{i,j}\geq \lambda_{ij}\Big\}\Big]\geq 	\displaystyle\mathbb{P}\Big[\cap_{i=1}^{I}\cup_{j=1}^J \Big\{X_{i,j}\geq \lambda_{ij}\Big\}\Big].
	$$
	\label{lem:gor}
\end{theorem}

\begin{lemma}[\cite{was}]
	Let $d\geq 1$ and $\mathcal{P}(\mathbb{R}^{d})$ be the set of all probability measures on $\mathbb{R}^{d}$. For $\mu\in \mathcal{P}(\mathbb{R}^{d})$, we consider an i.i.d sequence $({\bf x}_k)_{k\geq 1}$ of $\mu$-distributed random variables and denote for $N\geq 1$, the empirical measure:
	$$
	\mu_N:=\frac{1}{N}\sum_{k=1}^N \boldsymbol{\delta}_{{\bf x}_k}.
	$$
	For $\alpha,\gamma>0$, denote by $\mathcal{E}_{\alpha,\gamma}$ the quantity 
	$$
\mathcal{E}_{\alpha,\gamma}:=\int_{\mathbb{R}^{d}}\exp(\gamma \|{\bf x}\|^\alpha)\mu(dx).
	$$
	Let $r\geq \frac{d}{2}$. Assume that there exists $\alpha >r$ and $\gamma>0$ such that  $\mathcal{E}_{\alpha,\gamma}<\infty$.   If $r>\frac{d}{2}$, then for any $0<\epsilon<1$, 
	$$
	\mathbb{P}\Big[\big(\mathcal{W}_r(\mu,\mu_N)\big)^r\geq \epsilon\Big]\leq C\exp(-cN\epsilon^2)
	$$
	where $C$ and $c$ are constants that depend only on $r$ and $d$. 
	\label{lem:convergence_empirical_rate}
\end{lemma}

\begin{definition}The Rademacher average of a bounded set $\mathcal{A}\subset \mathbb{R}^n$ is given by:
$$
R_n(\mathcal{A})=\mathbb{E}\Big[\sup_{{\bf a}\in \mathcal{A}} \Big|\frac{1}{n}\sum_{i=1}^n \sigma_i [{\bf a}_i]\Big|\Big]
$$
where $\sigma_1,\cdots,\sigma_n$.
\end{definition}
\begin{lemma}[Finite Class Lemma, \cite{peter}]
Let $\mathcal{A}=\{{\bf a}_1,\cdots,{\bf a}_N\}$ a finite set with $\|{\bf a}_i\|\leq L$ for all $i=1,\cdots,N$. Then, 
$$
R_n(\mathcal{A})\leq \frac{2L \log N}{n}.
$$
\label{lem:finite_class_lemma}
\end{lemma}
\begin{lemma}[Corollary 5.35 in \cite{vershynin2012introduction}]
Let ${\bf X}$ be a $m\times n$ matrix with i.i.d standard Gaussian entries. Then 
$$
\mathbb{P}\Big[\|{\bf X}{\bf X}^{T}\|\leq 9\max(m,n)\Big]\geq 1-2\exp(-{\rm max}(m,n)/2).
$$
\label{lem:conc_spectral_norm}
\end{lemma}
\begin{lemma}
	Let ${\bf c}$ and ${\bf d}$ two distinct vectors in $\mathbb{R}^{m}$. Let $b$ be a positive constant. Consider the following convex problem:
	\begin{align*}
		m=  	&\min_{{\bf e}\in\mathbb{R}^m}  \|{\bf c}+{\bf e}\|^2, \\
		{\rm s.t.} \  &\|{\bf e}+{\bf d}\|^2\leq b.
	\end{align*}
	Then, the above problem admits a unique minimizer given by:
	$$
	{\bf e}^{\star}=\frac{-{\bf c}-\lambda^\star {\bf d} }{1+\lambda^\star}
	$$
	where
    $$
    \lambda^\star=\left\{\begin{array}{ll}
    0 & \text{if } \|{\bf d}-{\bf c}\|\leq \sqrt{b}\\
    -1+\frac{\|{\bf d}-{\bf c}\|}{\sqrt{b}} & \text{otherwise}
    \end{array}
    \right..
    $$
    Moreover, at optimum, the optimal cost is given by:
	$$
	m=(-\sqrt{b}+\|{\bf c}-{\bf d}\|)^2	.
	$$
	Additionally, for all feasible ${\bf e}$, 
	$$
	\|{\bf c}+{\bf e}\|^2\geq m+\|{\bf e}-{\bf e}^\star\|^2.
	$$
	\label{lem:KKT}
\end{lemma}

 \bibliographystyle{IEEEtran}
	\bibliography{ref}
	
	\vfill
 \end{document}